\documentclass[11pt, authoryear]{elsarticle}

\usepackage[latin1]{inputenc}
\usepackage{csquotes}
\usepackage[english]{babel}

\usepackage{setspace}
\usepackage{amsthm}
\usepackage{amsmath}
\usepackage{amssymb}

\usepackage{accents} \newcommand{\ubar}[1]{\underaccent{\bar}{#1}}
\usepackage{array}

\usepackage{color}
\usepackage{comment}
\usepackage{enumerate}

\usepackage{geometry} 
\usepackage{graphicx}
\usepackage[pdfborder={0 0 1}]{hyperref}
\usepackage{makeidx}
\usepackage{mathptmx}
\usepackage{multirow}
\usepackage{tikz,pgfplots} \usetikzlibrary{arrows,automata,positioning}
\usepackage{textcomp}
\usepackage{units}

\newtheorem{theorem}{Theorem}
\newtheorem{example}{Example}

\newtheorem{proposition}{Proposition}

\newtheorem{corollary}{Corollary}
\newtheorem{definition}{Definition}

\newtheorem{fact}{Fact}
\newtheorem{lemma}{Lemma}
	
\newtheorem{condition}{Condition}

\begin{document}	
\title{Sorting in Networks: Adversity and Structure \tnoteref{t1}} 
\tnotetext[t1]{The author is grateful to Peter Norman Sørensen who has supervised his PhD. A large thanks to Jan Eeckhout, Matthew Jackson, Jesper Rüdiger, Hans Keiding, Bartosz Redlicki, John Kennes, Thomas Jensen for fruitful discussions. I would also like to thank seminar participants at CoED 2015, CoopMAS 2015, EDGE 2015, University of Copenhagen, Stanford and Universitat Pompeu Fabra.
This paper is a revised version of the first chapter in the author's PhD dissertation at University of Copenhagen from February 2016. The work was by University of Copenhagen as part of the author's PhD.}
\date{}
\author{Andreas Bjerre-Nielsen}
\ead{andreas.bjerre-nielsen@econ.ku.dk}
\address{Øster Farimagsgade 5, Copenhagen K. 1353.}
\begin{abstract} People choose friendships with people similar to themselves, i.e. they sort by resemblence. Economic studies have shown when sorting is optimal and constitute an equilibrium, however, this presumes lack of beneficial spillovers. We investigate formation of economic and social networks where agents may form or cut ties. We combine a setup with link formation where agents have types that determine the value of a connection. We provide conditions for sorting in friendships, i.e. that agents tend to partner only with those with those sufficiently similar to themselves. Conditions are provided with and without beneficial spillovers from indirect connections. We show that sorting may be suboptimal, yet a socially stable outcome, despite otherwise obeying the conditions for sorting in \cite{becker_theory_1973}. We analyze policy tools to mitigate suboptimal sorting. Another feature is that agents with higher value are more central in networks under certain conditions; a side effect is sorting by degree centrality under certain conditions. Finally we illustrate the limits to patterns of sorting and centrality. 
\end{abstract}
\maketitle


\onehalfspacing

\section{Introduction}\label{sec:introduction}

A ubiquitous finding in studies of social relations is the tendency to form more ties with people similar to one-self, i.e. the pattern known as \textit{sorting} or homophily cf. the meta-study by \cite{mcpherson_birds_2001}.	
Economic research has contributed to the understanding of sorting by providing mathematically sufficient conditions for when marriage-labor markets and groups get sorted and when this is optimal. However, no research has reconciled the framework of \cite{becker_theory_1973} on sorting with the literature on network formation as pioneered by \cite{jackson_strategic_1996}. 
Our investigation yields new insights on network formation by rational agents: when is sorting stable and optimal if beneficial network externalities are either not present or present, and; which other patterns arise in these stable networks? 
In addition we demonstrate how policies may be enacted that can curb sorting when excessive, and thus provide higher welfare. These new insights may help policymakers to optimally design schools or corporations.

While a modest tendency to sort may be beneficial by forming smaller communities that encourage trust and shared values, however, too much sorting can be detrimental. The problems of sorting include slowing down the spread of information and higher inequality.%
\footnote{\cite{golub_how_2012} show that homophily affects information diffusion. \cite{eika_educational_2014} find empirical evidence that increased marriage sorting have contributed to higher inequality.}
As recent decades have seen a rise in residential income segregation and assortative mating by education it is a great concern for policy makers how to tackle sorting.\footnote{See \cite{reardon_income_2011} on segregation and \cite{schwartz_trends_2005} on marriage sorting.}
The insights from our analysis can be applied to understand the the indirect effects of sorting. For instance when dividing school cohorts by ability as is the current practice in many countries, cf. \cite{oecd_equity_2012}.

Our framework explores a setting with agents choosing partners under three core assumptions. First, agents are heterogeneous in type for creating value in partnerships (this is in line with research on peer effects).%
\footnote{Some examples of empirical papers finding this include \cite{zimmer_peer_2000}, \cite{sacerdote_peer_2001}, \cite{falk_clean_2006}. Note that empirical estimates should be cautiously interpreted as the research on peer effects is still in an infant stage cf. \cite{manski_identification_1993}, \cite{shalizi_homophily_2011}, \cite{angrist_perils_2014}.} 
Type may refer to productive and non-productive capabilities such as skill, social aptitude or interests and ethnicity. Second, there are possible externalities from friends of friends as in the 'connections model' of \cite{jackson_strategic_1996} which captures spillovers of ideas and favors (which does not contain types). Third, agents choose a limited number of partners reflecting constraints of time and effort.\footnote{Limited partners is also consistent with empirical research; \cite{ugander_anatomy_2011} shows this for the entire Facebook network and \cite{miritello_time_2013} in phone calls for millions of people.} Note our framework could also model corporations or self-governing organizations forming bilateral partnerships among themselves. 
	
In this setting we investigate robustness of network structures in the following sense: no agents can form and/or delete links from network and be better off when allowing for transfers of utility. 
Although the literature on formation of networks under externalities is vast the share that investigates heterogeneous agents has been limited. \cite{jackson_strategic_1996} pioneered in their investigation of tension between stability and efficiency under mutual acceptance of links. A related investigation of non-mutually accepted links was pioneered by \cite{bala_noncooperative_2000}. 
\textcolor{black}{
The most relevant paper is  \cite{de_marti_segregation_2016} which also investigates segregation in a model with complementarity and externalities.%
\footnote{Note this paper was developed independently and without awareness of \cite{de_marti_segregation_2016}. Further discussion below.}}
Another relevant paper is \cite{goeree_search_2009} which explores network formation and structure when there is value heterogeneity, i.e. some agent types make links formed more valuable. 

\paragraph{Contributions}
	
Below we list our contributions to the literature and discuss some applications. The results contribute mainly to our understanding of network formation and assortative matching. Some of our results are derived when only evaluating links that are formed between a pair of agents - this corresponds to pairwise (Nash) equilibrium. Other results require that any coalition of agents can form links between themselves, i.e. the strong (Nash) equilibrium.%
	\footnote{To the best of our knowledge this is the first demonstration of the novel features.}
	
The first contribution is towards the understanding of conditions for sorting. We establish sufficient conditions for stable networks to exhibit sorting in type.
In the simple case without externalities our Proposition  \ref{claim:assort:elementwise:monotonic} extend the classic work to a one-sided matching setting with many partners. 
In the case of externalities we provide a characterization for the set of pairwise (Nash) stable networks: when  \cite{becker_theory_1973}'s complementarity hold then that set coincides with the set of perfectly sorted networks where every agent within a given type is connected and uses its quota of links; see Theorem \ref{claim:suboptimal-sorting_hyperbolic_decay} for the case of non-constant decay (or Theorem \ref{claim:assort:ext:pairwise:excess} in Appendix \ref{app:oversort_generalize} for constant decay). We also show that although these networks are stable they are not efficient when there is a small or moderate level of externalities due to a lack of connectivity across types. This problem of sub-optimality is shown to be mitigated through policy, see Proposition \ref{claim:policy_accomodate_oversorting}. 
Finally our Proposition \ref{claim:assort:ext:asymptoticperfectPAM} establishes that if the agent population is large then sorting is the unique strongly stable outcome (i.e. the core) when Becker's complementarity hold and that externalities satisfy asymptotic independence in social connections.

The above generalization of sorting can be applied to investigate models of homophily in social networks. Moreover, we provide a clear intuition how pathological sorting can constitute an undesirable, yet, stable configuration of a social network. The intuition is analogue to a classic public goods problem: there is a lack in the provision of links across types. This result also can be seen as providing a generic network structure, which is commonly observed, and that suffers from inefficiencies in network formation. Finally, to mitigate the sorting issue we demonstrate the validity of a policy instrument that provides incentives to link across types. 
	
The second contribution is to bring the literature on network structure into the research on game theory and network formation. Our focus on structure is limited to measures of network centrality for individuals. We show that if a monotonicity condition on contribution to the value of link in type is satisfied then high type agents are more central in all stable networks; see Propositions \ref{claim:network:monotonicdegree:talent} and \ref{claim:monotonic_decaycentrality}. Without externalities sorting in degree emerges, see Proposition \ref{claim:network:elementwise:monotonicdegree}, if the conditions for sorting by type and degree centrality are met. An application of this result is that sorting by type may be detected simply from the network structure if there is sorting by degree.
	
Our final contribution is to display the limits to conditions for sorting and centrality. This is done by showing that a network where a low type agent is excessively central can be stable, see Example \ref{example:monotonic_fail} and Corollary \ref{claim:fail_corollary}. When the condition for monotonic centrality is not satisfied this entails a less productive agent could be replaced in a central position and thus everyone could be better off. This highlights another fundamental network configuration which is undesirable under pairwise link formation. This can be applied to understand how some agents can maintain their central brokering position although this is detrimental.
	
\paragraph{Related literature}
The following paragraphs compare our results with the most relevant studies among existing work by assessing the contributions above. As there is a large body of literature the review is limited to only the most similar and noteworthy. 

\textcolor{black}{The most relevant paper on sorting is \cite{de_marti_segregation_2016} which investigates conditions for segregation/sorting in networks with externalities but with noteworthy differences (see below).
As our setup builds on an assortative matching framework, unlike \cite{de_marti_segregation_2016}, we can directly compare our results with those already established.} 
The work related to stability of sorting has a strong tradition for two-sided matching e.g. labor and dating markets starting with \cite{becker_theory_1973}. The research on one-sided assortative matching, i.e. the basis for our setup, has been limited to formation of clubs (not networks) under various technologies cf. \cite{farrell_partnerships_1988}, \cite{kremer_o-ring_1993}, \cite{durlauf_is_2003}, \cite{legros_assortative_2002}, \cite{pycia_stability_2012}, \cite{baccara_homophily_2013}. All the work on one-sided assortative matching find conditions for sorting which is similar to type complementarity in \cite{becker_theory_1973}.
Yet, none of these papers allow for linking between groups nor consider network externalities which are the extensions we treat. Note that Proposition \ref{claim:assort:elementwise:monotonic} can be seen as a direct extension of the classic result by \cite{becker_theory_1973} for networks without externalities - our contribution is to propose a new measure of sorting which is tractable in equilibrium. 

We contribute to the knowledge on sorting under externalities by building upon \cite{de_marti_segregation_2016}. In terms of equilibrium outcomes we remove their requirement that either one or all agents of a given type are linked to all other agents of the same type (i.e. respectively a star or a clique). This requirement imposes an unrealistic structure implying that an agent may possess infinitely many links in large networks.
Furthermore, we provide a complete characterization of pairwise (Nash) equilibria for moderate externalities. Furthermore, we demonstrate that for large populations sorting constitutes the unique strongly stable outcome under asymptotic independence of externalities.
Thus we advance their results on sorting by making them apply more widely by allowing for variation in structures of the subnetworks and for multiple types of agents (see details accompanying each result) while keeping the number of links bounded for each agent.
The divergence in results from \cite{de_marti_segregation_2016} are driven by their choice of modeling complementarity as endogenous while ours is exogenous.\footnote{The endogenous complementarity of \cite{de_marti_segregation_2016} implies that sorting may be self-reinforcing as the incentive to sort increases in other agents tendency to sort.} 
Other essential differences are: we employ convex costs where they use linear, and; 
we use pairwise Nash stability which is stricter than their pairwise stability as it allows for substitution of partners.


Research on suboptimal networks under externalities has a long tradition (cf. \cite{koopmans_assignment_1957}; \cite{katz_network_1985};  \cite{farrell_installed_1986}; \cite{jackson_strategic_1996}; \cite{jackson_strategic_1996};\footnote{Note our concept of  pairwise stability differs from \cite{jackson_strategic_1996} - see Section \ref{sec:model} (the model).} \cite{bloch_formation_2007}). 
Among these generic analyses of inefficient networks \cite{bloch_formation_2007} employs the most similar setup. \cite{bloch_formation_2007} use a setup without type complementarities but also employ a framework if network formation where transfers can be extended even to agents with whom they are not linked. Unlike the above papers our main contribution for inefficient networks is not to provide any new generic insights on suboptimal network formation; instead we show two new classes of inefficient networks. The first class of inefficient networks suffers from an excessive level of sorting captured by Theorems \ref{claim:suboptimal-sorting_hyperbolic_decay} and \ref{claim:assort:ext:pairwise:excess} (see last paragraph for second class).
\textcolor{black}{
\cite{de_marti_segregation_2016} also investigates sorting which is suboptimal despite being stable - thus this part of their analysis is subject to the same critique as their results on sorting (see above). 
Moreover, their result on inefficiency only holds when there are a handful of agents of each type (see discussion accompanying Theorem \ref{claim:suboptimal-sorting_hyperbolic_decay} for details). On the contrary, our result on inefficient sorting holds for any number of agents and only requires a connected subnetwork within a given type - this reveals that the scope of the underlying problem is much wider.
Another, perhaps more profound difference, is that the source of inefficiency in their setting is different: their endogenous incentive to sorting implies that the penalty from linking across types is removed when no agents sort (i.e. sort negatively) - thus there are always beneficial spillovers from connecting across types. Thus inefficiency is due to a lack of coordination rather than misaligned incentives.
}

	
Our results on centrality relate to work within economics and network science. The only relevant economic research related to the relationship between type and centrality is \cite{goeree_search_2009}. They show in two examples a star network with a high value agent in center is efficient and Nash stable. Note their analysis is under one-way formation of links (only one side accepts) and their examples are not generalized beyond those for few agents. In their experiments the efficient network emerges more often and the frequency increases with repetitions. 

The work that relates most to sorting in degree are \cite{currarini_economic_2009,currarini_simple_2016,iijima_social_nodate} which also model strategic agents, however, without externalities. One big difference is that degree sorting their models are different and all driven by composition of types rather than preferences/complementary which is the source in this setup. There are several non-economic models in statistical physics which predict sorting in degree, yet, these have no strategic considerations. The most prominent is preferential attachment, see \cite{barabasi_emergence_1999}. Note that an analogy to our pattern of degree sorting is found in \cite{kremer_o-ring_1993} and \cite{farrell_partnerships_1988} where more productive groups can sustain a larger scale of organizations. 

Our results on the limits of sorting indicate that an agent can amass so much network capital that the agent becomes the most central despite not being the best candidate. In these networks there is a failure of sorting and monotonic centrality. The work on network externalities in two-sided settings by \cite{katz_network_1985} and \cite{farrell_installed_1986} find an analogous pattern. In their setups the source to inefficiency is that a suboptimal producer is supplying the consumers who cannot coordinate on choosing the better producer and get "locked-in".

\paragraph{Paper organization}
The rest of the paper is structured as follows: Section \ref{sec:model} introduces the model; Section \ref{sec:assort-talent} investigates sorting and its potential sub-optimality; Section \ref{sec:network:structure} analyzes the network structure implied by the model; Section \ref{sec:limits} rounds up the analysis by outlining the limits to sorting and network structure, and; Section \ref{sec:discussion} concludes in a discussion of assumptions. All proofs are found along in Appendix \ref{app:proofs}.

\section{Model}\label{sec:model}

Let $N =\{1,..,n \}$ constitute a set of agents. 
Each agent $i \in N $ is endowed a fixed measure of \emph{type}, $x_{i }\in  X $ where $X \subset \mathbb{R}$ is the set of realized levels of types for agents in $N $. Let $\bar{x} =\max X $ and $\underaccent{\bar}{x} =\min X $. 
Let agents' type be sorted descending in their label and let $\mathcal{X} =(x_1,x_2,...,x_n )$ where $x_l \ge x_{l +1}$ for $l =1,..,n -1$. 

\paragraph{Linking and networks}

Two agents $i ,j \in N $ may \textit{link} if mutually accepted; a link may be broken by both agents. A link between $i $ and $j $ is denoted ${i }{j }\in \mu $ where the set $\mu $ consists of links and is called a \textit{network}. The set of all networks is denoted $M  = \{\mu |\,\mu \subseteq \mu^c \}$ where $\mu^c $ is the \textit{complete} network (all agents are linked). 

A \textit{coalition} of agents is a group $t \subseteq N $ such that $t \in T $ where $T $ is the superset of $N $ excluding the empty set.  
For a given group $t $ define $\mathcal{X} (t )$ as the vector of types ordered descending for agents in $t $.
A \textit{coalitional move} is a set of actions implemented by a coalition that moves the network from one state to another. A move from $\mu $ to $\tilde{\mu} $ is \emph{feasible} for coalition $t $ if: added links,  $\tilde{\mu} \backslash\mu $ are formed only between members of coalition $t $; deleted links, $\mu \backslash\tilde{\mu} $ must include a member of coalition $t $.

\paragraph{Network measures}

The \textit{neighborhood}, $\nu $, is the set agents whom an agent links to: $\nu_{i }(\mu )=\{j \in N :\, {i }{j }\in \mu \}$. The number of neighbors is a measure of centrality called \textit{degree} and denoted $k_{i }(\cdot )$ for $i $. 
A \textit{path} is a subset of links $\{i_{1}i_{2},i_{2}i_{3},...,i_{l -1}i_{l}\}\subseteq \mu $ where no agent is reached more than once. The \textit{distance} between two agents $i ,j $ in a network is the length of the shortest path between them - this is denoted $p_{i j }:M \rightarrow\mathbb{N}_0$. When no path exist then distance is infinite.

\paragraph{Utility}

The utility accruing to agent $i $ is denoted $u_i$. An agent's utility equals benefits less costs, expressed mathematically as $u_i=b_i-c_i$. The aggregate utility is denoted $U (\cdot )$. 	

We use two approaches for modeling costs: a quota of links which is employed in Section \ref{sec:assort-talent} on sorting or the more general used in Section \ref{sec:network:structure} on centrality. The \textit{degree quota}, $\kappa $, is the maximum number of links for any agent: for $i \in N $ it holds $k_{i }(\cdot )\le \kappa $. The quota implies an opportunity cost of linking. If a degree quota is employed then \textit{no linking surplus} means all agent has a degree equal to the degree quota, i.e. $\forall i \in N :k_i =\kappa $. 
A type $x \in X $ is \textit{self-sufficient} if $n_x >\kappa $.

Benefits to agent $i $ is a weighted sum consisting of two elements; network and individual value:	
\begin{equation}\label{eq:def:benefits}
b_{i }(\mu )
=\Sigma_{j \ne i }w_{i j }(\mu )\cdot z_{i j }
\end{equation}

The network factor, $w_{i j }(\mu )$ is a function of network distance. The individual link value is $z_{i j }$ which measures the personal value to $i $ of linking to $j $ - the value is a function of the two partners' type $z_{i j }=z (x_{i },x_{j })$. The function $z $ is assumed twice differentiable as well as taking positive and bounded values.\footnote{The upper bound rules out an infinite number of links in equilibrium.} 	Let the \textit{(total) link value} be defined as the value of linking for the pair, i.e. $Z_{i j }=z_{i j }+z_{j i }$.%
\footnote{Note most results use only the total link value. However, the results in Section \ref{sec:limits} rely on decomposability of total link value into individual link value.}	
The two components, network and linking value, are further restricted in the subsequent sections depending on their usage.

\paragraph{The game framework}
This paper explores a static setting of one period.  Agents' information about payoffs of other agents is complete. Together the players, action, utility and information constitutes a game we will now outline the stability concept for.

Any pair of agents can transfer 'utility' between them.  Let a \textit{net-transfer} from agent $j $ to agent $i $ be denoted as $\tau_{i j }\in \mathbb{R}$ such that  $\tau_{i j }=-\tau_{j i }$ which implies non-wastefulness of utility. The matrix of net transfers is denoted $\tau $.  For each agent $i $ its net-payoff is defined as $s_{i }=u_{i }(\mu )+\Sigma\tau_{i j }$. In addition it is assumed that the net-transfer between two agents with a link can only be changed if this is mutually agreed or the link is broken.

\paragraph{Stability}

We define network stability using the concept of coalitional moves. A coalition $t $ is \emph{blocking} a network $\mu $ with net-transfers $\tau $ if there is a feasible coalitional move from network $\mu $ to network $\tilde{\mu} $ with $\tilde{\tau} $ where all members in $t $ have higher net-payoff after the move. 

We employ two concepts of stability. The first is  \emph{strong stability}: this is satisfied for a network if there exist transfers such that no coalition (of any size) may have a feasible move which is profitable for all its members. The second concept, \textit{pairwise (Nash) stability},\footnote{This also known as bilateral stability, cf. \cite{goyal_structural_2007}} is similar but has weaker requirements: it holds when there exist transfers such that no coalitions of at most two agents may block. The sets of stable networks are denoted respectively $M ^{s-stb }$ and $M ^{p-stb }$. A discussion of the stability concepts and its relation to the literature is found in Section \ref{sec:discussion}.

Our pairwise definition of stability is stricter than that of \cite{jackson_strategic_1996}. However, the stricter requirement is to allow for substitution of links (simultaneous deletion and formation) which is a standard requirement in the matching literature.\footnote{For instance, pairwise stability coupled with and \cite{becker_theory_1973}'s condition for sorting, i.e. supermodularity, would not imply that sorting be the unique pairwise stable outcome in the marriage market as in \cite{becker_theory_1973}.}

A noteworthy feature is that strongly stability implies pairwise stability; thus any condition valid for pairwise stability also applies to strong stability. In addition without network externalities every pairwise stable network is also strongly stable, see Lemma \ref{claim:auxiliary:equivalence:strong:pairwisetype}.  Note also that any strongly stable network requires efficiency (coalition of all agents can implement any network). Thus we can employ efficiency to derive structure in strongly stable networks.

\section{Sorting in type}\label{sec:assort-talent}
		
This section investigates stratification in type, i.e. based on the agents' exogenous types. In this section we employ the degree quota. Recall that no linking surplus means that any agent $i $ must use all possible links, i.e. $k_i =\kappa $. The no linking surplus will be a requirement in most equilibria in this section.

In order to derive results a restriction of payoffs is necessary. Recall the total link value is the sum of individual link for the two partners in the pair. The essential feature of the total link value for sorting is complementarity in type. \cite{becker_theory_1973} shows that type complementarity holds when:

\begin{definition}\label{def:complementarity}The link value has \textbf{supermodularity} if
	$\frac{\partial ^{2}}{\partial  x  \partial  y }
	Z (x ,y )>0$ - this entails:
	\begin{equation}
	Z (x ,\tilde{x} )+
	Z (y ,\tilde{y} )>
	Z (x ,y )+
	Z (\tilde{x} ,\tilde{y} ),
	\qquad 
	x >\tilde{y} 
	,\,
	\tilde{x} >y .
	\label{eq:supermodularity}
	\end{equation}
\end{definition}

The concept of sorting that we employ is a generalization of the sorting when there is a single partner such as Becker's marriage market. The shape of sorting is such that a high type agent has partners which weakly dominate in type when compared partner-by-partner with the partners of a lower type agent. Note the comparison is done over the sorted set of partners type $\mathcal{X} $. The sorting pattern is mathematically defined as:

\begin{definition} \textbf{Sorting in type} holds in $\mu $ if for all pair $i ,j $ such that $x_{i }>x_{j }$ it holds that: 
\[\mathcal{X} (\nu_{i }(\mu )/\{j \})_l \ge \mathcal{X} (\nu_{j }(\mu )/\{i \})_{l +l ^{*}},\qquad \forall l \in \{1,..,k ^{*} \},\] where $k ^{*}=\min (k_{i }(\mu ),k_{j }(\mu ))$ and $l ^{*}=\max (k_j (\mu )-k_i (\mu ),0)$.
\end{definition}

\begin{definition} A network $\mu $ is \textbf{perfect sorted (in type)} if $x_i =x_j $ for any $i j \in \mu $. Denote the set of perfectly sorted networks as $M ^{p-srt }$.\end{definition}

\subsection{No externalities}
We begin with the setting where network externalities are absent and thus indirect connections are irrelevant. Our first result is that sorting in type emerges under similar conditions as in \cite{becker_theory_1973} when network externalities are absent:
		
\begin{proposition}\label{claim:assort:elementwise:monotonic} If there is supermodularity and no externalities then for any pairwise stable network there is sorting in type.
\end{proposition}

In the following subsections we extent results to a setting with externalities. Note that the presence of externalities may imply sorting is not always equilibrium outcome - see Example \ref{example:monotonic_fail} and Corollary \ref{claim:fail_corollary}. 

\subsection{Externalities}		

We proceed to a more general context where indirect connections matter for utility. Whenever we allow for externalities we restrict our attention to forms of linking utility.
\begin{equation}
w_{i j }(\mu )=\begin{cases}
\delta ^{p_{i j }(\mu )-1},&\mbox{constant decay,}\\
\mathbf{1}_{=1}(p_{i j }(\mu ))+\delta \cdot \mathbf{1}_{\in (1,\in fty)}(p_{i j }(\mu )),&\mbox{hyperbolic decay,}
\end{cases}		
\label{eq:utility_shortest_paths}
\end{equation} 

where $\mathbf{1}_{\in (1,\in fty)}(l )$ is the Dirac measure/indicator function of whether $1<l <\in fty$. 

The first and more general setting is where utility from connections decays over increasing distance at a constant exponential rate - this corresponds to benefits from linking in the 'connections-model' from \cite{jackson_strategic_1996}. The other case is when externalities from indirect connections are discounted equally at any distance if there is a connection, i.e. a finite path length. This second case is referred to as hyperbolic decay and entails that there is no decay beyond that from distance one (linked) to distance two. 

In the remainder of this section we employ hyperbolic decay when analyzing inefficient sorting as it provides for more intuitive and immediate results without restrictions on the network. As noted earlier, a more general exposition is found in Supplementary Appendix~\ref{app:oversort_generalize}.

\paragraph{Finite population}

In the following paragraphs we reexamine our results on sorting by illustrating a class of sorted networks that are suboptimal when introducing network externalities. The sub-optimality holds despite fulfilling supermodularity. We begin here by exhibiting a simple example where the problem emerge.

Excessive sorting refers to a network with perfect sorting according to type where the segregated groups could collectively benefit from connecting. However, they fail to connect as incentives do not internalize externalities under pairwise network formation. The relevant set of networks are those networks which are perfectly sorted but within each type all agents are connected. 
Let networks which are perfectly sorted and connected among each type be denoted:
\[M ^{p-srt:conn}=\{\mu \in  M ^{p-srt }\,|\, \forall  i ,j \in N , x_i =x_j : p_{i j }(\mu )<\in fty\}\]	

The aim is to show there exist an open region of thresholds $(\ubar{\delta}, \bar{\delta})$ such that if the level of externalities is within the region then any network in $M ^{p-srt:conn }$ where there is no degree surplus is pairwise (Nash) stable for some transfers $\tau$ but inefficient.
An introduction to this problem is found in Example \ref{example:excessiveassortative} below in a stylized, simple manner. The example is graphically represented in Figure \ref{fig:excess_network}.

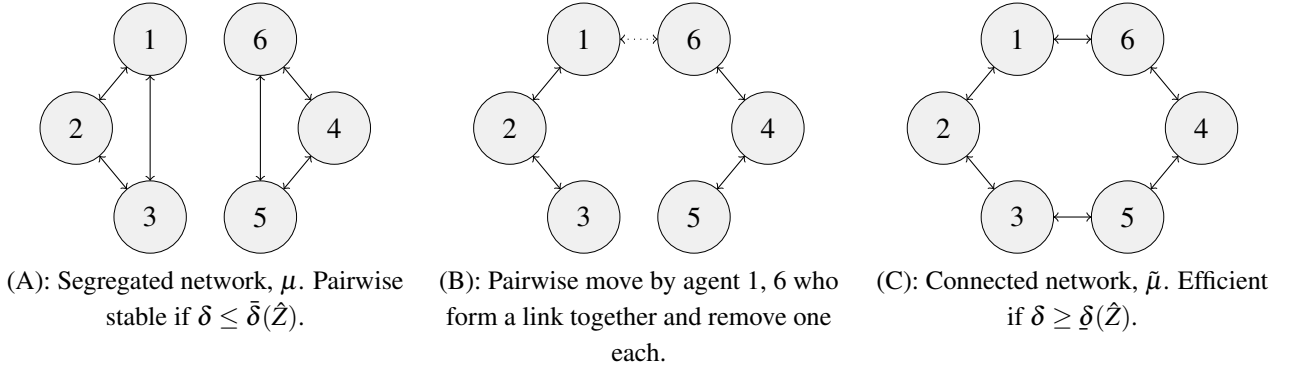
\begin{figure}[t!]
	\begin{tabular*}{15.8cm}{@{\extracolsep{\fill} }>{\centering}p{0.33\linewidth}>{\centering}p{0.33\linewidth}>{\centering}p{0.33\linewidth}}
		\begin{tikzpicture}[nodes={fill=gray!12,circle, ultra thin}, node distance=2cm, auto]
		\draw[help lines] (0,0);
		\node[state] (i_1) {$2$};
		\node[state] (i_2) [above right=0.5cm and 0.3cm of i_1] {$1$};
		\node[state] (i_3) [below right=0.5cm and 0.3cm of i_1] {$3$};
		\node[state] (i_4) [right=0.5cm of i_2] {$6$};
		\node[state] (i_5) [right=0.5cm of i_3] {$5$};
		\node[state] (i_6) [below right=0.5cm and 0.3cm of i_4] {$4$};
		\path[<->] (i_1) edge node [left, fill=none, font=\footnotesize]{} (i_2);
		\path[<->] (i_1) edge node [above, fill=none, font=\footnotesize]{} (i_3);
		\path[<->] (i_2) edge node [below, fill=none, font=\footnotesize]{} (i_3);
		\path[<->] (i_4) edge node [above, fill=none, font=\footnotesize]{} (i_5);
		\path[<->] (i_4) edge node [below, fill=none, font=\footnotesize]{} (i_6);
		\path[<->] (i_5) edge node [left, fill=none, font=\footnotesize]{} (i_6);
		\end{tikzpicture}
		& 
		\begin{tikzpicture}[nodes={fill=gray!12,circle, ultra thin}, node distance=2cm, auto]
		\draw[help lines] (0,0);
		\node[state] (i_1) {$2$};
		\node[state] (i_2) [above right=0.5cm and 0.3cm of i_1] {$1$};
		\node[state] (i_3) [below right=0.5cm and 0.3cm of i_1] {$3$};
		\node[state] (i_4) [right=0.5cm of i_2] {$6$};
		\node[state] (i_5) [right=0.5cm of i_3] {$5$};
		\node[state] (i_6) [below right=0.5cm and 0.3cm of i_4] {$4$};
		\path[<->] (i_1) edge node [left, fill=none, font=\footnotesize]{} (i_2);
		\path[<->] (i_1) edge node [above, fill=none, font=\footnotesize]{} (i_3);
		\path[<->, dotted] (i_2) edge node [below, fill=none, font=\footnotesize]{} (i_4);
		\path[<->] (i_4) edge node [below, fill=none, font=\footnotesize]{} (i_6);
		\path[<->] (i_5) edge node [left, fill=none, font=\footnotesize]{} (i_6);
		\end{tikzpicture} 
		& 
		\begin{tikzpicture}[nodes={fill=gray!12,circle, ultra thin}, node distance=2cm, auto]
		\draw[help lines] (0,0);
		\node[state] (i_1) {$2$};
		\node[state] (i_2) [above right=0.5cm and 0.3cm of i_1] {$1$};
		\node[state] (i_3) [below right=0.5cm and 0.3cm of i_1] {$3$};
		\node[state] (i_4) [right=0.5cm of i_2] {$6$};
		\node[state] (i_5) [right=0.5cm of i_3] {$5$};
		\node[state] (i_6) [below right=0.5cm and 0.3cm of i_4] {$4$};
		\path[<->] (i_1) edge node [left, fill=none, font=\footnotesize]{} (i_2);
		\path[<->] (i_1) edge node [above, fill=none, font=\footnotesize]{} (i_3);
		\path[<->] (i_2) edge node [below, fill=none, font=\footnotesize]{} (i_4);
		\path[<->] (i_3) edge node [above, fill=none, font=\footnotesize]{} (i_5);
		\path[<->] (i_4) edge node [below, fill=none, font=\footnotesize]{} (i_6);
		\path[<->] (i_5) edge node [left, fill=none, font=\footnotesize]{} (i_6);
		\end{tikzpicture}
		\tabularnewline
		{\small (A): Segregated network, $\mu $. Pairwise stable if $\delta \le \bar{\delta} (\hat{Z} )$.} & 
		{\small (B): Pairwise move by agent 1, 6 who form a link together and remove one each.} &
		{\small (C): Connected network, $\tilde{\mu} $. Efficient if $\delta \ge \ubar{\delta} (\hat{Z} )$.}\tabularnewline
	\end{tabular*}
	\caption{\textit{Sorted network is stable but inefficient}. The above three networks depict Example \ref{example:excessiveassortative}. The network in (A) is pairwise (Nash) stable for some parameters and the network in (B) is the only kind of feasible deviation. The network in (C) is an efficient network.\label{fig:excess_network}}\end{figure}
		
\begin{example}\label{example:excessiveassortative} 
	There are six agents; three of high type (1,2,3) and three of low type (4,5,6). Moreover, there is supermodularity, degree quota of two ($\kappa $=2) and constant decay. Define two networks:
	a network with perfect sorting, $\mu =\{{1}{2},{1}{3},{2}{3},{4}{5},{4}{6},{5}{6}\}$, see Figure \ref{fig:excess_network}.A; 	
	a network which is bridged of $\mu $, defined as $\tilde{\mu} =\{{1}{2},{2}{3},{3}{4},{4}{5},{5}{6},{6}{1}\}$, see Figure \ref{fig:excess_network}.C. 
	We show in this example that for a range of decay-factors that $\mu $ is pairwise stable yet suboptimal. 
	
	In this setup there is a unique move which is both feasible and payoff relevant.\footnote{Under pairwise stability at most one link can be formed in a single move. Without transfers all formed links have value and thus deletion of a link always leads to a loss. Thus only coalitional moves where new links are formed can be valuable. For network $\mu $ all links to same types are already formed. Thus only forming a link with the other types.} 
	This move consists of forming a link across types when both delete a link. Such a move could be agents 1,6 forming a link while deleting links to 3 and 4 which we denote $\hat{\mu} =\mu \cup\{{1}{6}\}\backslash\{{1}{3},{4}{6}\}$ and seen as Figure \ref{fig:excess_network}.B. Benefits for agents 1 and 6 from network $\mu $ and deviating from it are:
	\begin{eqnarray*}
		u_{1}(\hat{\mu} )+u_{6}(\hat{\mu} )&=& (1+\delta )\cdot [z (\bar{x} ,\bar{x} )+z (\underaccent{\bar}{x} ,\underaccent{\bar}{x} )]+
		[1+2\delta ]\cdot [z (\bar{x} ,\underaccent{\bar}{x} )+z (\underaccent{\bar}{x} ,\bar{x} )],\\
		&=&(1+\delta )\cdot \tfrac{1}{2}\cdot [Z (\bar{x} ,\bar{x} )+Z (\underaccent{\bar}{x} ,\underaccent{\bar}{x} )]+
		[1+2\delta ]\cdot Z (\bar{x} ,\underaccent{\bar}{x} ),\\
		u_{1}(\mu )+u_{6}(\mu )&=&2\cdot [z (\bar{x} ,\bar{x} )+z (\underaccent{\bar}{x} ,\underaccent{\bar}{x} )]=Z (\bar{x} ,\bar{x} )+Z (\underaccent{\bar}{x} ,\underaccent{\bar}{x} ).
	\end{eqnarray*}
	
	The condition for deviating to $\hat{\mu} $ not being pairwise profitable is $u_{1}(\mu )+u_{6}(\mu )>u_{1}(\hat{\mu} )+u_{6}(\hat{\mu} )$; this is sufficient for pairwise stability due to payoff symmetry in $\mu $ and no transfers. 
	
	We now turn to deriving condition for when segregating is inefficient. The aggregate benefits over all agents of the two networks $\mu $ and $\tilde{\mu} $ is expressed below in the two equations. 
	\begin{eqnarray*}
		U (\tilde{\mu} )&=&(2+\delta )\cdot [Z (\bar{x} ,\bar{x} )+Z (\underaccent{\bar}{x} ,\underaccent{\bar}{x} )]+
		\left[2+7\delta \right]\cdot Z (\bar{x} ,\underaccent{\bar}{x} ), \\
		U (\mu )&=& 3\cdot [Z (\bar{x} ,\bar{x} )+Z (\underaccent{\bar}{x} ,\underaccent{\bar}{x} )].
	\end{eqnarray*}
	
	Sorting is inefficient when: $U (\mu )<U (\tilde{\mu} )$.	The two inequalities governing pairwise stability and inefficiency has the following positive solution:
		\begin{equation*}
		\begin{array}{rcl}
		\ubar{\delta} (\hat{Z} )&=&\frac{\hat{Z} }{\hat{Z} +\tfrac{5}{2}},
		\\
		\bar{\delta} (\hat{Z} )&=&\frac{\hat{Z} }{\hat{Z} +1},
		\end{array}	
		\qquad \hat{Z} =\tfrac{Z (\bar{x} ,\bar{x} )+Z (\underaccent{\bar}{x} ,\underaccent{\bar}{x} )}{2Z (\bar{x} ,\underaccent{\bar}{x} )}-1,
		\end{equation*}
	where $\ubar{\delta} $ and $\bar{\delta} $ are thresholds for respectively when network $\mu $ becomes inefficient and unstable when $\delta $ increases.		
	\smallskip		
\end{example}

\textcolor{black}{The example above demonstrates that sorting can be inefficient when there are network effects despite presence of supermodularity. The inefficiency stems from a novel source - the pairwise formation of links. The intuition is that under pairwise deviation the two agents do not internalize the total value created for other agents number of indirect links between a high and a low agent. Note that the above example has a very close correspondence to Propositions 1.ii and 4.iii in \cite{de_marti_segregation_2016} - the only addition of their setup to our example is that they allow for a variation in the number of agents for each type.}

We proceed with a generalization of the example above which holds for various structures of the subnetworks within types and for multiple types. As we proceed with our generalization we introduce two new concepts to convey the results.

Let the subset of agents who a link in $\mu $ be denoted $N (\mu )=\{i \in N :\exists  j  \mbox{ s.t. } i j \in \mu  \}$.
A \textit{component}, $\tilde{\mu} $ of network $\mu $ is a subnetwork (i.e. $\tilde{\mu} \subseteq \mu $) where for any agent $i \in N (\tilde{\mu} )$ it holds that: agent $i $ is connected in $\tilde{\mu} $ to all other agents $j \in N (\tilde{\mu} )$, and; for all $j \in N $ if $i j \in \mu $ then $i j \in \tilde{\mu} $. A perfectly sorted network is \textit{locally connected} if there is one component for each type; the set of networks which are perfect sorted and locally connected is denoted $M ^{p-srt:conn }$.

	Our next result is to extend the above example to a general class of sorted networks, denoted by the set $\hat{M} $. The set of network $\hat{M} $ is formed by networks where is only links to same-type agents, all agents use the full quota of available links and all agents of same type are connected.

	\begin{theorem}\label{claim:suboptimal-sorting_hyperbolic_decay}
	Suppose there is supermodularity, hyperbolic decay, more than one partner is allowed and type self-sufficiency then:
	\begin{enumerate}[(i)]
		\item any network in $\hat{M} $ is inefficient when $\delta >\ubar{\delta} $, i.e. $(\hat{M} \cap M ^{\max  U }_{\delta >\ubar{\delta} })=\emptyset$; 
		\item any network in $\hat{M} $ is pairwise stable when $\delta \le \bar{\delta} $, i.e. $\hat{M} \subseteq M ^{p-stb }_{\delta \le \bar{\delta} }$;
		\item if $|X |=2$ then a network is in $\hat{M} $ iff. it is pairwise stable, i.e. $\hat{M} =M ^{p-stb }_{\delta \in (0,\bar{\delta} )}$ 
		\item the set of networks $\hat{M} $ is non-empty when all agent type has even number greater than the degree quota, i.e. $\hat{M} \ne\emptyset$ if for all $x\in  X:$ $(\kappa\cdot  n_x)\in 2\mathbb{N}$.
	\end{enumerate}
	where $\hat{M} =M ^{p-srt:conn }\cap M ^{no-slack }$ and threshold bounds $\ubar{\delta} ,\bar{\delta} $ are defined as:
	\begin{eqnarray*}	\ubar{\delta} &\le &\min _{x ,\tilde{x} \in X }\left(\tfrac{\hat{Z}_{x ,\tilde{x} }}{\hat{Z}_{x ,\tilde{x} }+\frac{1}{2}n_x n_{\tilde{x} }}\right),\qquad 
	\hat{Z}_{x ,\tilde{x} }=\tfrac{Z (x ,x )+Z (\tilde{x} ,\tilde{x} )}{2Z (x ,\tilde{x} )}-1
	\\
	\bar{\delta} &\ge &\min _{x ,\tilde{x} \in X }\left(
	\tfrac{\hat{Z}_{x ,\tilde{x} }}{\hat{Z}_{x ,\tilde{x} }+\max (n_x ,n_{\tilde{x} })-|n_{x }-n_{\tilde{x} }|\cdot \hat{z}_{x ,\tilde{x} }}\right),\qquad \hat{z}_{x ,\tilde{x} }=\frac{z(\arg\min _{x ,\tilde{x} }n_{x },\,\arg\max _{x ,\tilde{x} }n_{x })}{Z (x ,\tilde{x} )}.
	\end{eqnarray*}
	\end{theorem}

	The theorem above is applicable to numerous settings where social networks are formed. If or example schools or firms pre-sort individuals according to talent or otherwise this may lead to a sorted network which is stable but suboptimal. A specific example could be a school where sorting by academic performance type is common in many countries and is often known as tracking. In such cases the sorting induced by the institution could lead to a stable network with no linking across despite links across having potential to raise welfare. 
	
	An alternative version of the above theorem under constant decay can be found in Appendix \ref{app:oversort_generalize} in Theorem \ref{claim:assort:ext:pairwise:excess}.	For both Theorems \ref{claim:suboptimal-sorting_hyperbolic_decay} and \ref{claim:assort:ext:pairwise:excess} one can view the conditions for stability of sorting as generalizing not only Example \ref{example:excessiveassortative} but also Propositions 1.ii and 4.iii from \cite{de_marti_segregation_2016}. These special cases of sorting correspond to a situation with two types and all agents of a given type link with one another, i.e. when  $n_{\bar{x} }=n_{\underaccent{\bar}{x} }=\kappa +1$ and $|X |=2$.
	The other part of the two theorems is on inefficiency of sorting and generalizes Example \ref{example:excessiveassortative} and Proposition 6 from \cite{de_marti_segregation_2016}. It advances their proposition by removing the 	restriction to two types and linking between all same type agents as well as doing away with the limitation of having very few agents. For instance, if there is an equal number of agents for each type then Proposition 6 in \cite{de_marti_segregation_2016} is valid only for five or less agents of each type (ten agents in total).

	A visualization of how the thresholds of externalities depend on population size and strength of complementarity is found in Figure \ref{fig:hyperbolic_threshold}. The upper part of the figure keeps the population size fixed while lower ones keep the complementarity strength fixed. In the upper part it is evident that the connection thresholds both approximately follow power-laws in the number of agents.

	\begin{figure}[b!]
		\includegraphics[width=1.0\textwidth]{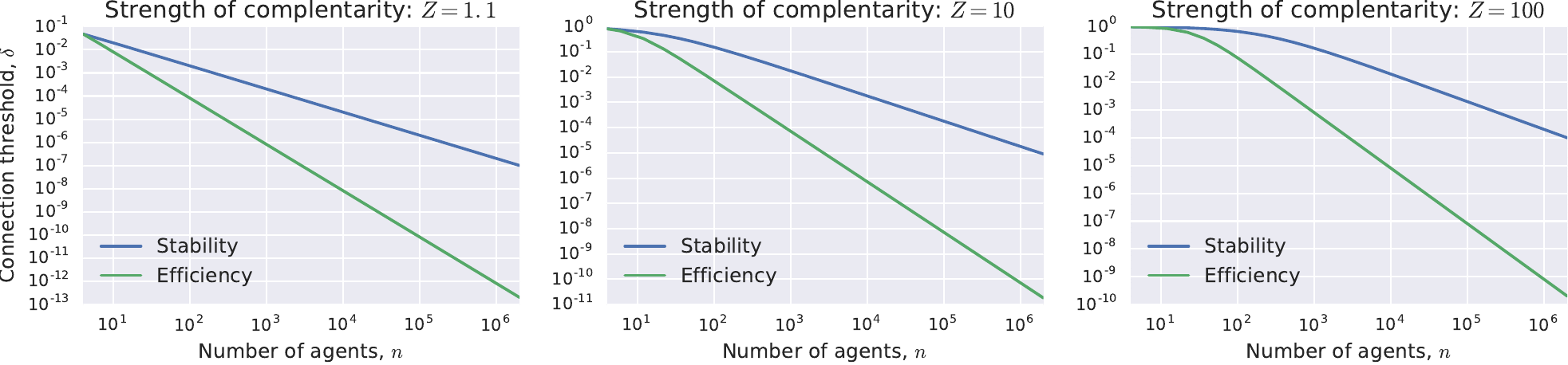}
		\vspace{0.2cm}					
		\includegraphics[width=1.0\textwidth]{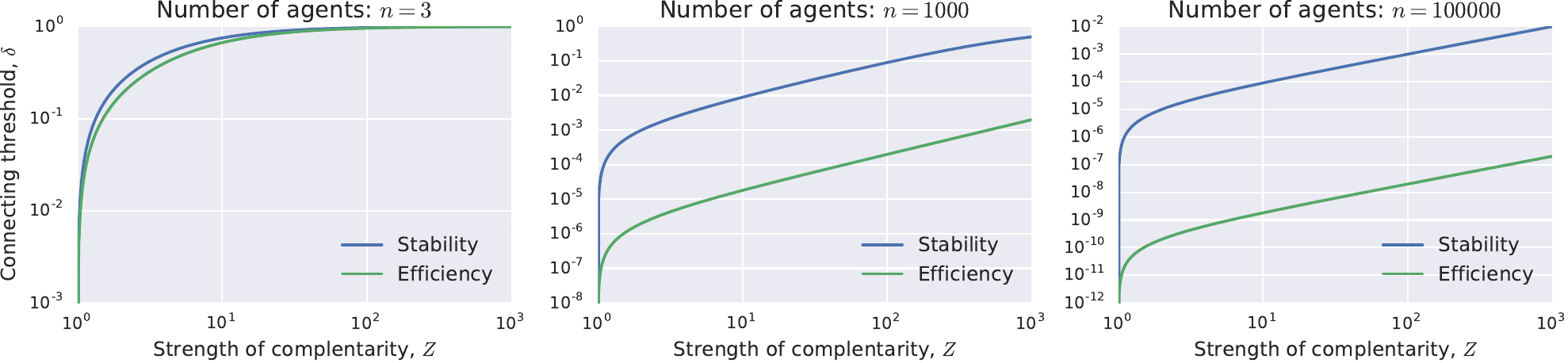}			
		\caption{\textit{Thresholds for connecting.} Visualization of thresholds for connecting from Theorem \ref{claim:suboptimal-sorting_hyperbolic_decay}. The upper part varying sizes of populations and fixed strength of complementarity. The lower part has has varying strength of complementarity and fixed populations size. It is assumed there are two types with an equal number of agents.}
		\label{fig:hyperbolic_threshold}
	\end{figure}

		The remainder of this subsection will sketch a policy intervention that can mitigate the problem of suboptimal sorting by improving welfare through encouraging connection. These interventions can be seen more generally as a design problem where the policy maker intervenes to induce a network that produces higher welfare.  		
		The tool that the policy maker employs is providing incentives to agents for forming specific links. 	
		Define a \textit{link-contingent contract} as a non-negative transfer $\mathcal{C}_{i j }$ to $i $ for linking with another agent $j $. Denote the vector of link-contingent contract as $\mathcal{C} $.
	
	\begin{definition}
		Let a network $\tilde{\mu} $ be \textbf{implementable} from $\mu ,\tau $ given $\mathcal{C} $ if there exist a sequence of tuples $(\mu_0,\tau_0),..,(\mu_l ,\tau_l )$ where $\mu_0=\mu $, $\mu_l =\tilde{\mu} $ and $\tau_l =\tau $ such that: 
		for $q =1,..,l $ from $\mu_{q -1}$ to $\mu_{q }$ is a feasible pairwise move which increases the pair's net-utility most given $\mathcal{C} $, and; $\tilde{\mu} $ is pairwise stable given $\tau_{q }$ and $\mathcal{C} $.
	\end{definition}

	\begin{proposition}\label{claim:policy_accomodate_oversorting}
		Suppose that conditions for suboptimal sorting from Theorem \ref{claim:suboptimal-sorting_hyperbolic_decay} are valid and there are two types then a policy maker can implement a welfare improving network when $\delta \in (\ubar{\delta} ,\bar{\delta} )$ from a sorted network $\mu \in \hat{M} $.
	\end{proposition}
	
	Note that the individual compensation for paid to agents for connecting to others may not be equal. In  particular the pay may also depend on the types. This is the case when there is both supermodularity and monotonicity in $Z $ then agents may receive compensation increasing their type. This would be the case in a sorted suboptimal network with no transfers where components have same number of agents and are isomorphic to another.

\paragraph{Large population}
We finalize this section by investigating what pattern of linking is exhibited when the set of agents becomes asymptotic infinite. In this large matching market we examine asymptotic perfect sorting, i.e. when measured share of links to same type agents converges to one:
\begin{definition}	
	Let \textbf{asymptotic perfect sorting} hold for a sequence of networks sets $\tilde{M }_{n }$ if for any network $\mu \in \tilde{M }_{n }$ where $n \rightarrow\in fty$ it holds that 
	$|\{i j \in \mu :\,x_{i }=x_{j }\}|/|\mu |\simeq1.$
\end{definition}			

Define \textit{asymptotic independence} as $\delta <(\kappa -1)^{-1}$. For large matching markets the sufficient conditions for asymptotic perfect sorting to emerge in strongly stable networks are:
\begin{proposition}\label{claim:assort:ext:asymptoticperfectPAM}
	If there is supermodularity, a degree quota and constant decay with asymptotic independence then there is asymptotic perfect sorting for strongly stable networks.				 	
\end{proposition}

The result above demonstrates that the availability of many agents for linking induces perfect sorting in strongly stable networks. This is the same prediction as the conclusion of \cite{becker_theory_1973} for the marriage market model but holds in the presence of externalities with constant decay. Note that the result relies on efficiency requiring perfect sorting and using this to infer properties of strong stability. 
However, we may interpret it differently, there is no loss of sorting when there are many agents.

\section{Network structure}\label{sec:network:structure}

	In this section we examine the structure of stable networks when measured by specific measures of network centrality such as the degree (i.e. number of links). 
	In order to make progress in this question it is necessary to change restrictions on utility. The most important new restriction in this context is on the value of linking where monotonicity in type is assumed; that is, a higher typesproduce higher total value when fixing the type of the partner. Although this new restriction may seem ad hoc it has a natural correspondence with circumstances where some agent types always increase value of the linking compared to others e.g. by being more capable or have better social skills.

\begin{definition}\label{def:monotone:linkvalue} The link value has \textbf{monotonicity} if 
	$\frac{\partial }{\partial x } Z (x ,y )>0$.
\end{definition}

	In contrast with the previous section here we model costs with an explicit cost function. The cost function for each agent is $c :\mathbb{N}_0\rightarrow\mathbb{R}$ which is positive, strictly increasing and strictly convex in degree (i.e. number of links), $k_{i }$.

\subsection{No externalities}

	We begin with a setting where there are no externalities. The relevant measure of centrality without externalities are agents' number of links, i.e. degree. This measure is commonly used in applied in research on networks.
	
	We start with the following question: what is the relation between network measures of centrality and type? The short answer is that type and connectivity go hand in hand: higher type implies (weakly) higher centrality. We call this degree monotonicity:

\begin{definition}\label{def:degree_mon}
	A set of networks $M $ has \textbf{degree monotonicity} in type if:
	\[\forall \mu \in M :\,\, x_{i }>x_{j } \Rightarrow k_{i }(\mu )\ge  k_{j }(\mu ).\]
\end{definition}

\begin{proposition}
	\label{claim:network:monotonicdegree:talent}	Suppose there are no externalities and monotonicity in link value then the set of pairwise stable networks has degree monotonicity.
\end{proposition}

	The intuition behind the above result is that the marginal benefits for forming an additional link is strictly increasing in type when the link value is monotonic. Thus an agent with higher type has an increased incentive to form links jointly with other agents.

	We turn to investigating whether or not more gregarious agents (agents with many partners), tend to associate with agents who are as gregarious as themselves? That is, is there sorting in degree. The simple answer is yes, under certain circumstances as elaborated below. 
	
	We build our intuition on two established results: firstly, there is a positive relationship between higher talent and (weakly) more partners (Proposition  \ref{claim:network:monotonicdegree:talent}), and; secondly, partners are chosen assortative in talent (Proposition \ref{claim:assort:elementwise:monotonic}). By combining these two claims it can be shown there will also be degree assortativity. Although the intuition of combining the two previous results is straightforward additional structure needs to be imposed to show the result. Let \textit{complete heterogeneity} be satisfied if for every agent has a unique level of talent, i.e. $n =|X |$. 
	
	Our concept of degree sorting is defined analogously to sorting in type. That is, every agent with strictly higher degree must have partners which weakly dominate in degree when compared partner-by-partner with the partners of agent with lower degree:
	
	\begin{definition}
		A network $\mu $ has \textbf{sorting in degree} when:		
		\[\forall  i,j\in N , k_{i }(\mu )>k_{j }(\mu ):\qquad 
		\mathcal{K} (\nu_{i }(\mu )/\{j \})_l \ge \mathcal{K} (\nu_{j }(\mu )/\{i \})_l ,\qquad l =1,...,k_{j }(\mu ),\]
		where $\mathcal{K} (t )$ as the vector of degrees for agents in coalition $t $ ordered descending.
	\end{definition}

	\begin{proposition}\label{claim:network:elementwise:monotonicdegree} If there are supermodularity, monotonicity in link value as well as complete heterogeneity but no externalities then every pairwise stable network has sorting in degree.\end{proposition}
	
	The above result is identical to Proposition \ref{claim:assort:elementwise:monotonic} with the exception that the assortative property in this result is degree and this result has an added requirement. The emergence of this assortative pattern is not only a theoretical curiosity - this pattern is an underlying property observed across various empirical social networks, see the literature review. 
	
	The feature emerging in Proposition  \ref{claim:network:elementwise:monotonicdegree} is also known as degree assortativity is relevant as unlike \cite{currarini_economic_2009}
	
	When we move to a setting of externalities our pattern of sorting by degree does not hold in general. The reason is that either sorting by talent or monotonic degree may fail. This is captured by the follows from Corollary \ref{claim:core_fail_externalities} in the next section as it shows both sorting in type and monotonic centrality may fail.
	
	\subsection{Externalities}
	
	We proceed by exploring the case of network externalities. We limit the examination to the relation between type and centrality, not neighbors' centrality.
	
	We generalize the pattern of monotonic centrality in talent. However, in order to gain generality the restrictions on the link value $Z (\cdot )$ and a stricter equilibrium concept. We define monotonicity of  $\delta $-decay centrality by replacing the measure in the definition of degree monotonicity. Also define \textit{no modularity} when $\frac{\partial ^{2}}{\partial  x  \partial  y } Z (x ,y )=0$.

	\begin{proposition}\label{claim:monotonic_decaycentrality}
		Suppose there are externalities as well as monotonicity and no modularity then the set of strongly stable networks has $\delta $-decay monotonicity.
	\end{proposition}
	
	The result above emerges from the fact that there are only monotone and independent effects of more talent (due to the absence of modularity). The independence of effects entails that the link value can be split into a part for each partner, $Z (x ,\tilde{x} )=\hat{Z }(x )+\hat{Z }(\tilde{x} )$. It is possible to show that a given agent $i $'s total contributions to the network is the product between $\hat{Z }(x_{i })$ and the sum of network weights. Under these circumstances monotonicity implies that more talented agents must have higher measure of network weights or else efficiency implied by strong stability is violated.

	We conclude our analysis of monotonic centrality by pointing to the limitations of our results. 
	When no modularity is violated for instance then Proposition \ref{claim:monotonic_decaycentrality} may not hold, see
	Example \ref{ex:supermod:nonmonotonic:centrality} found in Appendix \ref{app:undersorting}. Likewise by relaxing strong stability in Proposition \ref{claim:monotonic_decaycentrality} to pairwise stability it follows that monotonic centrality also fails. This is seen from Example \ref{example:monotonic_fail} and Corollary \ref{claim:core_fail_externalities} in the next section.

	\section{Limits to sorting and centrality monotonicity} \label{sec:limits}
	
	The previous two sections derive patterns for sorting and network centrality which we readdress here. We show the failure of those patterns is possible despite fulfilling some of the conditions that are sufficient in the previous sections: monotone and supermodular link value. We focus on a generic network where the agent with least type is in the center:
	
	\begin{definition}
		A \textbf{star}, $\mu $, is a network where some agent called the center of $\mu $ and each other agent has only one single link is this a link to the center. 
	\end{definition}

	\begin{definition}
		A star is \textbf{low value sponsored} when its center has the minimal type.
	\end{definition}

	When there is monotonicity a low value sponsored star lacks efficiency. The lack of inefficiency stems from the fact that under monotonicity substituting a high-value agent to become the center will increase the value of the network. The following example demonstrates that the low value sponsored star can be  stable under monotonicity and thus also inefficient:	
	
	\begin{example}\label{example:monotonic_fail} Let there be three agents where agents 1,2 are high types while agent 3 is low type. Suppose additionally there is weak supermodularity, strict monotonicity and and convex costs. Moreover network externalities have either constant or hyperbolic decay.

		Let network $\mu $ be the low value sponsored star, i.e. $\mu =\{13,23\}$. Network $\mu $ is inefficient as $\tilde{\mu} =\{12,13\}$ provides strictly higher aggregate utility as for $\delta \in [0,1)$: 		
		\begin{eqnarray*}
			U (\tilde{\mu} )&>&U (\mu ),\\
			Z (\bar{x} ,\bar{x} )+(1+\delta )\cdot Z (\bar{x} ,\underaccent{\bar}{x} )&>& \delta \cdot Z (\bar{x} ,\bar{x} )+2Z (\bar{x} ,\underaccent{\bar}{x} ),\\
			(1-\delta )\cdot Z (\bar{x} ,\bar{x} )&>& (1-\delta )\cdot Z (\bar{x} ,\underaccent{\bar}{x} ).
		\end{eqnarray*}
		
		We show that $\mu $ is pairwise stable for certain transfers. We begin with conditions for when deviating is unprofitable when only are deleted links (no links formed). Cutting all links corresponds to an outside option with net-utility of zero - this is Inequality \ref{eq:monotone:fail:ex:deleteall:solo}. Another option is agent 3 deleting one link and getting net value of the remaining link - see Inequality \ref{eq:monotone:fail:ex:deleteone:agent1}: 	
		\begin{eqnarray}
		\label{eq:monotone:fail:ex:deleteall:solo}	
		s_{i }(\mu ,\tau )
		& \ge  &
		0,\qquad  i \in N ,
		\\
		\label{eq:monotone:fail:ex:deleteone:agent1}	
		\tau_{3 i }&\ge & c (2)-c (1)-z ({\underaccent{\bar}{x} },{\bar{x} }),\qquad i \in  \{2,3\}.
		\end{eqnarray}

		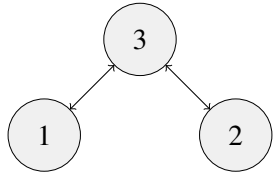
\begin{figure}[b!]
			\vspace{1cm}
			\begin{center}
				\begin{tikzpicture}[nodes={fill=gray!12,circle, ultra thin}, node distance=1.8cm, auto]
				\draw[help lines] (0,0);
				\node[state] (i_1) {3};
				\node[state] (i_2) [below left of=i_1] {1};
				\node[state] (i_3) [below right of=i_1] {2};
				\path[<->] (i_1) edge node [fill=none]{} (i_2);
				\path[<->] (i_1) edge node [fill=none]{} (i_3);
				\end{tikzpicture}
			\end{center}
			\caption{\textit{Limits to monototonic centrality.} The above network depicts the limits of sorting and monotonic centrality from Example \ref{example:monotonic_fail}. The low value agent ($i$=3) is excessively central which is inefficient. Moreover there is insufficient sorting as the high value agents ($i$=1,2) are not linked.}
		\end{figure}
		
		It remains to be shown the condition that prevent any pairwise deviation by forming a link. Only agent 1,2 can form a link. There are three options for forming link 12: deleting no links, deleting one link or deleting all - captured by respectively Inequality \ref{eq:monotone:fail:ex:form:nodelete}, 
		\ref{eq:monotone:fail:ex:form:onedelete} and \ref{eq:monotone:fail:ex:onlyonelink:allbutonelink}:
		\begin{eqnarray}
		\label{eq:monotone:fail:ex:form:nodelete}
		2\cdot [c (3)-c (2)]
		& \ge  &
		(1-\delta )\cdot Z (\bar{x} ,\bar{x} ),
		\\
		\label{eq:monotone:fail:ex:form:onedelete}
		s_{1}(\mu ,\tau )+s_{2}(\mu ,\tau )
		& \ge  &
		Z (\bar{x} ,\bar{x} )+z (\bar{x} ,\underaccent{\bar}{x} )\cdot (1+\delta )-c (1)-c (2)+\max \{\tau_{13},\tau_{23}\},
		\\
		\label{eq:monotone:fail:ex:onlyonelink:allbutonelink}
		s_{i }(\mu ,\tau )+s_{j }(\mu ,\tau )
		& \ge  &
		Z_{i j }-2\cdot c (1),\qquad i \ne j .
		\end{eqnarray}

		We conclude the example by indicating when the above conditions are satisfied. From Lemma \ref{claim:example:monotonic_fail:addition} in Appendix \ref{app:undersorting} the sufficient conditions for Inequality (\ref{eq:monotone:fail:ex:deleteall:solo})-(\ref{eq:monotone:fail:ex:onlyonelink:allbutonelink}) to be satisfied are, linear costs ($\tilde{c} $ of forming a link) and $\delta $ being sufficiently large as well as, 
		\[2\cdot Z (\bar{x} ,\underaccent{\bar}{x} )+Z (\bar{x} ,\bar{x} )\ge 4\tilde{c} \quad\mbox{and}\quad
		Z (\bar{x} ,\underaccent{\bar}{x} )>\tilde{c} .\]				
	\end{example}

	The reason why the low value sponsored star is stable in Example \ref{example:monotonic_fail} is that agents fail to coordinate on displacing the low value agent with a high value agent as center. The failure arises despite having the potential to make everyone better off. 
	
	The intuition from Example \ref{example:monotonic_fail} can extended to a more general result - low value sponsored stars may be stable even for many agents. This failure is expressed in the following result:\footnote{
		The proof follows from the auxiliary results in Lemma  \ref{claim:monotonic_fail} and  \ref{claim:monotonic_fail_remark} from Appendix \ref{app:undersorting}.}

	\begin{corollary}
		\label{claim:core_fail_externalities}
		If there are externalities, monotonicity, supermodularity and a cost function then the low value sponsored star is inefficient and can be pairwise stable.
	\end{corollary}
	
	The above result can be seen as an analogy to the results \cite{farrell_standardization_1985} and \cite{katz_network_1985}. In their work on markets of consumers and sellers under consumer externalities a loss productive company may become the main supplier despite this being inefficient. This is similar to this result, the agents are locked in to an inefficient star despite this being inefficient. This can be applied to understand how decentralized organizations may fail to promote the best candidate to become the leader but rather choose the incumbent.
	
	As alluded to previously the above corollary can also be interpreted in terms of its implications for networks:
	
	\begin{corollary}
		\label{claim:fail_corollary} When there are externalities then under pairwise stability: sorting in type may fail; monotonic degree centrality and $\delta$-decay centrality may fail.
	\end{corollary}
	
	The corollary above demonstrates that despite meeting the conditions of \cite{becker_theory_1973} these conditions are not sufficient for sorting in social networks with externalities.

\section{Concluding discussion}\label{sec:discussion}		
	
	Our analysis is based on strict assumptions which we discuss in this final section. 
	We avoid discussing sorting under search as there is a large literature e.g.  \cite{shimer_assortative_2000}. 
	
	One severe caveat with our analysis, and stable networks in general, is that these networks may not exist. The classic example is the room-mate problem, cf. \cite{gale_college_1962}. Furthermore, the gross substitutes conditions from \cite{kelso_job_1982} which ensures existence of stable matchings in related settings are not satisfied in our setting with externalities.\footnote{The lack of gross substitutes is due to the fact that a change in one active link can imply a change the value of other links. This fact will violate gross substitutes.} Neverthless, the scope in lack of existence may not be excessive. We can derive the cost of stability (the necessary payments to induce stability) from \cite{bachrach_cost_2009}.
	
	There is also a number of restrictive assumptions on payoff. The most crucial assumption are payoff separability and fixed structure of externalities. However, the results should be robust, for instance to introduction of utility from triads etc. which is common in the economic literature on networks.
	Another critical assumptions is supermodularity and monotonicity along with perfect transferability. Nevertheless, as mentioned in the introduction, these two assumptions can be replaced by monotonicity in individual link value and perfect non-transferability, which is also in line with some research on peer effects.\footnote{Or more broadly by generalized increasing in differences from \cite{legros_beauty_2007}.} 
	Finally although our model did not include a measure of effort this can incorporated by rewriting the setup as in \cite{baccara_homophily_2013}.
	
	Finally it is important to note that the derivation of Proposition \ref{claim:assort:ext:asymptoticperfectPAM} and \ref{claim:monotonic_decaycentrality} rely on efficiency being a necessary condition for strong stability to hold. Although efficiency is a unique property for strong stability (and does not hold for weaker concepts) then strong stability should be seen as a refinement with desirable properties which makes it more likely when it exists.\footnote{In some circumstances the existence of contracts where an agent may subsidize or penalize another agent's link formation with alternative agents may imply that strong stability even if contracts were limited to being pairwise specified, cf. \cite{bloch_formation_2007}. }

\bibliographystyle{model2-names}
\bibliography{anf_abn}

\appendix
\section{Auxiliary results}\label{app:equivalence:strong:pair}

	\begin{lemma}\label{claim:auxiliary:equivalence:strong:pairwisetype}
		In the absence of network externalities then the set of strongly stable networks is equivalent to the set of  pairwise stable networks, i.e., $M ^{p-stb }_{\delta =0}=M ^{s-stb }_{\delta =0}$.
	\end{lemma}

	\begin{proof}		
	By definition it holds that $M ^{p-stb }_{\delta =0}\subseteq M ^{p-stb }_{\delta =0}$, thus we need to show that $M ^{p-stb }_{\delta =0}\subseteq M ^{p-stb }_{\delta =0}$ to prove the claim. This claim is shown using similar to arguments to \cite{klaus_stable_2009}'s Theorem 3.i.
	
	Let $\mu $ with associated contracts $\tau $ be a network which is blocked by a coalition. It will be shown that for every coalition $t \in T $ that blocks, within the coalition there is a subset of no more than two members that also wishes to block the network. 
	Let $\tilde{\mu} $ be the alternative network that the blocking coalitions implements through a feasible coalition move and $\tau $ be the transfers associated with $\tilde{\mu}$.
	
	It is always possible to partition the set of deleted links  $\mu \backslash\tilde{\mu} $ into two: 
	(i) a subset denoted $\hat{\mu} $ where for each link $ij$ that can be deleted where one of the two partners can benefit, i.e. it holds that either $z_{i j }+\tau_{i j }-[c_{i}(\mu )-c_{i }(\mu \backslash{i }{j })]<0$
	or $z_{j i }+\tau_{j i }-[c_{j }(\mu )-c_{i }(\mu \backslash{i }{j })]<0$; (ii) a subset denoted $\breve{\mu} $ where for each link ${i }{j }$ neither of the previous two inequalities are satisfied.
	
	Suppose that the first partition is non-empty, i.e. $\hat{\mu}  \ne\emptyset$.  However, as deleting links can be done by a single agent on its own then the move only takes needs the coalition of that agent to delete the link. Thus any part of a coalitional move that only involves profitably removing links can be performed in parts by a coalition with a single agent - therefore this move is also a pairwise block. 
	
	Thus it remains to be shown that the remaining part of coalitional move also can be performed as a pairwise block, i.e. when forming $\tilde{\mu} \backslash\mu $ and deleting $\breve{\mu} $. 
	This part of the coalitional move must entail forming links as no links can be deleted profitably. The set of formed links $\tilde{\mu} \backslash\mu $ can be partitioned into a number of $|\tilde{\mu} \backslash\mu |$ feasible submoves of adding a single link while deleting links by each of the agents $i $ and $j $ who form a link. The feasibility for each of the partitioned moves is always true when there is a cost function as moves are unrestricted. It is now argued that each of the partitioned moves are feasible when there is a degree quota. If the network $\mu \cup{i }{j }$ is feasible then the move of simply adding the link is feasible. If $\mu \cup{i }{j }$ is not feasible, then agents $i $ and $j $ can delete at most one link each and if both $\mu $ and $\tilde{\mu} $ are feasible then this also feasible as the degree quota is kept.
	
	For the coalitional move to $\tilde{\mu} $ it must be that at least at least one link among the implemented links $\tilde{\mu} \backslash\mu $ has a strictly positive value that exceeds the loss from deleting at most one link for each of two agents forming the link. This follows as it is known that deleting one or more links cannot add any value and thus must have weakly negative value and that by definition the total value to the blocking coalition must be positive. As every one of the partitioned moves is feasible, it follows that for every coalitional move there are two agents who can form link while potentially destroying current links and both be better off. In other words, for every coalition that blocks, there is a pairwise coalition that blocks. 
	\end{proof}

	\begin{fact}\label{claim:construct:network:degreequota:size}
	For every $\kappa ,n $ such that $n >\kappa $ and $n \cdot \kappa $
	is even there exists a network $\mu_{n ,\kappa }$ where all agents have exactly $\kappa $ neighbors. Moreover, if $\kappa \ge 2$ then $\mu_{n ,\kappa }$ is connected.
	\end{fact} 
	
	\begin{proof}
	Suppose $n $ is even. Let $\%$ be the modulus operator. We can construct the following networks.
	\begin{eqnarray*}
	\hat{\mu}_{n ,\kappa } & = & \left\{i j :\,i \in \{1,...,\tfrac{n }{2}\},\,
	j \in \{(\tfrac{n }{2}+i \,\%\,\tfrac{n }{2}),...,(\tfrac{n }{2}+\left[i +\kappa -1\right]\,\%\,\tfrac{n }{2}\})\right\},\,\,\kappa \le \tfrac{n }{2},\\
	\tilde{\mu}_{n ,\kappa } & = & \begin{cases}
	\hat{\mu}_{n ,\kappa }, & \kappa \le \frac{n }{2},\\
	\mu_{c}\backslash\hat{\mu}_{n ,n -\kappa -1}, & \kappa >\frac{n }{2}.
	\end{cases}
	\end{eqnarray*}

	Letting $\mu_{n ,\kappa }=\tilde{\mu}_{n ,\kappa }$ is sufficient for $n $ is
	even. 
	When $n $ is odd we know that $\kappa $ is even and thus we can use the following amended procedure instead:
	\begin{eqnarray*}
		\iota_{n  ,\kappa }(\iota )=
	\begin{cases}
		\tfrac{n -1}{2}+\iota ,&\kappa \le \tfrac{n -1}{2}\\
		\tfrac{n -1}{2}+(\iota +\kappa )\,\%\,\tfrac{n -1}{2},&\kappa >\tfrac{n -1}{2}
	\end{cases}
	\end{eqnarray*}
	\begin{eqnarray*}
			\mu_{n ,\kappa } & = &
		\tilde{\mu}_{n -1,\kappa }\backslash\left\{i j :\,i \in \{1,..,\tfrac{\kappa }{2}\},\,j =\iota_{n  ,\kappa }(i )\right\}\cup
		\{i j :\,i =n ,\,j \in (\cup_{\iota \{1,..,\tfrac{\kappa }{2}\}}\{\iota ,\iota_{n  ,\kappa }(\iota )\})\}		
	\end{eqnarray*}
	
	We now show that if $\kappa \ge  2$ it follows that $\tilde{\mu}_{n ,\kappa }$ is connected. Assume that $n $ is even and suppose $\kappa \le \tfrac{n }{2}$; for any $i \in N : i <\frac{n }{2}$ where $i' =i +1$ and let $j =\frac{n }{2}+i +1$ where $i j ,i' j \in \tilde{\mu}_{n ,\kappa }$; thus for all $i ,i' \in \{1,..,\frac{n }{2}\}$ it holds that $p_{i i' }(\tilde{\mu}_{n ,\kappa })<\in fty$. In addition, as for any $i \in N : i \le \frac{n }{2}, j =\frac{n }{2}+i $ it holds that $i j \in \tilde{\mu}_{n ,\kappa }$ it follows that $\tilde{\mu}_{n ,\kappa }$ is connected. If instead $\kappa >\tfrac{n }{2}$ then by construction   $i i' \in \tilde{\mu}_{n ,\kappa }$ if either $\max (i ,i' )\le \frac{n }{2}$ or $\min (i ,i' )>\frac{n }{2}$ as $i i' \notin\hat{\mu}_{n ,n -\kappa -1}$. Moreover, for $i \in N : i <\frac{n }{2}$ and $j =\tfrac{n }{2}+(i +\kappa )\,\%\,\tfrac{n }{2}$ it holds that $i j \notin\hat{\mu}_{n ,n -\kappa -1}$; thus $i j \in \tilde{\mu}_{n ,\kappa }$. Therefore $\tilde{\mu}_{n ,\kappa }$ must be connected. 
	
	Assume instead that $n $ is odd. By the above argument there are at least two connected subnetworks consisting of agents in $\cup_{\iota \{1,..,\tfrac{\kappa }{2}\}}\{\iota ,\iota_{n  ,\kappa }(\iota )\}$ and agents who are connected through agent, $n $, i.e. $N \backslash(\cup_{\iota \{1,..,\tfrac{\kappa }{2}\}}\{\iota ,\iota_{n  ,\kappa }(\iota )\})$. 	 If $\kappa \le \tfrac{n -1}{2}$ where $i =\frac{\kappa }{2}$, $i' =\frac{\kappa }{2}+1$ and $j =\tfrac{n -1}{2}+\frac{\kappa }{2}+1$ then $i j ,i' j \in \tilde{\mu}_{n ,\kappa }$ and thus $\tilde{\mu}_{n ,\kappa }$ is connected. If $\kappa >\tfrac{n -1}{2}$ where $i =\frac{\kappa }{2}$, $i' =\frac{\kappa }{2}+1$ and $j =\tfrac{n -1}{2}+(\iota +\kappa +1)\,\%\,\tfrac{n -1}{2}$ then $i j ,i' j \in \tilde{\mu}_{n ,\kappa }$ and thus $\tilde{\mu}_{n ,\kappa }$ is connected.	
	\end{proof}

\section{Proofs}\label{app:proofs}

\subsection{Sorting: no externalities}\label{app:assort:talent}

\noindent
\textbf{Proposition \ref{claim:assort:elementwise:monotonic}:} If there is supermodularity and no externalities then for any pairwise stable network it holds that there is sorting in type.
	\begin{proof}
	Suppose the claim is false. Let $q $ be the lowest index for which the condition fail: for all $l <q $ it holds that $\mathcal{X} (\nu_{i }(\mu )/\{j \})_l \ge \mathcal{X} (\nu_{j }(\mu )/\{i \})_{l +l ^{*}}$ where $l ^{*}=\max (k_j (\mu )-k_i (\mu ),0)$. Thus there are two agents $i' ,j' $ such that:
	
	\[\begin{array}{rll}
	x_{j' }&=\mathcal{X} (\nu_{j }(\mu ))_{q }, & j' \in (\nu_{j }(\mu )\backslash(\nu_{i }(\mu )\cup\{i \})),\\
	x_{i' }&<\mathcal{X} (\nu_{j }(\mu ))_{q }, & i' \in (\nu_{i }(\mu )\backslash(\nu_{j }(\mu )\cup\{j \})).
	\end{array}\]
	
	Recall $k ^{*}=\min (k_{i }(\mu ),k_{j }(\mu ))$. The argument why there must exist an agent $i' $ in $\nu_{j }(\mu )$ but not in  $(\nu_{j }(\mu )\cup\{j \})$ is that 
	$|\{\iota \in \nu_i (\mu ): x_\iota <x_{j' }\}|>
	 |\{\iota \in \nu_j (\mu ): x_\iota <x_{j' }\}|$. This follows as by construction it holds that $|\{\iota \in \nu_i (\mu ): x_\iota <x_{j' }\}|=k ^{*}-q +1$ and $|\{\iota \in \nu_j (\mu ): x_\iota <x_{j' }\}|\le k ^{*}-q $. 
	
	The agents are such that $x_{i }>x_{j },x_{i' }<x_{j' }$ as well as ${i }{j' },{j }{i' }\notin\mu $. However, this fact implies that there is a violation of strong stability: agents $i ,i' ,j ,j' $ can deviate by destroying $\{ i j ,i' j'  \}$ and forming $\{i j' ,i' j \}$ and thus increase payoffs due to supermodularity (cf. Equation \ref{eq:supermodularity}). From Lemma \ref{claim:auxiliary:equivalence:strong:pairwisetype} it follows that pairwise stability is also violated if strong stability is violated.
	\end{proof}

	\subsection{Sorting: externalities - finite population}\label{app:inefficient:assort}
		
		\begin{lemma}\label{claim:sametype:noncritical:link}
			Suppose that $\min _{x\in  X} n_x>\kappa$, $\kappa\ge 2$;  if $\exists  i\in  N:$ 
			i) $|\{i'\in  \nu_i(\mu): x_{i'}=x_i\}|<n_x-1$,
			ii) $\min _{i'\in  N_x\backslash\nu_i(\mu)}k_{i'}(\mu)=\kappa$,  and 
			iii) $\max _{i'\in  N_x\backslash\nu_i(\mu)}|\{i''\in  \nu_{i'}(\mu): x_{i''}\ne x\}|=0$ 
			then $\exists  i',i''\in  \mu$ such that $i',i''\notin \nu_i(\mu)$ and $p_{i'i''}(\mu\backslash\{i'i''\})<\in fty$			
		\end{lemma}
		\begin{proof}
			Suppose that for $i\in  N$ the conditions i)-iv)  are met but the lemma is not true.
			If $i' \in  N_x$ and $ii'\notin\mu$ then there must exist some $i''\in  N_x$ such that $i'i''\in  \mu$ and $i''\notin \nu_i(\mu)$ as $\min _{i'\in  N_x\backslash\nu_i(\mu)}k_{i'}(\mu)=\kappa$  and 
			$\max _{i'\in  N_x\backslash\nu_i(\mu)}|\{i''\in  \nu_{i'}(\mu): x_{i''}\ne x\}|=0$. If $p_{i'i''}(\mu\backslash\{i'i''\})<\in fty$ then the proof is terminated so we must assume $p_{i'i''}(\mu\backslash\{i'i''\})=\in fty$.

			As $p_{i'i''}(\mu\backslash\{i'i''\})=\in fty$ then there must exist two components in $\mu$,  $\mu',\mu''\subseteq \mu\backslash(\{\iota'_l\iota''_l\}$, where in each component $\mu'$ or $\mu''$ there are at least $\kappa+1$ agents of type $x$ (as for any $\iota\in (\nu_{i'}(\mu)\cup\nu_{i''}(\mu))$ it holds that $x_\iota=x$). Moreover, $i$ can at most be connected to one of $i',i''$ in $\mu\backslash(\{i'i''\}$ as otherwise $i',i''$ would be connected in $\mu\backslash\{i'i''\}$. Denote the component $\{\mu',\mu''\}$ where $i$ is part of as $\tilde{\mu}$ and define $\tilde{N}=\{\iota\in  N_x\backslash\nu_i(\mu): \exists \iota'\in  N: \iota\iota'\in \tilde{\mu}\}$.
			
			Let $\iota_0 \in  \arg\max _{\iota\in  i',i''}p_{\iota i}$ and iteratively $\iota_{l}\in  \nu_{\iota_{l-1}}(\mu), l\in  \mathbb{N}$. Moreover, there must be only a unique path in $\mu\backslash\{i'i''\}$ between any two agents $\iota,\iota'\in  \tilde{N}$ as otherwise $i\iota,i\iota'\notin\mu$ but $p_{\iota\iota'}(\tilde{\mu}\backslash\{i'i''\})<\in fty$ - in this case we could instead denote $i'=\iota$ and $i''=\iota'$. 
			However, the single unique path entails that at level $l$ or below there are $\sum_{q=0}^{l}(\kappa-1)^{q}$ agents; thus $n_x\ge  \sum_{q=0}^{l}(\kappa-1)^{q}$. Let $l$ be the minimal $q$ such that $\forall \in \iota\in  \tilde{N}: p_{i\iota}\le  q$; as $n_x$ is finite such a $q$ must exist. In addition, as there is a unique path between agents in $\mu$ then any agent $\iota\in \tilde{N}: p_{\iota \iota_0}=l$ has only one link, and thus its degree is less than $\kappa$ (as $\kappa\ge 2$). This violates the condition that all $i'\in  N$ where $x_{i'}=x$ has $k_{i'}=\kappa$.
		\end{proof}

	\noindent
	\textbf{Theorem} \ref{claim:suboptimal-sorting_hyperbolic_decay}:
	Suppose there is supermodularity, hyperbolic decay, a non-singular degree quota and type self-sufficiency then:
	\begin{enumerate}[(i)]
		\item $(\hat{M} \cap M ^{\max  U }_{\delta >\ubar{\delta} })=\emptyset$; 
		\item $\hat{M} \subseteq  M ^{p-stb }_{\delta \le \bar{\delta} }$;
		\item $\hat{M} =M ^{p-stb }_{\delta \in (0,\bar{\delta} )}$ if $|X |=2$ 
		\item $\hat{M} \ne\emptyset$ if for all $x\in  X:$ $(\kappa\cdot  n_x)\in 2\mathbb{N}$.
	\end{enumerate}
	where $\hat{M} =M ^{p-srt:conn }\cap M ^{no-slack }$ and threshold bounds $\ubar{\delta} ,\bar{\delta} $ are defined as:
	\begin{eqnarray*}	\ubar{\delta} &\le &\min _{x ,\tilde{x} \in X }\left(\tfrac{\hat{Z}_{x ,\tilde{x} }}{\hat{Z}_{x ,\tilde{x} }+\frac{1}{2}n_x n_{\tilde{x} }}\right),\qquad 
		\hat{Z}_{x ,\tilde{x} }=\tfrac{Z (x ,x )+Z (\tilde{x} ,\tilde{x} )}{2Z (x ,\tilde{x} )}-1
		\\
		\bar{\delta} &\ge &\min _{x ,\tilde{x} \in X }\left(
		\tfrac{\hat{Z}_{x ,\tilde{x} }}{\hat{Z}_{x ,\tilde{x} }+\max (n_x ,n_{\tilde{x} })-|n_{x }-n_{\tilde{x} }|\cdot \hat{z}_{x ,\tilde{x} }}\right),\qquad \hat{z}_{x ,\tilde{x} }=\frac{z(\arg\min _{x ,\tilde{x} }n_{x },\,\arg\max _{x ,\tilde{x} }n_{x })}{Z (x ,\tilde{x} )}.
	\end{eqnarray*}

		\smallskip

		\begin{proof}
			We begin showing properties (i) and (ii). Suppose $\mu\in M ^{p-srt:conn }\cap M ^{no-slack }$ - we will demonstrate there are thresholds on $\delta $ such that $\mu$ pairwise stability and efficiency. We're only interested in the minimal thresholds such that for all values of externalities below those then stability and efficiency holds. Thus it is sufficient to evaluate the deviations from the network where the net gains are highest. 
			
			The losses from breaking a link $ij\in \mu$ can be shown to have bounded from below such that: $\ge \delta \cdot (1-Z (x ,x ))$. Suppose that $n_x =\kappa+1,x \in X $ then $\{ij\in \mu: x_i =x , x_{i' }=x  \}$ is a clique (i.e. any $i ,i' $ of type $x $ are linked). This entails that $p_{i i' }(\mu \backslash\{i i' \})=2$ and thus $p_{i i' }(\mu \backslash\{i i' \})<\in fty$.
			Suppose instead that $n_x <\kappa+1,x \in X $ then by Lemma \ref{claim:sametype:noncritical:link} there exists some $i,j$, both of type $x $ such that $p_{i i' }(\mu \backslash\{i i' \})<\in fty$.
			Thus when evaluating losses at the threshold we can assume that when deleting some link $ij$ that $i,j$ are connected in $\mu \backslash\{i j \}$.
			Although the length of the shortest paths may increase, there will still be an indirect connection and therefore no loss of utility for anyone but the two agents who lose their link.  
			Therefore we assume throughout that when evaluating thresholds if $i i' $ is deleted in $M ^{p-srt:conn }\cap M ^{no-slack }$ then only agents $i ,i' $, who must be of same type, will each lose $(1-\delta )\cdot  z (x ,x )$ while no other agents of incur a loss.
			 			
			Suppose two agents $i ,j $ of distinct types respectively $x ,\tilde{x} $ deviate by forming a link and delete a link each from $\mu $. The total loss for $i $ and $j $ for deleting a link each is:
			\begin{eqnarray*}
			(1-\delta )\cdot [z (x ,x )+z (\tilde{x} ,\tilde{x} )]				=  (1-\delta )\cdot (\hat{Z}_{x ,\tilde{x} }+1)\cdot  Z (x ,\tilde{x} ).
			\end{eqnarray*}
			
			The benefits gained for agent $i $ for establishing a link to $j $ is $[1+(n_{\tilde{x} }-1)\cdot (1-\delta )]\cdot  z (x ,\tilde{x} )$.
			Thus the total benefits gained for $i $ and $j $  from pairwise deviation can be bounded as follows.
			\begin{eqnarray*}
				&  & [1+(n_x -1)\cdot \delta ]\cdot  z (x ,\tilde{x} )+[1+(n_{\tilde{x} }-1)\cdot \delta ]\cdot  z (\tilde{x} ,x ),\\
				& = & \left\langle 1+\left[\max (n_{x },n_{\tilde{x} })-|n_{x }-n_{\tilde{x} }|\cdot \hat{z}_{x ,\tilde{x} }-1\right]\cdot \delta\right\rangle \cdot  Z (x ,\tilde{x} ).
			\end{eqnarray*}
			
			where $\hat{z}_{x ,\tilde{x} }=\frac{z(\arg\min _{x ,\tilde{x} }n_{x },\,\arg\max _{x ,\tilde{x} }n_{x })}{Z (x ,\tilde{x} )}$.
						
			We can derive the threshold for pairwise stability:
			\begin{eqnarray}
			(1-\delta )\cdot (\hat{Z}_{x ,\tilde{x} }+1)\cdot  Z (x ,\tilde{x} ) & = & 
			\left\langle 1+\left[\max (n_{x },n_{\tilde{x} })-|n_{x }-n_{\tilde{x} }|\cdot \hat{z}_{x ,\tilde{x} }-1\right]\cdot \delta\right\rangle
			\cdot  Z (x ,\tilde{x} ),\nonumber
			\\	
			(1-\delta )\cdot (\hat{Z}_{x ,\tilde{x} }+1) & = & 	\left\langle 1+\left[\max (n_{x },n_{\tilde{x} })-|n_{x }-n_{\tilde{x} }|\cdot \hat{z}_{x ,\tilde{x} }-1\right]\cdot \delta\right\rangle
			,\nonumber
			\\
			\hat{Z}_{x ,\tilde{x} } & = & 	\left[\max (n_{x },n_{\tilde{x} })-|n_{x }-n_{\tilde{x} }|\cdot \hat{z}_{x ,\tilde{x} }+\hat{Z}_{x ,\tilde{x} }\right]\cdot \delta,\nonumber
			\\
			 \delta &=&\frac{\hat{Z}_{x ,\tilde{x} }}{\max (n_{x },n_{\tilde{x} })-|n_{x }-n_{\tilde{x} }|\cdot \hat{z}_{x ,\tilde{x} }+\hat{Z}_{x ,\tilde{x} }}.\label{eq:hyperbolic_pairwise_condition}
			\end{eqnarray}
			
			Thus we can establish a lower bound for $\bar{\delta} $ (i.e. the upper bound in $\delta $ for pairwise stability of $\mu $) by taking the minimum of left-hand-side in Equation \ref{eq:hyperbolic_pairwise_condition}; thus it follows that:
			$\bar{\delta} \ge \min _{x ,\tilde{x} \in X }\left(
			\frac{\hat{Z}_{x ,\tilde{x} }}{\max (n_x ,n_{\tilde{x} })+\hat{Z}_{x ,\tilde{x} }}\right)$.

			We can do a similar procedure for finding $\ubar{\delta} $ for when the aggregate gains of establishing links across types is zero. We will evaluate the bridged network $\tilde{\mu} =\mu \cup\{i j ,i' j' \} \backslash \{i i' ,j j' \}$ where $x_{i' }=x_{i }$ and $x_{j' }=x_{j }$. The total loss $U (\mu )-U (\mu  \backslash \{i i' ,j j' \})$ is equal to double that of $i ,j $ suffers, i.e. 	$2(1-\delta )\cdot \hat{Z}_{x ,\tilde{x} }\cdot  Z (x ,\tilde{x} )$.
			
			The gains from connecting are two links of value $Z (x ,\tilde{x} )$
			as well as $n_x \cdot n_x -2$ indirect connections (between all agents of type $x $ and $\tilde{x} $)
			of value $Z (x ,\tilde{x} )\cdot \delta $. It follows that the total gains in benefits are $Z (x ,\tilde{x} )\cdot [2+(n_x \cdot n_{\tilde{x} }-2)\cdot \delta ]$.
			The threshold $\underline{\delta }$ can be found from finding when
			gains equal losses:
			\begin{eqnarray}
			2(1-\delta )\cdot (\hat{Z}_{x ,\tilde{x} }+1)\cdot  Z (x ,\tilde{x} ) & = & 2Z (x ,\tilde{x} )+(n_x \cdot n_{\tilde{x} }-2)\cdot  Z (x ,\tilde{x} )\cdot \delta \nonumber,\\
			(1-\delta )\cdot (\hat{Z}_{x ,\tilde{x} }+1) & = & 1+(\tfrac{1}{2}n_x \cdot n_{\tilde{x} }-1)\cdot \delta \nonumber,\\
			\frac{\hat{Z}_{x ,\tilde{x} }}{\hat{Z}_{x ,\tilde{x} }+\frac{1}{2}n_x \cdot n_{\tilde{x} }} & = & \delta . \label{eq:hyperbolic_efficiency}
			\end{eqnarray}

			It follows that the threshold $\ubar{\delta} $ (which is such that every  $\delta $ above implies there exists an efficient deviation from $\mu $) must be at least the minimum of left-hand-side in Equation \ref{eq:hyperbolic_efficiency} over possible types, i.e. it must hold that $\ubar{\delta} \le \min _{x ,\tilde{x} \in X }\left(\frac{\hat{Z}_{x ,\tilde{x} }}{\hat{Z}_{x ,\tilde{x} }+\frac{1}{2}n_x \cdot n_x }\right)$. This terminates the proof for part (i) of the theorem.

We move on to property (iii) of the proof where we establish characterization of pairwise stable networks. We need to show that $\hat{M} =M ^{p-stb }_{\delta \le \bar{\delta} }$ if $|X |=2$. From property (ii) $\hat{M} \subseteq M ^{p-stb }_{\delta \le \bar{\delta} }$ thus it remains to show: that $M ^{p-stb }_{\delta \le \bar{\delta} }\subseteq \hat{M} $. As \(|X|=2\) it holds that \(X=\{\underline{x},\overline{x}\}\). Define \(\Upsilon\):
\[ \Upsilon = (1-\delta)\cdot  [z(\bar{x},\bar{x})+z(\ubar{x},\ubar{x})] - [1+(n_{\ubar{x}}-1)\cdot \delta]\cdot  z(\ubar{x},\bar{x}) - [1+(n_{\bar{x}}-1)\cdot \delta]\cdot  z(\bar{x},\ubar{x}).\]

As $\delta<\ubar{\delta}$ it follows from Equation \ref{eq:hyperbolic_pairwise_condition} that: 
\begin{eqnarray}\Upsilon>0.\label{eq:min:netgains}\end{eqnarray}

Suppose \(\mu\notin M^{p-sort}\). Define a sequence of agent pairs, \(i_0j_0,i_1j_1,...\) as follows. Let agents \(i_0,j_0\in  N\) be such that \(x_i\ne x_j\) and \(ij\in  \mu\); such $i_0,j_0$ must exist as $\mu\notin M^{p-sort}$. Without loss
of generality let \(x_{i_0}=x\) and \(x_{j_0}=\tilde{x}\) where \(x,\tilde{x}\in  X\), and assume that:
\begin{eqnarray}
-\tau_{i_0j_0} & > & (1-\delta)\cdot  z(\tilde{x},\tilde{x})-[1+(n_{\tilde{x}}-1)\cdot \delta]\cdot  z(\tilde{x},x).\label{eq:hyperbolic_pairwise_deviation_first}
\end{eqnarray}

The above inequality must hold for either type $x $ or $\tilde{x} $ as 
we substitute labels for \(i,j\) as well as \(x,\tilde{x}\) due to \(\Upsilon>0\).

Let \(l\in \mathbb{N}\). It is assumed that for any
\(q<l:x_{i_q}=x, x_{j_q}=\tilde{x}\) . Also assume an associated set collection of links,
\(\mu_{l-1}\subset \mu\), and let the set be defined as
\(\mu_{l-1}=\cup_{q=0}^{l-1}\{i_qj_q\}\) such that for each \(q<l\):
\(i_qi_{q-1}\notin\mu\) if \(q\) is odd and \(j_qj_{q-1}\notin\mu\) if
\(q\) is even. At step \(q\in \mathbb{N}\) let \(\iota_q=i_{q-1}\) if \(q\) is even else
denote \(\iota_q=j_{l-1}\). Also let
\(\eta_q\in \{i_{q-1},j_{q-1}\}: \eta_q\ne\iota_q\). This entails that
\(\iota_1=i_0\) and \(\eta_1=j_0\).

Define $x_q=x,\tilde{x}_q=\tilde{x}$ if $q$ is even else vice versa. Also let $N_q=\{\iota\in  N: x_\iota=x_q\}$.

Suppose that at every $q\in \mathbb{N}:q<l$ it holds that $\iota'_q\notin\nu_{\iota_q}(\mu),\,x_{\iota'_l}=x_l$ and let $\eta'_l\in \nu_{\iota'_l}(\mu)$. Finally also define at every \(q<l\) the move
\(\Delta\mu_q=\mu\cup\{\iota_q\iota'_q\}\backslash\{\iota_q\eta_q,\iota'_q\eta'_q\}\) and let:
\begin{eqnarray}
\Delta U_{q }&=&
u_{\iota_{q}}(\Delta\mu_{q})-u_{\iota_{q}}(\mu) + u_{\iota'_{q}}(\Delta\mu_{q})-u_{\iota'_{q}}(\mu)\label{eq:aggregate:utility:moves}
\\
\Delta \hat{U}_{q }&=&			
u_{i_q}(\Delta\mu_{q+\mathbf{1}_{q:even}})-u_{i_{q}}(\mu)+
u_{j_q}(\Delta\mu_{q+\mathbf{1}_{q:odd}}) -u_{j_{q}}(\mu)\label{eq:aggregate:utility:change}
\end{eqnarray}

Note that \(\Delta\mu_q=\{i_{q}i_{q-1}\}\cup\Delta\tilde{\mu}_q\) if
\(q\) is even and
\(\Delta\mu_q=\{j_{q}j_{q-1}\}\cup\Delta\tilde{\mu}_q\) if \(q\) is
even;
\(\Delta\tilde{\mu}_q=\mu\cup\backslash\{i_{q-1}j_{q-1},i_{q}j_{q}\}\).
By inserting \(i,j\) for \(\iota,\eta\) we yield the following
expression:
\begin{eqnarray}
\sum_{q=l'}^{l-1}\Delta U_{q }=\sum_{q=l'}^{l-2}\Delta \hat{U}_{q }+ u_{\iota'_{l-1}}(\Delta\mu_{l-1})-u_{\iota'_{l-1}}(\mu) + u_{\iota_{l'}}(\Delta\mu_{l'})-u_{\iota_{l'}}(\mu)\label{eq:aggrgate:transform}
\end{eqnarray}

Assume that for every $q\in \mathbb{N}$ where $q<l$: 
\begin{eqnarray}
|\{\iota\in  N: x_\iota=x_q\,\wedge\,p_{\iota\iota_q}(\mu)<\in fty\,\wedge\,p_{\iota\iota_q}(\mu\cup\{\iota_q\iota'_q\}\backslash\{\eta_q\iota_q\})=\in fty\}|&=&0
\label{eq:connected:new:partner}
\\
|\{\iota\in  N: x_\iota=x_q\,\wedge\,p_{\iota\iota'_q}(\mu)<\in fty\,\wedge\,p_{\iota\iota'_q}(\Delta\mu_q)=\in fty\}|&=&0\label{eq:connected:new:partnerspartner}
\end{eqnarray}

Suppose Equation \ref{eq:connected:new:partner} is satisfied. It follows that net gains of benefits for $\iota_q$ from deleting the link with $\eta_q$ while forming a link together with $\iota'_q$ can be bounded: the upper bound on losses is when a connection is lost to all agents of type $\tilde{x}_q$: $[1+(n_{q}-1)\cdot \delta]\cdot  z(x_q,x_q)$; the lower bound on gains is $(1-\delta)\cdot  z(x_q,\tilde{x}_q)$ as the distance between $\iota_q\iota'_q$ is shortened to 1.
\begin{eqnarray}
u_{\iota_q}(\mu\cup\{\iota_q\iota'_q\}\backslash\{\iota_q\eta_q\})-u_{\iota_q}(\mu) & \ge  & (1-\delta)\cdot  z(x_q,x_q)-[1+(n_{q}-1)\cdot \delta]\cdot  z(x_q,\tilde{x}_q)\label{eq:newpartner:bounded:self:single}
\end{eqnarray}

Suppose Equations \ref{eq:connected:new:partner} and \ref{eq:connected:new:partnerspartner} - we can demonstrate that Equation \ref{eq:connected:new:partnerspartner} also where we replace $\iota'_q$ with $\iota_q$. If $p_{\iota_q\iota'_q}(\mu)<\in fty$ then as it also holds that $p_{\iota_q\iota'_q}(\Delta\mu_q)<\in fty$ it follows that $p_{\iota'_q\iota''_q}(\mu)<\in fty$ and $p_{\iota'_q\iota''_q}(\Delta\mu_q)=\in fty$ which violates Equation \ref{eq:connected:new:partnerspartner}. Thus it must be that $p_{\iota_q\iota'_q}(\mu)=\in fty$. Suppose instead $p_{\iota_q\iota'_q}(\mu)=\in fty$. If $\exists \iota''_q\in  N: p_{\iota_q\iota''_q}(\mu\cup\{\iota_q\iota'_q\}\backslash\{\iota_q\eta_q\})<\in fty\,\wedge\,p_{\iota_q\iota''_q}(\Delta\mu_q)=\in fty$ then it must be that $p_{\eta'_q\iota''_q}(\mu)<\in fty$ and thus $p_{\iota'_q\iota''_q}(\mu)<\in fty$ which implies that $p_{\iota_q\iota''_q}(\mu)=\in fty$. However, this is a violation of $p_{\iota_q\iota''_q}(\mu)<\in fty$.
\begin{eqnarray}
|\{\iota\in  N: x_\iota=x_q\,\wedge\,p_{\iota\iota_q}(\mu)<\in fty\,\wedge\,p_{\iota\iota_q}(\Delta\mu_q)=\in fty\}|&=&0\label{eq:connected:new:partnerspartnerself}
\end{eqnarray}

Analogue to the derivation of Inequality \ref{eq:newpartner:bounded:self:single} the net gains are bounded when Equations \ref{eq:connected:new:partner} and \ref{eq:connected:new:partnerspartner} are satisfied:
\begin{eqnarray}
\min _{\iota\in \{\iota_q,\iota'_q\}}[u_{\iota}(\Delta\mu_q)-u_{\iota}(\mu)] & \ge  & (1-\delta)\cdot  z(x_q,x_q)-[1+(n_{q}-1)\cdot \delta]\cdot  z(x_q,\tilde{x}_q), \label{eq:newpartner:bounded:counter:sametype}
\end{eqnarray}

One implication of Inequalities \ref{eq:hyperbolic_pairwise_deviation_first} and \ref{eq:newpartner:bounded:counter:sametype} if $l>1$:
\begin{eqnarray}
u_{\iota_1}(\Delta\mu_1)-u_{\iota_1}(\mu)-\tau_{\iota_1\eta_1}&\ge &(1-\delta)\cdot  z(x,x)-[1+(n_{x}-1)\cdot \delta]\cdot  z(x,\tilde{x})-\tau_{i_0j_0}\nonumber\\
u_{\iota_1}(\Delta\mu_1)-u_{\iota_1}(\mu)-\tau_{\iota_1\eta_1}&\ge &\Upsilon\label{eq:newpartner:bounded:first}
\end{eqnarray}

Another implication of Inequality \ref{eq:newpartner:bounded:counter:sametype} is that:
\begin{eqnarray}
u_{i_q}(\Delta\mu_{q+\mathbf{1}_{odd}(q)})-u_{i_q}(\mu)+u_{j_q}(\Delta\mu_{q+\mathbf{1}_{even}(q)})-u_{j_q}(\mu) & \ge  & \Upsilon,\qquad \forall  q\in [[1,l-1]]\label{eq:newpartner:bounded:counter:difftype}
\end{eqnarray}

In order for $\Delta\mu_q$ not to be a profitable pairwise deviation it must hold that:
\begin{eqnarray*}
u_{\iota_q}(\mu)+u_{\iota'_q}(\mu)+\tau_{\iota'_q\eta'_q}+\tau_{\iota_q\eta_q} & \ge  & u_{\iota_q}(\Delta\mu_q)+u_{\iota'_q}(\Delta\mu_q)\\
\tau_{\iota'_q\eta'_q} & \ge  & \Delta U_{q } + \tau_{\eta_q\iota_q}
\end{eqnarray*}

We can rewrite the above inequality using that \(\iota'_{q-1}=\eta_{q}, \eta'_{q-1}=\iota_q\) and thus
\(\tau_{\iota'_{q-1}\eta'_{q-1}}=\tau_{\eta_q\iota_q}\). We also substitute in Equation \ref{eq:aggregate:utility:moves} and assume the above inequality holds for any $q <l $:
\begin{eqnarray}			
\tau_{\iota'_{l -1}\eta'_{l -1}} & \ge  & \Delta U_{l -1} + \tau_{\iota'_{l -2}\eta'_{l -2}}\nonumber
\\ 
\tau_{\iota'_{l -1}\eta'_{l -1}} & \ge  & \sum_{q =l '}^{l -1}\Delta U_{q } + \tau_{\iota'_{l '-1}\eta'_{l '-1}}\label{eq:pairwise:sorted:nodeviation}
\end{eqnarray}

As	$\tau_{\eta_{l }\iota_{l }}=\tau_{\iota'_{l -1}\eta'_{l -1}}$ and $-\tau_{\iota_{l }\eta_{l }}=\tau_{\eta_{l }\iota_{l }}$ it follows that using Equation \ref{eq:aggrgate:transform}:
\begin{eqnarray}			
-\tau_{\iota_{l }\eta_{l }} & \ge  & \sum_{q =l '}^{l -1}\Delta U_{q } + \tau_{\iota'_{l '-1}\eta'_{l '-1}}\nonumber\\
&=& \sum_{q =1}^{l -2}\Delta \hat{U}_{q } + u_{\iota'_{l-1}}(\Delta\mu_{l-1})-u_{\iota'_{l-1}}(\mu) + u_{\iota_{1}}(\Delta\mu_{1})-u_{\iota_{1}}(\mu) + \tau_{\iota'_{0}\eta'_{0}}\nonumber\\
&=& \sum_{q =1}^{l -2}\Delta \hat{U}_{q } + u_{\eta_{l}}(\Delta\mu_{l-1})-u_{\eta_{l}}(\mu) + u_{\iota_{1}}(\Delta\mu_{1})-u_{\iota_{1}}(\mu) - \tau_{\iota_{1}\eta_{1}}
\label{eq:bound:transfers:aggregate}
\end{eqnarray}

Define the set of partners for $\iota_l$:
\begin{eqnarray}
\hat{N}_l^* (\iota_l, \mu_{l-1}) &=& \{\iota\in  N\backslash\{\iota_l\}: x_{\iota}=x_{\iota_l},\, p_{\iota\iota_l}(\mu)<\in fty,\, p_{\iota\iota_l}(\mu\backslash\{\iota_l\eta_l\})=\in fty\}\nonumber\\			
\hat{N}_l^{**}(\iota_l, \mu_{l-1})&=&\{\iota\in  N\backslash\{\iota_l\}: x_{\iota}=x_{\iota_l},\, \iota\iota_l\notin\mu\}\nonumber\\			
\hat{N}_l(\iota_l, \mu_{l-1})&=&\begin{cases} \hat{N}_l^*(\iota_l, \mu_{l-1}) \mbox{ if } \hat{N}_l^*(\iota_l, \mu_{l-1})\ne\emptyset,\\ N_l^{**}(\iota_l, \mu_{l-1}) \mbox{ else. } \end{cases}
\end{eqnarray}

A property of \(\hat{N}_l\) is that \(\hat{N}_l\ne\emptyset\); this follows as
\(\min _{\hat{x}\in  X}n_{\hat{x}}\ge \kappa+1\). Let $\iota'_l\in  \hat{N}_l$ which implies that  Equation \ref{eq:connected:new:partner} holds.

Suppose that \(k_{\iota'_l}(\mu)<\kappa\). 
As Equation \ref{eq:connected:new:partner} holds it follows that \[u_{\iota'_q}(\mu\cup\{\iota_q\iota'_q\}\backslash\{\iota_q\eta_q\})-u_{\iota'_q}(\mu)\ge (1-\delta)\cdot  z(x_q,x_q),\] and thus $u_{\iota'_q}(\mu\cup\{\iota_q\iota'_q\}\backslash\{\iota_q\eta_q\})>0$. 

We can also derive utility bounds using Inequality \ref{eq:bound:transfers:aggregate} along with Inequalities \ref{eq:min:netgains}, \ref{eq:newpartner:bounded:first}:
\begin{eqnarray}
&&u_{\iota_l}(\mu\cup\{\iota_l\iota'_l\}\backslash\{\iota_l\eta_l\})-u_{\iota_l}(\mu)-\tau_{\iota_l\eta_l}\nonumber,\\
&\ge & u_{\iota_l}(\mu\cup\{\iota_l\iota'_l\}\backslash\{\iota_l\eta_l\})-u_{\iota_l}(\mu)+ u_{\eta_{l}}(\Delta\mu_{l-1})-u_{\eta_{l}}(\mu) +  \sum_{q =1}^{l -2}\Delta \hat{U}_{q } + u_{\iota_{1}}(\Delta\mu_{1})-u_{\iota_{1}}(\mu) - \tau_{\iota_{1}\eta_{1}}\nonumber,\\
&\ge & l\cdot \Upsilon\nonumber,\\&>&0,\nonumber
\end{eqnarray}
thus $\iota_l,\iota'_l$ can profitably from deviate pairwise. Thus it must be that \(k_{\iota'_l}(\mu)=\kappa\). 

Suppose there exists \(\iota'_l\in  N_l\backslash\nu_i(\mu),\iota''_l\in  N_l\backslash\{\iota_l,\iota'_l	\}\) such that
\(\iota'_l\iota''_l\in \mu\), \(p_{\iota'_l\iota''_l}(\mu\backslash\{\iota'_l\iota''_l\})<\in fty\) and \(\tau_{\iota'_l\iota''_l}\le 0\).
This entails that \(u_{\iota'_l}(\Delta\hat{\mu}_l)-u_{\iota'_l}(\mu)\ge 0\)
where
\(\Delta\mu_l=\mu\cup\{\iota_l\iota'_l\}\backslash\{\iota_l\eta_l,\iota'_l\iota''_l\}\).
This follows from
\(u_{\iota'_l}(\Delta\mu_l)-u_{\iota'_l}(\mu)=u_{\iota'_l}(\Delta\mu_l)-u_{\iota'_l}(\mu\cap \Delta\mu_l)-[u_{\iota'_l}(\mu\cap \Delta\mu_l)-u_{\iota'_l}(\mu)]\)
and
\(u_{\iota'_l}(\Delta\mu_l)-u_{\iota'_l}(\mu\cap \Delta\mu_l)\ge 1-z(x,x)\)
and \(u_{\iota'_l}(\mu\cap \Delta\mu_l)-u_{\iota'_l}(\mu)=1-z(x,x)\). As \(\tau_{\iota'_l\iota''_l}\le  0\) it follows that
that utility for $\iota'_l$ is:
\begin{eqnarray}
u_{\iota'_l}(\Delta\hat{\mu}_l)-u_{\iota'_l}(\mu)-\tau_{\iota'_l\iota''_l}\ge 0.\nonumber
\end{eqnarray}

And utility for $\iota_l$ can bounded be as follows using Inequality \ref{eq:newpartner:bounded:self:single} for $u_{\iota_l}(\Delta\hat{\mu}_l)-u_{\iota_l}(\mu)$ as Equation \ref{eq:connected:new:partner} holds :
\begin{eqnarray}			
&& u_{\iota_l}(\Delta\hat{\mu}_l)-u_{\iota_l}(\mu)-\tau_{\iota_l\eta_l} \nonumber\\ 
& = & u_{\iota_l}(\Delta\hat{\mu}_l)-u_{\iota_l}(\mu)+\tau_{\eta_l\iota_l} \nonumber\\
&\ge & \sum_{q=1}^{l-1}\Delta U_q + u_{\iota_l}(\Delta\hat{\mu}_l)-u_{\iota_l}(\mu) + \tau_{j_0i_0}\nonumber\\
&=& \sum_{q=1}^{l-2}\Delta \hat{U}_q + u_{\iota_l}(\Delta\hat{\mu}_l)-u_{\iota_l}(\mu) + u_{\eta_l}(\Delta\mu_{l-1})-u_{\eta_l}(\mu) + u_{i_0}(\Delta\mu_{1})-u_{i_0}(\mu) -\tau_{i_0j_0}\nonumber\\
&\ge & l\cdot \Upsilon \nonumber\\
&>&0
\end{eqnarray}

The above inequalities entails that $\iota_l,\iota'_l$ can deviate profitably pairwise; this is a violation of pairwise stability and thus cannot be true. 
Thus there exists no \(\iota'_l\iota''_l\in \mu\) such that \(\iota'_l\in  N_l\backslash\nu_i(\mu),\iota''_l\in  N_l\backslash\{\iota_l,\iota'_l	\}\) as well as
\(p_{\iota'_l\iota''_l}(\mu\backslash\{\iota'_l\iota''_l\})<\in fty\) and \(\tau_{\iota'_l\iota''_l}\le 0\).

Suppose that
\(\forall  \iota'_l\in  N_l: \ne xists\eta'_l\in \nu_{\iota'_l}(\mu\backslash\mu_{l-1}):x_{\eta'_l}\ne x_l\).
This entails that
\(\forall  \iota'_l\in  N_l: \ne xists\eta'_l\in \nu_{\iota'_l}(\mu):x_{\eta'_l}\ne x_l\)
as \(k_{\iota'_l}(\mu\backslash\mu_{l-1})=k_{\iota'_l}(\mu)\). By Lemma
\ref{claim:sametype:noncritical:link} it follows there exists
\(\iota'_l,\iota''_l\in  N_l\backslash\nu_i(\mu)\) such that
\(p_{\iota'_l\iota''_l}(\mu\backslash\{\iota'_l\iota''_l\})<\in fty\),
\(\iota'_l\iota''_l\in \mu\) and \(\tau_{\iota'_l\iota''_l}\le 0\) which by the arguments above cannot be true. Therefore there has to exist some 
$\iota'_l\in  N_l$ for which there is an agent $\eta'_l\in \nu_{\iota'_l}(\mu\backslash\mu_{l-1})$ where it holds that $x_{\eta'_l}\ne x_l$.

A duplicate occurs if
\(i_{l-1},j_{l-1}\in \mu_{l-2}\). That is for some \(l'<l\) it holds that
either \(\iota_{l},\eta_{l}=\iota_{l'},\eta_{l'}\) if \(l-l'\) is even
or \(\iota_{l},\eta_{l}=\eta_{l'},\iota_{l'}\) if \(l-l'\) is odd.

If \(l-l'\) is odd, then
\(\tau_{\iota'_{l'-1}\eta'_{l'-1}}=-\tau_{\iota'_{l-1}\eta'_{l-1}}\) and
therefore we can reduce the Inequality \ref{eq:pairwise:sorted:nodeviation}:
\begin{eqnarray*}			
0 & \ge  & \sum_{q=l'}^{l-1}[u_{\iota_{q}}(\Delta\mu_{q}) - u_{\iota_{q}}(\mu)+u_{\iota'_{q}}(\Delta\mu_{q}) - u_{\iota'_{q}}(\mu)] + 2\tau_{\iota'_{l'-1}\eta'_{l'-1}}
\\& = & \sum_{q =l '}^{l -2}\Delta \hat{U}_{q } + u_{\iota'_{l-1}}(\Delta\mu_{l-1})-u_{\iota'_{l-1}}(\mu) + u_{\iota_{l'}}(\Delta\mu_{l'})-u_{\iota_{l'}}(\mu)+2\tau_{\iota'_{l'-1}\eta'_{l'-1}}
\\& = & \sum_{q =l '}^{l -2}\Delta \hat{U}_{q } + 2\cdot  \left\langle u_{\eta'_{l'-1}}(\Delta\mu_{l'})-u_{\eta'_{l'-1}}(\mu) + \tau_{\iota'_{l'-1}\eta'_{l'-1}}\right\rangle
\\& = & \sum_{q =l '}^{l -2}\Delta \hat{U}_{q } + 2\cdot  \left\langle u_{\eta'_{l'-1}}(\Delta\mu_{l'})-u_{\eta'_{l'-1}}(\mu) + \sum_{q=1}^{l'-1}\Delta U_{q } + \tau_{\iota'_{0}\eta'_{0}}\right\rangle
\\& = & \sum_{q =l '}^{l -2}\Delta \hat{U}_{q } + 2 \cdot  \sum_{q =1}^{l '-1}\Delta \hat{U}_{q }+2\cdot \left[u_{\iota_{1}}(\Delta\mu_{1})-u_{\iota_{1}}(\mu)- \tau_{\iota_{1}\eta_{1}}\right]
\\& \ge  & (l+l')\cdot \Upsilon\\&>&0,
\end{eqnarray*}

thus there must be a feasible pairwise deviation for $\iota_q,\iota'_q$ where $q\in [[1,l-1]]$.

If \(l-l'\) is even then
\(\tau_{\iota'_{l-1}\eta'_{l-1}}=\tau_{\iota'_{l'-1}\eta_{l'-1}}\); thus
Inequality \ref{eq:pairwise:sorted:nodeviation} for no pairwise deviation becomes: $0\ge 	\sum_{q=l'}^{l-1}\Delta U_{q}$. This can in turn be rewritten as follows:
\begin{eqnarray*}
0&\ge &	\sum_{q=l'}^{l-2}\Delta \hat{U}_{q}+ u_{\iota'_{l-1}}(\Delta\mu_{l-1})-u_{\iota'_{l-1}}(\mu) + u_{\iota_{l'}}(\Delta\mu_{l'})-u_{\iota_{l'}}(\mu)
\end{eqnarray*}

Using that \(\iota_{q}=\eta'_{q-1}\) and \(\eta'_{l'-1}=\eta'_{l-1}\) we get:
$0\ge \sum_{q=l'}^{l-1}\Delta \hat{U}_{q}.$ Recall that for all \(q\in \mathbb{N}: q<l\) it holds that
\(\Delta \hat{U}_{q}\ge \Upsilon\) where \(\Upsilon>0\). Thus there must be a feasible pairwise deviation.

Due to \(\iota'_l,\eta'_l=\eta_{l+1},\iota_{l+1}\) it follows that it
cannot be that \(\iota'_l\eta'_l\in \mu_{l-1}\) as otherwise
\(i_lj_l\in \mu_{l-1}\). This entails
\(\ne xists \iota'_l\in  N_l: \exists \eta'_l\in \nu_{\iota'_l}(\mu_{l-1})\).
Therefore we can assume
\(\forall  \iota'_l\in  N_l: \ne xists\eta'_l\in \nu_{\iota'_l}(\mu_{l-1})\)
and thus
\(\forall  \iota'_l\in  N_l: k_{\iota'_l}(\mu\backslash\mu_{l-1})=k_{\iota'_l}(\mu)\).

Suppose that Equation \ref{eq:connected:new:partnerspartner} is violated for any $\iota'_l\in  \hat{N}_l$. This is equivalent to it holds for any $\iota'_l\in  \hat{N}_l$ where $\eta'_l\in \nu_{\iota'}(\mu)$ that there is some other $\iota''_l\in \hat{N}_l$ such that $p_{\iota'_l\iota''_l}(\Delta\mu_l)=\in fty$. Let $\iota^{(1)}_l=\iota'_l$. As Equation \ref{eq:connected:new:partnerspartner} must hold for any $\iota'_l\in  \hat{N}_l$  we can reproduce the argument iteratively and thus for $\iota^{(q)}_l\in  \hat{N}_l,q\in \mathbb{N}$ there is some $\eta^{(q)}_l\in \nu_{\iota^{(q)}_l}(\mu)$ such that for some $\iota^{(q+1)}_l\in  \hat{N}_l\backslash\{\iota^{(1)}_l,..,\iota^{(q)}_l\}$ it holds that $p_{\iota^{(1)}_l\iota^{(q+1)}_l}(\Delta\mu_l)=\in fty$. However, as $n<\in fty$ it follows that there for some $q\in \mathbb{N}$ that $N_l\backslash\{\iota^{(1)}_l,..,\iota^{(q)}_l\}=\emptyset$. Thus let instead $\iota'_l=\iota^{(q)}_l$; for any $\eta'_l\in \nu_{\iota'_l}(\mu)$ there is no $\iota''_l\in  N_l$ such that $p_{\iota'_l\iota''_l}(\mu)=\in fty$. This contradicts that Equation \ref{eq:connected:new:partnerspartner} is violated for agent $\iota'_l=\iota^{(q)}_l$.

Suppose \(\mu\notin M^{no-surpl.}\). This would entail that
\(\exists  i\in  N: k_i(\mu)<\kappa\). As \(n_x>\kappa\) there has to
\(\exists  i'\in  N: x_{i'}=x_i, ii'\notin\mu\). Suppose that
\(k_{i'}<\kappa\) then
\(\sum_{\iota\in \{i,i'\}}[u_\iota(\mu\cup\{ii'\})-u_\iota(\mu)]>0\) and
thus \(ii'\) can be formed profitably pairwise. Moreover, as
\(k_{i'}(\mu)=\kappa\) it follows that
\(\exists  i''\in \nu_{i'}: ii''\notin\mu, x_{i''}=x_i\). By Lemma
\ref{claim:sametype:noncritical:link} it follows there exists
\(\iota,\iota'\in  \tilde{N}\backslash\nu_i(\mu)\) such that
\(p_{\iota\iota'}(\mu\backslash\{\iota\iota'\})<\in fty\),
\(\iota\iota'\in \mu\) and \(\tau_{\iota\iota'}\le 0\). This entails that
\(u_\iota(\mu)-u_\iota(\mu\backslash\{\iota\iota'\})+\tau_{\iota\iota'}\le  (1-\delta)z(x,x)\).
Moreover, as\\
\(\sum_{j\in \{i,\iota\}}[u_j(\mu\cup\{i\iota\}\backslash\{\iota\iota'\})-u_j(\mu\backslash\{\iota\iota'\})]\ge  (1-\delta)\cdot  Z(x,x)\)
it holds that:

\[\sum_{j\in \{i,\iota\}}[u_j(\mu\cup\{i\iota\}\backslash\{\iota\iota'\})-u_j(\mu)]-\tau_{\iota\iota'}\ge  (1-\delta)\cdot  z(x,x)\].

Thus \(i,\iota\) can deviate profitably pairwise which contradicts pairwise Nash
stability. Therefore it must be that \(\mu\in  M^{no-surpl.}\)

Suppose \(\mu\notin M^{p-sort+conn}\). As
\(\mu \in  M^{p-sort}\cap  M^{no-surpl.}\) there exist
\(i,i',j,j'\in  N: x_i=x_{i'}=x_j=x_{j'}\) and \(ij,i'j'\in \mu\) and
\(p_{ii'}(\mu)=\in fty\). Without loss of generality we assume that
\(\tau_{ij},\tau_{i'j'}\le 0\) (otherwise we could simply switch
identities some \(i\)'s and \(j\)'s). This entails:
\[\min _{\iota\in \{i,i'\}}[u_\iota(\mu\backslash\{ij,i'j'\})-u_\iota(\mu)]+\tau_{ij}+\tau_{i'j'}\le 2(1-\delta)\cdot  z(x,x)\]

Also we have that:
\[\min _{\iota\in \{i,i'\}}[u_\iota(\mu\cup\{ii'\}\backslash\{ij,i'j'\})-u_\iota(\mu\backslash\{ij,i'j'\})]\ge (\kappa+1)\cdot (1-\delta)\cdot  z(x,x)\]

This entails that
\(\sum_{\iota\in \{i,i'\}}[u_\iota(\mu\cup\{ii'\}\backslash\{ij,i'j'\})-u_\iota(\mu)]-\tau_{ij}-\tau_{i'j'}\ge  \kappa\cdot (1-\delta)\cdot  Z(x,x)\);
thus \(i,i'\) can deviate profitably. Thus we have shown that $M ^{p-stb }_{\delta \le \bar{\delta} }\subseteq \hat{M} $ which terminates the proof of property (iii).

Finally property (iv) follows as Fact \ref{claim:construct:network:degreequota:size} can be applied to the subset of agents associated with each type as $\forall  x\in  X: n_x>\kappa$ and $\kappa\cdot  n_x\in 2\mathbb{N}$.
\end{proof}

	\noindent
	\textbf{Proposition \ref{claim:policy_accomodate_oversorting}}:
		Suppose that conditions for suboptimal sorting from Theorem \ref{claim:suboptimal-sorting_hyperbolic_decay} are valid and there are two types then a policy maker can implement a welfare improving network when $\delta \in (\ubar{\delta} ,\bar{\delta} )$ from a sorted network $\mu \in \hat{M} $.

	\begin{proof}
		Let $\mu \in \hat{M} $ and $\delta \in (\ubar{\delta} ,\bar{\delta} )$. By construction there exists a network $\tilde{\mu} $ which has higher aggregate utility. Let the two pairs of agents $i i' ,j j' $ be agents such that $\tilde{\mu} =\mu \,\cup\,\{i j ,i' j' \}\backslash\{i i' ,j j' \}$ and $x_{i }=x_{i' }=x $ and $x_{j }=x_{j' }=x $. 
		Specify a link-contingent contract to $i ,j $ where $\hat{\mu} =\mu \cup\{i j \}\backslash\{i i' ,j j' \}$ such that:
		
		\small			
		\begin{eqnarray}
		\forall \iota \iota' \in \{i j , i' j'  \}: \qquad \,\,\, \mathcal{C}_{\iota \iota' }+\mathcal{C}_{\iota' \iota }&\in & (\tfrac{1}{2}[Z (x ,x )+Z (\tilde{x} ,\tilde{x} )-2Z (x ,\tilde{x} )],\,\,\,\tfrac{1}{2}[U (\tilde{\mu} )-U (\mu )])\label{eq:policy_value_deviators},\\
		\forall \iota \iota' \notin\{i j ,j i , i' j' ,j' i'  \}: \quad\,\,\,\,\,\, \mathcal{C}_{\iota \iota' }&=& 0. \label{eq:policy_value_others}
		\end{eqnarray}
		\normalsize
		
		By Theorem \ref{claim:suboptimal-sorting_hyperbolic_decay} we know that $\mu $ is pairwise stable. 
		Pairwise stability implies that  $\tfrac{1}{2}[Z (x ,x )+Z (\tilde{x} ,\tilde{x} )-2Z (x ,\tilde{x} )]>b_i (\mu )-b_i (\hat{\mu} )+b_{j }(\mu )-b_{j }(\hat{\mu} )$ as deviation is not profitable. Using this fact together with Inequality \ref{eq:policy_value_deviators} it follows that:
		 \[\mathcal{C}_{i j }+\mathcal{C}_{j i } > b_i (\mu )-b_i (\hat{\mu} )+b_{j }(\mu )-b_{j }(\hat{\mu} ).\] 
		 
		 The above inequality entails agents $i ,j $ are a blocking coalition that can gain by deviating to $\hat{\mu} $; this blocking move is also the only profitable move for $i ,j $ due to pairwise stability of $\mu $ and Equation \ref{eq:policy_value_others}.

		In network $\hat{\mu} $ agents $i' ,j' $ have an incentive to form a link with one another as both have surplus link capacity (i.e. degree below the quota) and forming a link is profitable from Inequality \ref{eq:policy_value_deviators}. Moreover, we show in the following that this move is the one that ensures the highest aggregate net benefits to $i' ,j' $. 
		
		We begin with showing that linking across types to other agents of type $x ,\tilde{x} $ is not profitable. Suppose $i' $ links across types to  another agent $j'' \in \{\iota \ne j' :x_{\iota }=x_{j' }\}$.
		First, note the pairwise deviation from $\mu $ to form $i' j'' $ is unprofitable (due to pairwise stability), thus it less profitable than forming $i' j' $ from $\mu $ (which is profitable by Inequality \ref{eq:policy_value_deviators}). 
		Second, the net-increase in value of the pairwise deviation to form $i' j' $ over $i' j'' $ increases from $\mu $ to $\hat{\mu} $ - this is true as $j' $ loses the link with $i $ from $\mu $ while $j'' $ has an unchanged number - thus $j' $ will have a weakly lower opportunity cost of deleting links in $\hat{\mu} $. The same argument can be applied to $j' $ for $i'' \in \{\iota \ne i' :x_{\iota }=x_{i' }\}$.
		
		We turn to showing that linking to other agents of same type (staying sorted) is not more profitable as well. Suppose $i' $ and $j' $ link to same types as themselves respectively, i.e. $i'' \in \{\iota \ne i' :x_{\iota }=x_{i' }\}$ and $j'' \in \{\iota \ne j' :x_{\iota }=x_{j' }\}$. Suppose $i i'' \in \mu $ then no feasible pairwise moves to same type can exist in $\hat{\mu} $ as the move can only involve deleting links; same is true if $j j'' \in \mu $. 
		Thus instead we use $i i'' ,j j'' \notin\mu $. 
		It must be that any pairwise deviation forming either $i i'' $ or $j j'' $ from $\mu $ is unprofitable (as $\mu $ is pairwise stable);
		this implies that for any $\iota \in \nu_{i'' }(\hat{\mu} )$ and $\iota' \in \nu_{j'' }(\hat{\mu} )$ it holds that:
		\begin{eqnarray}
		b_{i' }(\hat{\mu} \cup\{i' i'' \}\backslash\{i'' \iota \})-b_{i' }(\hat{\mu} )+b_{i'' }(\hat{\mu} \cup\{i' i'' \}\backslash\{i'' \iota \})-b_{i'' }(\hat{\mu} )-\tau_{i'' \iota }&\le & z (x ,x )\label{eq:other_deviation_i},\\
		b_{j' }(\hat{\mu} \cup\{j' j'' \}\backslash\{j'' \iota' \})-b_{j' }(\hat{\mu} )+b_{j'' }(\hat{\mu} \cup\{j' j'' \}\backslash\{j'' \iota' \})-b_{j'' }(\hat{\mu} )-\tau_{j'' \iota' }&\le & z (\tilde{x} ,\tilde{x} ).\label{eq:other_deviation_j}
		\end{eqnarray}
		
		As $b_{i' }(\tilde{\mu} )-b_{i' }(\hat{\mu} )+b_{j' }(\tilde{\mu} )-b_{j' }(\hat{\mu} )=z (x ,\tilde{x} )+z (\tilde{x} ,x )$ it follows that 
		\[b_{i' }(\tilde{\mu} )-b_{i' }(\hat{\mu} )+b_{j' }(\tilde{\mu} )-b_{j' }(\hat{\mu} )+\mathcal{C}_{i' j' }+\mathcal{C}_{j' i' }>z (x ,x )+z (\tilde{x} ,\tilde{x} ).\]

		The above inequality implies together with Inequalities \ref{eq:other_deviation_i} and \ref{eq:other_deviation_j} that the total gains for $i' $ and $j' $ exceeds the total value that could be generated from alternative deviations. Thus there are two pairwise moves from $\mu $ to $\hat{\mu} $ and from $\hat{\mu} $ to $\tilde{\mu} $ which both provide strictly higher utility to the deviating agents. 
		
		Pairwise stability follows from three arguments. First, all deviations among agents where only links in $\tilde{\mu} \cap \mu $ are deleted will provide at most the same value in $\tilde{\mu} $ that the deviations did in $\mu $ - this follows as these agents all have the same links and in $\tilde{\mu} $ all agents are connected in $\tilde{\mu} $ and thus only direct links matter. This upper limit too gains from deviations implies none of these moves can be profitable as they were unprofitable form $\mu $. Second, deviations that involve deletion of links in $\tilde{\mu} \backslash\mu $ are shown above to provide strictly higher value than any other deviations - thus deviating from $\tilde{\mu} $ must also provide strictly lower value. 
	\end{proof}

	\subsection{Sorting: externalities - infinite population}
	\noindent \textbf{Proposition \ref{claim:assort:ext:asymptoticperfectPAM}:} If there is supermodularity, a degree quota and constant decay with asymptotic independence then there is asymptotic perfect sorting for strongly stable networks.	
	\begin{proof}	
		
		Under asymptotic independence it follows that average per agent utility for type $x $ under asymptotic perfect sorting converges to (using a geometric series):
		\[
		\frac{\left(\kappa-1\right)\delta}{1-\left(\kappa-1\right)\delta}z(x,x)
		\]

		Let $\omega_{x \tilde{x} }=\kappa \cdot \mathbb{E}[\delta ^{p_{i j }} | x_i =x ,x_j =\tilde{x} ]$.
		Suppose that for two types, $x ,\tilde{x} $ there is not perfect sorting, and in particular there is some mixing between them, i.e. $\omega_{x \tilde{x} }>0$; the average per agent utility is:
		
		\[
		\left[\frac{\left(\kappa-1\right)\delta}{1-\left(\kappa-1\right)\delta}-\omega_{x}\right]\cdot  z(x,x)+\omega_{x}\cdot  z(x,\tilde{x}).
		\]

		Each agent will almost surely have $\kappa $ links as it is assumed that each link adds positive value and there are asymptotic infinite agents (only a finite number can then not fulfill the degree quota).
		
		As we have a finite set of types we can assume then for large populations there is a subset of types, $\hat{X} \subseteq  X $, where for every type $x \in \hat{X} $ it holds that there is an asymptotic strictly positive share of the total number of agents of that type, i.e., $\lim_{n \rightarrow\in fty}(|\{i \in N_n \}_{x_{i }=x }|/n )>0.$ If there is only one such type, i.e. $|\hat{X} |=1$, then asymptotic perfect sorting follows by assumption as the asymptotic number of links is $\kappa $.

		For any two types $x ,\tilde{x} \in \hat{X} $ which are mixing their average utility is:
		\begin{eqnarray*}
			\frac{\left(\kappa -1\right)\delta }{1-\left(\kappa -1\right)\delta }\left[\frac{n_{x }\cdot z (x ,x )+n_{\tilde{x} }\cdot z (\tilde{x} ,\tilde{x} )}{n_{x }+n_{\tilde{x} }}\right]-\frac{1}{2}\cdot  \left[\frac{n_{x }\cdot  \omega_{x \tilde{x} }}{n_{x }+n_{\tilde{x} }}\right]\cdot [Z (x ,x )+Z(\tilde{x} ,\tilde{x} )-2Z(x ,\tilde{x} )].
		\end{eqnarray*}
		As there is supermodularity it follows that $Z (x ,x )+Z(\tilde{x} ,\tilde{x} )-2Z(x ,\tilde{x} )>0$ and thus mixing must decrease utility. The same argument can be applied by mixing between multiple types. 
	\end{proof}

	\subsection{Network structure}\label{app:network:structure}

	\noindent
	\textbf{Proposition  \ref{claim:network:monotonicdegree:talent}:} Suppose there is monotonicity and no externalities in link value then the set of pairwise stable networks has degree monotonicity.
	
	\begin{proof}	
		Suppose the claim is false; that is, for some pairwise stable network $\mu \in M ^{p-stb }$ it holds for two agents $i $ and $j $ that $x_{i }>x_{j }$ but  $k_{i }(\mu )<k_{j }(\mu )$. The condition that $k_{i }(\mu )<k_{j }(\mu )$ entails there is another agent who is in $\nu_{j }(\mu )$ but not in $\nu_{i }(\mu )$ which entails that $\nu_{j }(\mu )\backslash(\nu_{i }(\mu )\cup\{i \})\ne\emptyset$. 
		
		Let $j' \in \nu_{j }(\mu )\backslash(\nu_{i }(\mu )\cup\{i \})$. From monotonicity of $Z $ it holds that  $Z_{i j' }>Z_{j j' }$ as $x_{i }>x_{j }$. 
		From the cost technologies (either convex or a degree quota) it must be that
		\[c_{i }(\mu \cup\{i j' \})-c_{i }(\mu )\le c_{j }(\mu )-c_{j }(\mu \backslash\{j j' \}),\] as $k_{i }(\mu )<k_{j }(\mu )$.
		Combining the restriction on benefits it follows that the value created by forming ${i }{j' }$ and deleting ${j }{j' }$ can be restricted:
		\[Z_{i j' }-
		(c_{i }(\mu \cup\{i j' \})-c_{i }(\mu ))
		>Z_{j j' }-
		(c_{j }(\mu )-c_{j }(\mu  \backslash \{j j' \})).\]
		
		As the move to $(\mu \cup{i }{j' })\backslash \{j j' \}$ is feasible (and respects the degree quota if there is one), it follows that strong stability is violated as it implies that $\mu $ is not efficient. From Lemma \ref{claim:auxiliary:equivalence:strong:pairwisetype} it follows that pairwise stability is also violated if the claim is false.\end{proof}

	\noindent
	\textbf{Proposition  \ref{claim:network:elementwise:monotonicdegree}:} If there are supermodularity, monotonicity in link value as well as complete heterogeneity but no externalities then every pairwise stable network has sorting in degree.
	
	\begin{proof}
		As there are no externalities, supermodularity and monotonicity in link value then Proposition  \ref{claim:assort:elementwise:monotonic} and \ref{claim:network:monotonicdegree:talent} must hold. In addition, as agents $i $ and $j $ are distinct under complete heterogeneity then either $x_{i }>x_{j }$ or $x_{i' }>x_{i }$. As $k_{i }(\mu )\ge k_{j }(\mu )$ it must be that $x_{i }>x_{i' }$ as the converse would violate Proposition \ref{claim:network:monotonicdegree:talent}. 
		
		From Proposition \ref{claim:assort:elementwise:monotonic} it is known that if there are two agents such that $x_{i }>x_{i' }$  then this entails that for $l =1,...,k_{j }(\mu ): 
		\mathcal{X} (\nu_{i }(\mu )/\{j \})_l \ge \mathcal{X} (\nu_{j }(\mu )/\{i \})_l $. This inequality entails there are exactly two possible cases for any index $l =1,...,k_{j }(\mu )$. 
		The first case is that $\mathcal{X} (\nu_{i }(\mu )/\{j \})_l =\mathcal{X} (\nu_{j }(\mu )/\{i \})_l $. In this case 
		$\mathcal{K} (\nu_{i }(\mu )/\{i \})_l =		    \mathcal{K} (\nu_{j }(\mu )/\{i \})_l $ as the agent linked to $i $ and $j $ must be the same due to complete heterogeneity.
		Else in the other case where  $\mathcal{X} (\nu_{i }(\mu )/\{j \})_l >
		\mathcal{X} (\nu_{j }(\mu )/\{i \})_l $ then by reapplying  Proposition \ref{claim:network:monotonicdegree:talent} it follows that, $\mathcal{K} (\nu_{i }(\mu )/\{i \})_l \ge 				\mathcal{K} (\nu_{j }(\mu )/\{i \})_l .$
	\end{proof}

	\noindent
	\textbf{Proposition  \ref{claim:monotonic_decaycentrality}:} Suppose there are externalities as well as monotonicity and no modularity then the set of strongly stable networks has $\delta $-decay monotonicity.

	\begin{proof}
		The no modularity condition implies that individual link value is independent of own talent and thus separable for any $\hat{x} \in X $:  $Z (\hat{x} ,\tilde{x} )=\tilde{Z} (\tilde{x} )+\tilde{Z} (\hat{x} )$ where $\tilde{Z} $ is the contribution to the link value for a given level of talent. Thus we can rewrite Equation \ref{eq:def:benefits}:
		\begin{eqnarray*}
			\sum_{i \in N }b_{i}(\mu )&=&
			\sum_{i \in N }\sum_{j \ne i }\delta ^{p_{i j }(\mu )-1}z (x_i ,x_j ),\\
			&=&\sum_{{i }{j }\in \mu^c }\delta ^{p_{i j }(\mu )-1}Z (x_i ,x_j ),\\
			&=&\sum_{i \in N }\sum_{j \ne i }\delta ^{p_{i j }(\mu )-1}\tilde{Z} (x_{i }),\\
			&=&\sum_{i \in N }d_{i }^{\delta }(\mu )\cdot \tilde{Z} (x_{i }).
		\end{eqnarray*}
		
		Due to monotonicity in link value it also holds that $\frac{\partial }{\partial x_{i }}Z (x_{i },x_{i' })=\frac{\partial }{\partial x_{i }}\tilde{Z} (x_{i' })>0$. This entails that a necessary condition for the sum of utilities to be maximal by some network $\mu $ is that for any two agents $i' , j' $ such that $x_{i' }>x_{j' }$ it holds that $d_{i' }^{\delta }(\mu )>d_{j }^{\delta }(\mu )$. The necessity is demonstrated in the following. 
		
		Denote an alternative network $\tilde{\mu} $ where $i' $ and $j' $ have switched positions: if $i' $ and $j' $  are not linked in $\mu $ then let  $\nu_{i }(\tilde{\mu} )=\nu_{i' }(\mu )$ and $\nu_{i' }(\tilde{\mu} )=\nu_{i }(\mu )$; else if $i' $ and $j' $  are  linked in $\mu $ then let $\nu_{i }(\tilde{\mu} )=\nu_{i' }(\mu )\cup\{i' \}/\{i \}$ and $\nu_{i' }(\tilde{\mu} )=\nu_{i }(\mu )\cup\{i \}/\{i' \}$. 
		Thus it holds that
		$d_{j }^{\delta }(\tilde{\mu} )
		=d_{i }^{\delta }(\mu )$ and
		$d_{i }^{\delta }(\tilde{\mu} )
		=d_{j }^{\delta }(\mu )$.
		A deviation from $\mu $ to $\tilde{\mu} $ is possible for the grand coalition.
		
		We will show that the alternative network $\tilde{\mu} $ will generate higher aggregate utility which violates efficiency of $\mu $ and thus also violates strong stability.
		Starting with costs there are two cases of cost technology: when there is quota in links then a deviation to the alternative network $\tilde{\mu} $ is consistent with the degree quota\footnote{Both agents have a number of links that do not exceed the degree quota in the original network - thus their degree quota cannot be exceed in the alternative network.} and has unchanged costs; when there are convex costs then the move will have no change in aggregate costs as the sum of costs for agents $i $ and $j $ is unchanged. However, the benefits will be higher under $\tilde{\mu} $, using that $i $ and $i $ switch neighborhoods - this follows as $d_{j }^{\delta }(\mu ) > d_{i }^{\delta }(\mu )$ is equivalent to:
		\begin{eqnarray*}
			d_{j }^{\delta }(\mu )\cdot 
			[\tilde{Z} (x_{i })-\tilde{Z} (x_{j })]
			& > &
			d_{i }^{\delta }(\mu )\cdot [\tilde{Z} (x_{i })-\tilde{Z} (x_{j })],
			\\
			d_{i }^{\delta }(\tilde{\mu} )\cdot \tilde{Z} (x_{i })+
			d_{j }^{\delta }(\tilde{\mu} )\cdot \tilde{Z} (x_{j })
			& > &
			d_{i }^{\delta }(\mu )\cdot \tilde{Z} (x_{i })+
			d_{j }^{\delta }(\mu )\cdot \tilde{Z} (x_{j }),
			\\
			\sum_{i \in N }b_{i}(\tilde{\mu} )&>&\sum_{i \in N }b_{i}(\mu ).
		\end{eqnarray*}
	\end{proof}

		\subsection{Limits to sorting and monotonic centrality}\label{app:undersorting}

		\begin{figure}[b!]\begin{center}
				\begin{tikzpicture}[nodes={fill=gray!12,circle, ultra thin}, node distance=1.8cm, auto]
				\draw[help lines] (0,0);
				\node[state] (i_1) {5};
				\node[state] (i_2) [below left=0.5cm and 0.5cm of i_1] {6};
				\node[state] (i_3) [below right=0.5cm and 0.5cm of i_1] {7};
				\node[state] (i_5) [right=0.5cm of i_3] {1};
				\node[state] (i_4) [above=0.5cm of i_5] {2};
				\node[state] (i_6) [right=0.5cm of i_4] {3}; 	
				\node[state] (i_7) [right=0.5cm of i_5] {4};
				\path[<->] (i_1) edge node [fill=none]{} (i_2);
				\path[<->] (i_1) edge node [fill=none]{} (i_3);
				\path[<->] (i_2) edge node [fill=none]{} (i_3);
				\path[<->] (i_4) edge node [fill=none]{} (i_5);
				\path[<->] (i_4) edge node [fill=none]{} (i_6);
				\path[<->] (i_7) edge node [fill=none]{} (i_5);
				\path[<->] (i_7) edge node [fill=none]{} (i_6);
				\end{tikzpicture}\end{center}
			\caption{\label{fig:monotonic_fail} \textit{Failing monotonic centrality.} The above networks depict failure of monotonic centrality in the presence of supermodularity (i.e. no absence of modularity) from Example \ref{ex:supermod:nonmonotonic:centrality}.}
		\end{figure}
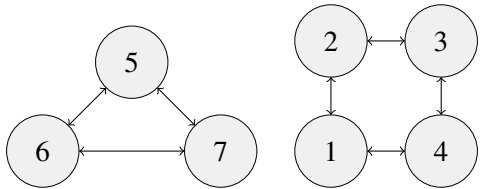

		\begin{example}\label{ex:supermod:nonmonotonic:centrality}
			This example shows how monotonic $\delta $-decay centrality can fail when there is supermodularity. The stability concept is strong stability.

			Suppose there are seven agents - three of high type (1-3) and four of low type (4-7). There is supermodularity; a degree quota of two, and; a level of network externalities $\delta $. Consider the segregated network $\mu =\{12,23,31\}\cup\{45,56,67,74\}$, and the connected network $\tilde{\mu} =\{12,23,34,45,56,67,71\}$. If the inequality below holds with "$>$" then $\mu $ is an efficient network; if "$<$" then $\tilde{\mu} $ is an efficient network. 	
			\begin{eqnarray}
			3\cdot Z (\bar{x} ,\bar{x} )+4\cdot Z (\underaccent{\bar}{x} ,\underaccent{\bar}{x} )\gtreqqless
			2\cdot Z (\bar{x} ,\bar{x} )+3\cdot Z (\underaccent{\bar}{x} ,\underaccent{\bar}{x} )+(2+4\delta +6\delta ^2)\cdot Z (\bar{x} ,\underaccent{\bar}{x} ).\label{eq:supermodularity:centrality}
			\end{eqnarray}
			When Inequality \ref{eq:supermodularity:centrality} holds with "$>$" then $\mu $ is furthermore strongly stable when for every pair of agents $i ,j \in N ;i \ne j $ it holds that there are no transfers between them, i.e. $\tau_{i j }=0$. In the network $\mu $ it holds that $d_{i }^{\delta }(\mu )>d_{j }^{\delta }(\mu )$ where $i \in \{1,2,3,4\},j \in \{5,6,7\}$. Note that Inequality \ref{eq:supermodularity:centrality} holds with "$>$" if $\hat{Z} >1+2\delta +3\delta ^2$ and $\hat{Z} >1$ due to supermodularity.
		\end{example}

		\begin{lemma}
			\label{claim:example:monotonic_fail:addition} In Example \ref{example:monotonic_fail} the sufficient conditions for existence of $\tau $ such that $\mu $ is stable are linear costs ($\tilde{c} $ of forming a link) and $\delta $ being sufficiently high along with the two inequalities:
			\[2\cdot Z (\bar{x} ,\underaccent{\bar}{x} )+Z (\bar{x} ,\bar{x} )\ge 4\tilde{c} \quad\mbox{and}\quad
			Z (\bar{x} ,\underaccent{\bar}{x} )>\tilde{c} .\]
		\end{lemma}  
		\begin{proof}	
			Linear cost entails that $c (2)=2c (1)$. Note that the total cost of establishing a link is $2\tilde{c} $. Throughout we assume symmetry such that $s_{1}=s_{2}$ and $\tau_{12}=0$ which implies 
			$\tau_{3 1}=\tau_{3 2}$. To shorten notation let  $s_{i }=s_{i }(\mu ,\tau )$.
			
			We begin with deriving the required transfers and net utility for the low type, i.e. agent 3, the relevant Inequalities are \ref{eq:monotone:fail:ex:deleteall:solo}, \ref{eq:monotone:fail:ex:deleteone:agent1}, \ref{eq:monotone:fail:ex:onlyonelink:allbutonelink}. The two inequalities (for $i =1,2$) in Inequality \ref{eq:monotone:fail:ex:deleteone:agent1}  can be rewritten using that $\tau_{31}=\tau_{32}$:
			\begin{eqnarray}
			\tau_{31} & \ge  & \tilde{c} -z (\underaccent{\bar}{x} ,\bar{x} ),\nonumber\\
			\tau_{31}+\tau_{32} & \ge  & 2\tilde{c} -2z (\underaccent{\bar}{x} ,\bar{x} ),\nonumber\\
			2\tilde{c} -2z (\underaccent{\bar}{x} ,\bar{x} )+s_{3} & \ge  & 2\tilde{c} -2z (\underaccent{\bar}{x} ,\bar{x} ),\nonumber\\
			s_{3} & \ge  & 0.\label{eq:monotone:fail:ex:implied1}
			\end{eqnarray}
			
			Thus Inequality \ref{eq:monotone:fail:ex:deleteone:agent1} is irrelevant when $\tau_{1 2}=\tau_{1 3}$ and Inequality \ref{eq:monotone:fail:ex:deleteall:solo} ($s_{3}\ge 0$) is satisfied.
			
			By adding the two inequalities in Inequality \ref{eq:monotone:fail:ex:onlyonelink:allbutonelink} where $i =3$ and $j =1,2$ together and using non-wastefulness ($\Sigma_{i \in N }s_{i }(\mu ,\tau )=U (\mu )$) the following must hold for the net utility of agent $3$:
			\begin{eqnarray}
			s_{3} & \ge  & 2Z (\bar{x} ,\underaccent{\bar}{x} )-4\tilde{c} -
			[\delta  Z (\bar{x} ,\bar{x} )+2Z (\bar{x} ,\underaccent{\bar}{x} )-4\tilde{c} ],\nonumber\\
			s_{3} & \ge  & -\delta  Z (\bar{x} ,\bar{x} ).\label{eq:monotone:fail:ex:implied1:irrelevant}
			\end{eqnarray}
			
			We see that Inequality \ref{eq:monotone:fail:ex:implied1:irrelevant} is satisfied when Inequality \ref{eq:monotone:fail:ex:deleteall:solo} holds ($s_{3}\ge 0$). Note that the two inequalities in Inequality \ref{eq:monotone:fail:ex:onlyonelink:allbutonelink} where $i =3$ and $j =1,2$ must still hold. 
			
			We now turn to derive to check relevant transfer and net utilities for agent 1 and 2. The relevant expressions are Inequalities \ref{eq:monotone:fail:ex:deleteall:solo}, \ref{eq:monotone:fail:ex:form:nodelete}, \ref{eq:monotone:fail:ex:form:onedelete} and \ref{eq:monotone:fail:ex:onlyonelink:allbutonelink}. For large enough $\delta $ (i.e. $\delta \rightarrow1$) then it always holds that Inequality \ref{eq:monotone:fail:ex:form:nodelete} is satisfied; thus it suffices to check Inequalities \ref{eq:monotone:fail:ex:deleteall:solo}, \ref{eq:monotone:fail:ex:form:onedelete} and \ref{eq:monotone:fail:ex:onlyonelink:allbutonelink}.

			We use a similar procedure to see the implication on net utility of agent $i \in \{1,2\}$. Let $j =\{1,2\}\backslash\{i \}$. By adding the two inequalities from Inequality \ref{eq:monotone:fail:ex:onlyonelink:allbutonelink} for the pair of $3,i $ and the pair $i ,j $, together and using non-wastefulness ($\Sigma_{i \in N }s_{i }(\mu ,\tau )=U (\mu )$) it must hold that:
			\begin{eqnarray}
			\min \{s_{1},s_{2}\} & \ge  & Z (\bar{x} ,\underaccent{\bar}{x} )+Z (\bar{x} ,\bar{x} )-4\tilde{c} -
			[\delta  Z (\bar{x} ,\bar{x} )+2Z (\bar{x} ,\underaccent{\bar}{x} )-4\tilde{c} ],\nonumber\\
			\min \{s_{1},s_{2}\} & \ge  & (1-\delta )Z (\bar{x} ,\bar{x} )-Z (\bar{x} ,\underaccent{\bar}{x} ).
			\label{eq:monotone:fail:ex:implied23:1}
			\end{eqnarray}
			
			For sufficiently large $\delta $ then Inequality \ref{eq:monotone:fail:ex:implied23:1} becomes irrelevant as in the limit of $\delta \rightarrow1$ its requirement is $\min \{s_{1},s_{2}\}\ge -Z (\bar{x} ,\underaccent{\bar}{x} )$ which is always satisfied when $\min \{s_{1},s_{2}\}\ge 0$. Note that it still remains to check all three conditions from Inequality \ref{eq:monotone:fail:ex:onlyonelink:allbutonelink} are valid.
			
			In addition, it is possible to rewrite Inequality \ref{eq:monotone:fail:ex:form:onedelete} in to Inequality \ref{eq:monotone:fail:ex:implied23:2}. Note that Inequality \ref{eq:monotone:fail:ex:implied23:2} is a sufficient condition for Inequality \ref{eq:monotone:fail:ex:implied23:2} for both agent 1 and 2.
			\begin{eqnarray}
			\nonumber
			s_{1}+s_{2} & \ge  & Z (\bar{x} ,\bar{x} )+(1+\delta )z (\bar{x} ,\underaccent{\bar}{x} )+\max \{\tau_{12},\tau_{13}\}-3\tilde{c} ,
			\\
			\nonumber
			s_{1}+s_{2}
			& \ge  &
			Z (\bar{x} ,\bar{x} )+(1+\delta )z (\bar{x} ,\underaccent{\bar}{x} )+\max \{s_{1},s_{2}\}-\delta  z (\bar{x} ,\bar{x} )-z (\bar{x} ,\underaccent{\bar}{x} )-2\tilde{c} ,
			\\
			\label{eq:monotone:fail:ex:implied23:2}
			\min \{s_{1},s_{2}\}
			& \ge  & (2-\delta )z (\bar{x} ,\bar{x} )+
			\delta z (\bar{x} ,\underaccent{\bar}{x} ) -2\tilde{c} .
			\end{eqnarray}	
			
			In the remainder of this proof we restrict that $s_{3}=0$ and thus $\tau_{12}=\tau_{13}=z (\underaccent{\bar}{x} ,\bar{x} )-\tilde{c} $. This restriction ensures the minimal requirements for the transfer and net utility for agent 3 are met (see Inequalities \ref{eq:monotone:fail:ex:deleteall:solo}, \ref{eq:monotone:fail:ex:deleteone:agent1}, \ref{eq:monotone:fail:ex:onlyonelink:allbutonelink}). This restriction allows us to assess when the remaining requirements are satisfied for agents 1 and 2 given that the requirements for agent  are minimally satisfied. When $s_{3}=0$ then Inequality \ref{eq:monotone:fail:ex:onlyonelink:allbutonelink} is satisfied if:
			\begin{eqnarray}
			\min \{s_{1},s_{2}\} & \ge  & Z (\bar{x} ,\underaccent{\bar}{x} )-2\tilde{c} .
			\label{eq:monotone:fail:ex:implied23:3}
			\end{eqnarray}
			
			It is possible to further restrict sufficient conditions for when the remaining relevant inequalities are satisfied when $s_{1}=s_{2}$ and $s_{3}=0$. Recall that it is necessary to check Inequalities \ref{eq:monotone:fail:ex:deleteall:solo}, \ref{eq:monotone:fail:ex:form:onedelete} and \ref{eq:monotone:fail:ex:onlyonelink:allbutonelink} for large enough $\delta $. 
			The sufficient condition for the inequalities is captured in Inequality \ref{eq:monotone:fail:ex:implied23:4} below. 
			The elements in the set from which the maximal element is chosen derived from respectively: 
			Inequality \ref{eq:monotone:fail:ex:deleteall:solo}; 
			Inequality \ref{eq:monotone:fail:ex:implied23:2} which implies Inequality \ref{eq:monotone:fail:ex:form:onedelete} is satisfied;						
			Inequality \ref{eq:monotone:fail:ex:onlyonelink:allbutonelink} for the pair 1,2;
			Inequality \ref{eq:monotone:fail:ex:implied23:3} which implies Inequality \ref{eq:monotone:fail:ex:onlyonelink:allbutonelink} for the pair 1,3 and the pair 2,3.
			\begin{eqnarray}
			s_{1}+s_{2}	&\ge &\max \{0,\,\,
			2(2-\delta )z (\bar{x} ,\bar{x} )+2\delta z (\bar{x} ,\underaccent{\bar}{x} )-4\tilde{c} ,\,\,
			Z (\bar{x} ,\bar{x} )-2\tilde{c} ,\,\,
			2Z (\bar{x} ,\underaccent{\bar}{x} )-4\tilde{c} 
			\}.
			\label{eq:monotone:fail:ex:implied23:4}
			\end{eqnarray}
			
			Using that $s_{3}=0$ it holds that  $s_{1}+s_{2}=s_{1}+s_{2}+s_{3}$. Combining non-wastefulness ($\Sigma_{i \in N }s_{i }(\mu ,\tau )=U (\mu )$) with Inequality \ref{eq:monotone:fail:ex:implied23:4} we can derive the following four inequalities - each inequality correspond to an n'th element in the set from which the maximal element is chosen.
			\begin{eqnarray}
			2Z (\bar{x} ,\underaccent{\bar}{x} )+\delta Z (\bar{x} ,\bar{x} )-4\tilde{c} &\ge & 0,
			\label{eq:monotone:fail:ex:implied:all:1}
			\\&&\nonumber\\
			2Z (\bar{x} ,\underaccent{\bar}{x} )+
			\delta Z (\bar{x} ,\bar{x} )-4\tilde{c} &\ge & 2(2-\delta )z (\bar{x} ,\bar{x} )+2\delta z (\bar{x} ,\underaccent{\bar}{x} )-4\tilde{c} ,
			\nonumber\\
			2\cdot (\underaccent{\bar}{x} ,\bar{x} )+2(1-\delta )(\bar{x} ,\underaccent{\bar}{x} )&\ge & 
			2(1-\delta )Z (\bar{x} ,\bar{x} ),\label{eq:monotone:fail:ex:implied:all:2}
			\\&&\nonumber\\
			2Z (\bar{x} ,\underaccent{\bar}{x} )+
			\delta Z (\bar{x} ,\bar{x} )-4\tilde{c} &\ge & Z (\bar{x} ,\bar{x} )-2\tilde{c} ,\nonumber\\
			2Z (\bar{x} ,\underaccent{\bar}{x} )&\ge &(1-\delta ) Z (\bar{x} ,\bar{x} )+2\tilde{c} ,
			\label{eq:monotone:fail:ex:implied:all:3}		
			\\&&\nonumber\\
			2Z (\bar{x} ,\underaccent{\bar}{x} )+
			\delta Z (\bar{x} ,\bar{x} )-4\tilde{c} 
			&\ge & 2Z (\bar{x} ,\underaccent{\bar}{x} )-4\tilde{c} ,\nonumber\\
			(1-\delta ) Z (\bar{x} ,\bar{x} )&\ge &0.
			\label{eq:monotone:fail:ex:implied:all:4}					
			\end{eqnarray}
			
			Of the above we see that Inequality \ref{eq:monotone:fail:ex:implied:all:2} holds when $\delta $ is large enough while Inequality \ref{eq:monotone:fail:ex:implied:all:4} always holds. Thus all the above inequalities can be satisfied for large enough $\delta $ when $2Z (\bar{x} ,\underaccent{\bar}{x} )+	\delta Z (\bar{x} ,\bar{x} )\ge 4\tilde{c} $ and $Z (\bar{x} ,\underaccent{\bar}{x} )>\tilde{c} $.
		\end{proof}
		
		In what follows we demonstrate two auxiliary results: Lemma  \ref{claim:monotonic_fail} and Lemma  \ref{claim:monotonic_fail_remark}. These two results imply Corollary \ref{claim:core_fail_externalities} and \ref{claim:fail_corollary}. 
		The first auxiliary lemma rely on Condition \ref{cond:failurecentral} below; this condition is a generalization of the conditions from Example \ref{example:monotonic_fail}.
		
		\begin{condition}\label{cond:failurecentral}
			Let $\mu =\cup_{i \in {1,..,n -1}}\{i n \}$; transfers are restricted as follows:
			\begin{eqnarray}\label{eq:monotone:fail:claim:nonwastefulness}
			\Sigma_{i \in N }s_{i }&=&U (\mu )\nonumber\\&=&
			\Sigma_{i \in \{1,..,n -1\}}Z_{i  n}+
			\Sigma_{i ,j \in \{1,..,n -1\},i >j }		\delta Z_{i  j }-(n -1)\cdot c (1)-c (n -1),\label{eq:nonwastefulness}
			\end{eqnarray}
			along with the following three set of inequalities,
			\begin{eqnarray}
			\label{eq:monotone:fail:claim:deleteall:solo}			
			\forall {i }\in N :\quad
			s_{i }&\ge &0,
			\\
			\label{eq:monotone:fail:claim:deleteone:agent1}				
			\forall {i }\ne{n }:\quad 
			\tau_{n  i }&\ge & c (n -1)-c (n -2)-z_{n i },
			\\
			\label{eq:monotone:fail:claim:onlyonelink:allbutonelink}
			\forall {i }, {j }\in N ,{i }\ne{j }:
			\quad
			s_{i }+s_{j }
			& \ge  &
			Z_{i  j }-2\cdot c (1).
			\end{eqnarray}		
			and finally the following two set of inequalities must hold for any two ${i }, {j }\in (N \backslash\{n \}),{i }\ne{j }$:
			\begin{eqnarray}
			\label{eq:monotone:fail:claim:form:onedelete}
			s_{i }+s_{j }
			& \ge  &
			Z_{i  j }
			+z_{i  n }+\delta z_{j  n }
			+\delta \Sigma_{l \notin\{i ,j ,n \}}[z_{i l }+\delta z_{j l }]-c (1)-c (2)+\tau_{i n },
			\\
			\label{eq:monotone:fail:claim:form:nodelete}
			2\cdot [c (2)-c (1)]
			& \ge  &
			(1-\delta )\cdot Z_{i  j }.
			\end{eqnarray}\end{condition}

		\begin{lemma}
			\label{claim:monotonic_fail}Suppose there are at least three agents ($n \ge 3$), weak supermodularity and monotonicity, a cost function and an interval $[\tilde{\delta} ,\hat{\delta} ]$ where all inequalities from Condition  \ref{cond:failurecentral} are satisfied, then for every  $\delta \in [\tilde{\delta} ,\hat{\delta} ]$ the low value sponsored star is pairwise stable.
		\end{lemma}

		\begin{proof} 
			The network $\mu $ be a low value sponsored star with agent $n $ as center, i.e.
			$\mu =\cup_{i \in {1,..,n -1}}\{i n \}$. We show that $\mu $ is pairwise stable if the relevant conditions are met.			
						
			First we check for deviations from $\mu $ where an agent deletes all its links in $\mu $. This move cannot be profitable when Inequality \ref{eq:monotone:fail:claim:deleteall:solo} is satisfied for all individuals. This implies for any agent but agent $n $ it is not profitable to delete their single link in $\mu $. Furthermore for agent $n $ deleting any link is not profitable due to due to Inequality \ref{eq:monotone:fail:claim:deleteone:agent1} being satisfied. This also implies that deleting any number of links for agent $n $ is not beneficial as each is beneficial at the highest marginal cost.
			
			Second network $\mu $ is robust to deviations where two agents delete all links and form a single link as any two agents can only gain less by such a move cf. Inequality \ref{eq:monotone:fail:claim:onlyonelink:allbutonelink}.
			
			Finally the network $\mu $ is robust to deviation where a link is formed by two agents of which one agent deletes its link to $n $ but another agent keeps its link to agent $n $ when Inequality \ref{eq:monotone:fail:claim:form:onedelete} is satisfied. The network $\mu $ is robust to deviations where both agents keep their current link with agent $n $ and form a new link together if Inequality \ref{eq:monotone:fail:claim:form:nodelete} is satisfied.
		\end{proof}
		
		\begin{lemma}
			\label{claim:monotonic_fail_remark} Sufficient conditions for a low value sponsored star, with agent $n $ as center, to be stable are: linear linking costs; many agents; $\delta  \in  [1-\nicefrac{2\tilde{c} }{Z (x_1,\,x_2)},1)$; transfers s.t. $s_{n }\ge 0$ and
			\begin{eqnarray*}				
			s_{i }
			& \ge  & \max _{j \notin\{i ,n \}}
			\left[(1-\delta ^2)Z_{i j }-
			(\delta -\delta ^2)z_{j  i }\right]+\delta z_{i  n }+\delta ^2\Sigma_{\iota \notin\{i ,n \}}z_{i  \iota }
			-2\tilde{c} ,\qquad i \ne n .
			\end{eqnarray*} 
		\end{lemma} 
		
		\begin{proof}
			The setting is as specified by Lemma  \ref{claim:monotonic_fail} where costs are linear such that total costs are $c (k )=k \cdot \tilde{c} $. Note the total cost of establishing a link then is $2\tilde{c} $ as two agents are required. 
			We use a procedure similar to the one for Lemma \ref{claim:example:monotonic_fail:addition} to find the restrictions on net utility for agents. We first investigate restrictions for agent $n $ and then the remaining agents.

			
			 We examine the implication for the lower bound on net-utility for agent $n $ from the $n $-1 inequalities for agent $n $ with the other agents in Inequality \ref{eq:monotone:fail:claim:deleteone:agent1} by summing over the inequalities:
			\begin{eqnarray}
			\Sigma_{i \ne n }\tau_{n  i } & \ge  &
			(n -1)\tilde{c} -\Sigma_{i \ne n }z_{n i },\nonumber\\
			s_{n } & \ge  & 0,\label{eq:monotone:fail:claim:implied1}
			\end{eqnarray}
			
			Thus if the $n -1$ conditions from Inequality \ref{eq:monotone:fail:claim:deleteone:agent1} hold it is a sufficient condition for Inequality \ref{eq:monotone:fail:claim:deleteall:solo} also to hold for agent $n $; thus we only require that $\tau_{n  i }\ge \tilde{c} -z (\underaccent{\bar}{x} ,x_{i })$ for agent $i \ne n $.
			
						
			We analogously examine the implication of the other $n -1$ inequalities in Inequality \ref{eq:monotone:fail:claim:onlyonelink:allbutonelink} for agent $n $'s lower bound on net-utility by summing and subtracting Equation \ref{eq:nonwastefulness}:
			\begin{eqnarray}
			s_{n } & \ge  & 		-\delta \cdot (
			\Sigma_{i ,j \in \{1,..,n -1\},i >j }
			Z_{i j }).\label{eq:monotone:fail:claim:implied1:irrelevant}
			\end{eqnarray}
						
			The above inequality is less strict than Inequality \ref{eq:monotone:fail:claim:deleteall:solo} thus without implication on agent $n $.

			We now investigate restrictions for other agents. 
			By summing over Inequality \ref{eq:monotone:fail:claim:onlyonelink:allbutonelink} for agent $i \in \{1,..,n -1\}$ with other agents and subtracting Equation \ref{eq:monotone:fail:claim:nonwastefulness} we get:
			\begin{eqnarray}
			s_{i }\cdot (n -2) & \ge  & \Sigma_{j \notin\{i ,n \}}Z_{i j }
			-\Sigma_{j \notin\{i ,n \}}Z_{j  n}-
			\delta \Sigma_{i ,j \in \{1,..,n -1\},i >j }Z_{i  j }.
			\nonumber
			\end{eqnarray}
				
			Due to monotonicity it holds that we can bound link value for any $i ,j : Z (\bar{x} ,\bar{x} )\ge Z_{i j }\ge Z (\underaccent{\bar}{x} ,\underaccent{\bar}{x} )$. A sufficient condition for the above inequality  is: 
			\begin{eqnarray}
			s_{i } & \ge  & \frac{(n -1)\cdot Z (\bar{x} ,\bar{x} )-
			[(n -1)+\delta (n -2)^2]\cdot Z (\underaccent{\bar}{x} ,\underaccent{\bar}{x} )}{n -2}.
			\nonumber
			\end{eqnarray}
				
			For $\delta >0$ and sufficiently many agents (as $n \rightarrow\in fty$) then the above inequality 
			is irrelevant as the right hand side is negative for $i \ne n $; thus Inequality \ref{eq:monotone:fail:claim:onlyonelink:allbutonelink} will not be binding a constraint.

			
			In addition, for any two agents $i ,j \ne1$ it is possible to rewrite Inequality \ref{eq:monotone:fail:claim:form:onedelete} in to the following (using that $\tau_{i j }=0$ when $i ,j \ne n $):
			\begin{eqnarray}
			\nonumber
			s_{i }+s_{j } & \ge  & Z_{i j }
			+z_{j  n }+\delta z_{i  n }
			+\Sigma_{\iota \notin\{i ,j ,n \}}[\delta z_{j  \iota }+\delta ^2z_{i  \iota }]-3\tilde{c} +\tau_{j n },
			\\
			\nonumber
			s_{i }+s_{j } & \ge  &
			Z_{i j }
			+\delta z_{i  n }
			+\Sigma_{\iota \notin\{i ,j ,n \}}\delta ^2z_{i  \iota }-2\tilde{c} -\delta z_{j  i }+s_{j },
			\\
			\label{eq:monotone:fail:claim:implied:n:1}
			s_{i }
			& \ge  & z_{i  j }+
			(1-\delta )z_{j  i }+\delta z_{i  n }+\Sigma_{\iota \notin\{i ,j ,n \}}\delta ^2z_{i  \iota }-2\tilde{c} .
			\end{eqnarray}

			The maximum of Inequality \ref{eq:monotone:fail:claim:implied:n:1} for agent $i $ is:
			\begin{eqnarray}
			s_{i }
			& \ge  & \max _{j \notin\{i ,n \}}
			\left[z_{i  j }+
			(1-\delta )z_{j  i }+\delta z_{i  n }+\delta ^2\Sigma_{\iota \notin\{i ,j ,n \}}z_{i  \iota }\right]
			-2\tilde{c} ,
			\nonumber
			\\
			s_{i }
			& \ge  & \max _{j \notin\{i ,n \}}
			\left[(1-\delta ^2)z_{i  j }+
			(1-\delta )z_{j  i }\right]+\delta z_{i  n }+\delta ^2\Sigma_{\iota \notin\{i ,n \}}z_{i  \iota }
			-2\tilde{c} ,
			\nonumber
			\\
			s_{i }
			& \ge  & \max _{j \notin\{i ,n \}}
			\left[(1-\delta ^2)Z_{i j }-
			(\delta -\delta ^2)z_{j  i }\right]+\delta z_{i  n }+\delta ^2\Sigma_{\iota \notin\{i ,n \}}z_{i  \iota }
			-2\tilde{c} .
			\nonumber
			\end{eqnarray}
		
			The above inequality can be satisfied for all agents apart excluding agent $n $ when:
			\begin{eqnarray}
			\Sigma_{i \in \{1,..,n -1\}}s_{i }
			& \ge  &
			\Sigma_{i \in \{1,..,n -1\}}	[(1-\delta ^2)Z_{i 1}-
			(\delta -\delta ^2)z_{1i }+\delta z_{i  n }]+\nonumber
			\\&&\Sigma_{i ,j \in \{1,..,n -1\},i >j }
			\delta ^2Z_{i j }
			-(2n -2)\tilde{c} .
			\label{eq:monotone:fail:claim:implied:n:3}
			\end{eqnarray}
			
			We check validity of Inequality \ref{eq:monotone:fail:claim:implied:n:3} using non-wastefulness in Equation \ref{eq:monotone:fail:claim:nonwastefulness}:
			\begin{eqnarray}
			s_{n }&=&\Sigma_{i \in \{1,..,n \}}s_{i }- \Sigma_{i \in \{1,..,n -1\}}s_{i }\nonumber\\
			& \le  &
			\Sigma_{i \in \{1,..,n -1\}}Z_{i  n }+
			\Sigma_{i ,j \in \{1,..,n -1\},i >j }		\delta Z_{i j }-(n -1)\cdot c (1)-c (n -1),
			\nonumber
			\\&&
			-[\Sigma_{i \in \{1,..,n -1\}}	[(1-\delta ^2)Z_{i  1}-
			(\delta -\delta ^2)z_{1i }+\delta z_{i  n }]+
			\Sigma_{i ,j \in \{1,..,n -1\},i >j }
			\delta ^2Z_{i j }
			-(2n -2)\tilde{c} ],
			\nonumber\\&=&
			\Sigma_{i \in \{1,..,n -1\}}[(1-\delta )Z_{i  n }-(1-\delta ^2)Z_{i  1}+
			(\delta -\delta ^2)z_{1i }+\delta z_{n  i }]+
			\Sigma_{i ,j \in \{1,..,n -1\},i >j }
			(\delta -\delta ^2)Z_{i j },
			\nonumber\\&=&
			(1-\delta )U (\mu )-\Sigma_{i \in \{1,..,n -1\}}[
			(1-\delta ^2)Z_{i  1}-
			(\delta -\delta ^2)z_{1i }
			-\delta z_{n i }].\nonumber
			\end{eqnarray}
			
			For many agents it must be that $(1-\delta )U (\mu )\ge \Sigma_{i \in \{1,..,n -1\}}[
			(1-\delta ^2)Z_{i  1}-
			(\delta -\delta ^2)z_{1i }
			-\delta z_{n i }]$ as $U (\mu )$ is quadratic in size while the right hand side is linear. Thus $s_{n }\ge 0$ can be satisfied for large populations along with the conditions for other agents.
			
			Finally we see that Inequality \ref{eq:monotone:fail:claim:form:nodelete} is satisfied for all agents when:
			\begin{eqnarray*}
				2\tilde{c}  & \ge  & (1-\delta )\cdot Z (x_1,x_2),\\
				\nicefrac{2\tilde{c} }{Z (x_1,x_2)} & \ge  & (1-\delta ),\\
				\delta  & \ge  & 1-\nicefrac{2\tilde{c} }{Z (x_1,x_2)}.
			\end{eqnarray*}
			
		\end{proof}

	\section{Supplementary appendix: Suboptimal sorting for constant decay}\label{app:oversort_generalize}
		
		This appendix shows a how suboptimally sorted networks are also prevalent under constant decay. It is split into two sub-appendices: sub-appendix \ref{app:oversort_local_tree} which deals with demonstrating the results and sub-appendix \ref{app:local_trees} which only contains auxiliary results.

	\subsection{Suboptimal sorting in local trees}\label{app:oversort_local_tree}
	
		We show sorting may be pairwise stable but suboptimal under constant decay for a subclass of networks. We begin by describing this subclass. Informally put, the relevant subclass of perfectly sorted networks where each subnetwork for a given type has a certain structure. 
		The structure of each subnetwork is such that from the perspective of every agent (i.e. the ego-network) each subnetwork appears as a tree when disregarding the links of the agents furthest away. Note that a \textit{tree} is network where every pair of agents are connected by a unique path. Thus these subnetworks are called local trees as they are not trees in a global sense but only when disregarding most distant agents. 
		
		The formal definition is as described below. The definition employs the network \textit{diameter} which is the maximum distance between any two agents, i.e.  $m (\mu )=\sup_{i ,j \in N }p_{i j }(\mu )$. 
		
		\begin{definition}\label{def:local_tree}
			A network $\mu $ is a \textbf{local tree} when each agent $i $ has $\kappa $ links where:
			\begin{itemize}
				\item 			for each other agent $j \ne i $ at distance $p_{i j }(\mu )\le  m_{n ,\kappa }-2$ there are $\kappa -1$ links between agent $j $ and $j' $ such that $j' $ is one step further away, i.e. $p_{i j }(\mu )=p_{i j' }(\mu )-1$;
				\item the network diameter $m (\mu )=m_{n ,\kappa }$,
				\begin{eqnarray}
				m_{n ,\kappa } & = & \arg\min _{m }\{ m:\,\Sigma_{l =1}^{m }(\kappa (\kappa -1)^{l -1})+1\ge  n \}. \label{eq:maximum_path_length_local_tree}
				\end{eqnarray}		
			\end{itemize}
		\end{definition}
		
		The structure of local trees entails that each agent has $\kappa \cdot (\kappa -1)^{p -1}$ agents at distance $p <m $, where $m =m_{n ,\kappa }$.
		At distance $p =m $ there are  $n -\sum_{l =1}^{m -1}\kappa \cdot (\kappa -1)^{l -1}$ (all remaining agents). This structure implies that every agent's utility is maximized subject to the constraint of all agents having at most $\kappa $ links;\footnote{The maximization of utility follows from the observation that each agent has at most $\kappa $ links, so at distance $p $ there can be at most $\kappa \cdot (\kappa -1)^{p -1}$ agents.} a side effect is that utility before transfers is symmetric.
		
		A necessary condition for local trees to exist is that there is no link surplus, i.e. degree quota is binding ($\forall i \in N :k_i =\kappa $). Note this binding condition is only possible when $n \cdot \kappa $ is even.
		
		When a local tree network fulfills $n =\sum_{l =1}^{m }\kappa \cdot (\kappa -1)^{l -1}$ then it is an \emph{exact local tree}. See the next sub-appendix for an elaborate treatment of structure of exactly local trees. Two subclasses of exact local trees which are worth mentioning. 
		The first is a network known as a cycle or a ring. The cycle is characterized by having a minimal possible degree quota ($\kappa =2$) among local trees and a maximal diameter $(m =\left\lceil \frac{n -1}{2}\right\rceil)$. The second is a \textit{clique} where all agents are linked, i.e. the complete network. Cliques have maximal degree quotas $(\kappa =n -1)$ and minimal diameters $(m =1)$. Both subclasses has a network which exists for any $n $. Note that in Example \ref{example:excessiveassortative} each of the two components is both a cycle and a clique. Note that there exist non-trivial networks beyond the cycle and the clique.\footnote{An example is  $\{i_1i_2,i_1i_3,i_1i_4,i_2i_5,i_2i_6,i_3i_7,i_3i_8,i_4i_9,i_4i_{10},i_5i_7,i_5i_9,i_6i_8,i_6i_{10},i_7i_{9},i_{8}i_{10}\}$ when $n =10,\kappa =3$.}

		In order to derive our results it is necessary to restrict ourselves to  a subset of local trees. The subset are those local trees where the deletion of links leads to equal losses to both of agents whose link is deleted; thus we refer to these local trees as having symmetric losses:
		\begin{definition}
			\label{def:local_tree_symmetric_losses}
			A local tree $\mu $ has \textbf{symmetric losses} when at every distance $p =1,..,m $ it holds that $|\{i \in N : p_{\iota i }(\mu \backslash\{\iota \iota' \})=p \}|
			=
			|\{i \in N : p_{\iota' i }(\mu \backslash\{\iota \iota' \})=p \}|$. 
		\end{definition}
		
		Denote the set of perfectly sorted networks where the subnetwork for each type is a local tree with symmetric losses as $M ^{p-srt:symm.\,loc-tree }$.
		
		Whether or not symmetric losses is a generic property for all local trees is an open question. However, in simulations that we perform it holds all network configurations which are local trees up to size $n =10$ have symmetric losses (see result below and proof for exhibition of examples).  Moreover for size up to $n =16$ it has been shown to hold for any networks examined in the simulation. 
		
		A generalization of stable but suboptimal sorting under constant decay is expressed below. While allowing for constant decay rather than hyperbolic it the set of networks are further restricted. 
		\begin{theorem}\label{claim:assort:ext:pairwise:excess}
			
			Suppose there is supermodularity, a degree quota $\kappa $ and each type has equal number of agents then 
					\begin{enumerate}[(i)]
						\item $\hat{M} \cap M ^{\max  U }_{\delta >\ubar{\delta} }=\emptyset$; 
						\item $\hat{M} \subseteq M ^{p-stb }_{\delta \le \bar{\delta} }$;
					\end{enumerate}							
			where $\hat{M} =M ^{p-srt:symm.\,loc-tree }$ and thresholds $\ubar{\delta} ,\bar{\delta} \in (0,1)$ where $\ubar{\delta} <\bar{\delta} $
		\end{theorem}

			\begin{proof}
				We show properties (i) and (ii) together. Let $\mu $ be a network which is segregated into $|X |$ components where each component is a local tree with $n /|X |$ agents. Let there be no transfers between any agents. 
				
				As each subnetwork for a given type is a local tree it is stable against deviations by agents of the same type - this follows as local trees provides maximal possible benefits among feasible structures of the subnetwork for all agents in the subnetwork. Thus only two agents of different types may have a profitable deviation which is feasible.
				
				Let $\iota ,j $ be agents of respectively types $x $ and $\tilde{x} $. These two agents can deviate by each deleting a link to $\iota' $ and $j' $ respectively while jointly forming a link. 
				The new network resulting from deletion is denoted $\hat{\mu} =\mu \backslash\{\iota \iota' ,j j' \}$. 
				The move resulting from deletion and forming a link is denoted $\breve{\mu} =\hat{\mu} \cup\{\iota j \}$. 
				An alternative network is $\tilde{\mu} $, the type-bridged network of $\mu $, where the links $\iota \iota' ,j j' $ are removed while the links $\iota j ,\iota' ,j' $ have been formed; thus  $\tilde{\mu} =\hat{\mu} \cup\{\iota j ,\iota' j' \}$.

				Define the gross loss of benefits for $i $ as  $u_i (\hat{\mu} )-u_i (\mu )$ while the gross gains are $u_i (\tilde{\mu} )-u_i (\hat{\mu} )$. There must exist a threshold of externalities $\bar{\delta} \in (0,1)$ where $\mu $ is no longer pairwise stable as cost of deviation monotonically decreases and approaches zero as $\delta \rightarrow1$ while gains are monotonically increasing. 
				The monotonicity of losses is a consequence of the fact that gross loss consists of shortest paths from $\mu $, where $\iota \iota' $ is included in the shortest path, which have longer length in $\hat{\mu} $ and thus are discounted more. Therefore the gross loss is mitigated by a higher $\delta $ as the longer shortest paths are punished less.
				The monotonicity of gains follows as the gains consist of new shortest paths to agents of type $\tilde{x} $ through $\iota j $ and $j' \iota' $ the value of these increases for higher $\delta $.
				
				Exploiting the that Fact \ref{fact:gross_loss_lower_attenuated} and \ref{fact:gross_gain_higher_attenuated} from Appendix \ref{app:local_trees} hold for local trees it follows that for any other agent $i $ of type $x $ (i.e. $i $ is in $N \backslash\{\iota ,\iota' \}$ and $x_{i }=x $):
				\[u_i (\tilde{\mu} )-u_i (\mu )>\delta ^{\min (p_{i \iota }(\tilde{\mu} ),p_{i \iota' }(\tilde{\mu} ))}[u_\iota (\breve{\mu} )-u_\iota (\mu )].\] 
				Aggregating for all agents this implies:
				
				\small
				\begin{eqnarray*}
					U (\tilde{\mu} )-U (\mu )
					&>&
					[u_\iota (\breve{\mu} )-u_\iota (\mu )]\cdot 
					\sum_{x_i =x }
					\delta ^{\min (p_{i \iota }(\tilde{\mu} ),p_{i \iota' }(\tilde{\mu} ))}+ [u_j (\breve{\mu} )-u_j (\mu )]\cdot \sum_{x_i' =\tilde{x} }\delta ^{\min (p_{i j }(\mu ),p_{i j' }(\mu ))}.
				\end{eqnarray*}
				\normalsize

				where $m =m_{n ,\kappa }$. The inequality above implies the following:
				if $U (\tilde{\mu} )-U (\mu )=0$ then $u_\iota (\breve{\mu} )-u_\iota (\mu )+u_j (\breve{\mu} )-u_j (\mu )<0$;  $U (\tilde{\mu} )-U (\mu )>0$ when $u_\iota (\breve{\mu} )-u_\iota (\mu )+u_j (\breve{\mu} )-u_j (\mu )=0$. 	
				It can also be argued that there must exist a threshold, $\ubar{\delta} $, such that when $\delta =\ubar{\delta} $ then $U (\tilde{\mu} )-U (\mu )=0$ and that $\ubar{\delta} <\bar{\delta} $. This follows as $U (\tilde{\mu} )-U (\mu )<0$ for $\delta =0$ and $U (\tilde{\mu} )-U (\mu )>0$ when $u_\iota (\tilde{\mu} )-u_\iota (\mu )+u_j (\tilde{\mu} )-u_j (\mu )=0$ as well as continuity of $U (\tilde{\mu} )-U (\mu )$ in $\delta $.
				
				This entails that for $\delta >\ubar{\delta} $ then $\tilde{\mu} $ provide higher aggregate payoff. Moreover we showed previously that for $\delta <\bar{\delta} $ then $\mu $ is pairwise (Nash) stable. Thus we have proven properties (i) and (ii). 
				
				
			\end{proof}

		For constant decay the thresholds governing when sorting is respectively suboptimal and stable, i.e. $\ubar{\delta} , \bar{\delta} $, can be determined explicitly by solving polynomial equations for every deviation. Moreover, for exact local trees there is a unique solution. 
		In Figure \ref{fig:assort_connection_threshold} the two thresholds from Theorem \ref{claim:assort:ext:pairwise:excess},  $\ubar{\delta} (\hat{Z} ),\bar{\delta} (\hat{Z} )$.

		\begin{figure}[h!]			
			\includegraphics[width=1.0\textwidth]{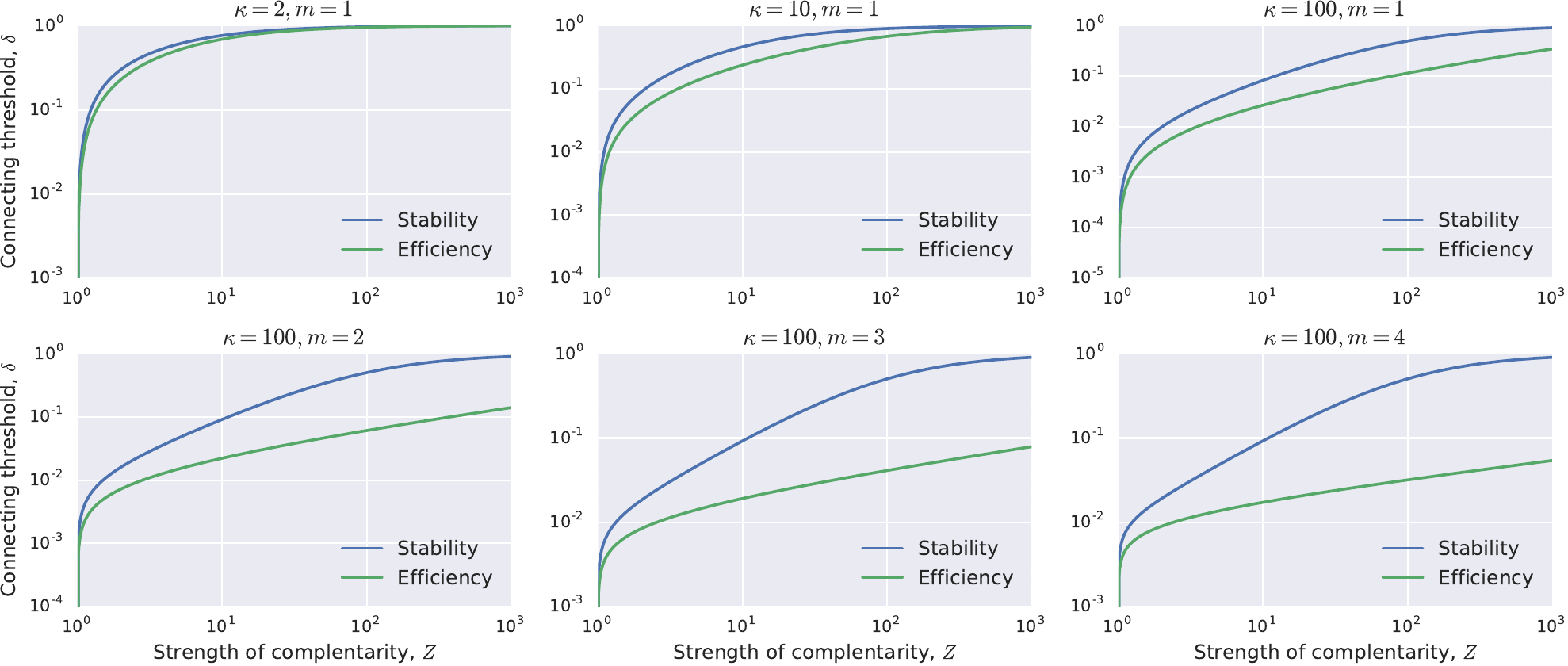}			
			\caption{Visualization of thresholds for connecting from Theorem \ref{claim:assort:ext:pairwise:excess}. The upper diagrams correspond to cliques and the lower ones to exact local trees (where thresholds stem from Equations 
				\ref{eq:utility_loss_exact_local_pairwise},					\ref{eq:utility_loss_exact_local_agg},
				\ref{eq:utility_gains_exact_local_pairwise}, 		\ref{eq:utility_gains_exact_local_agg}).		
				\label{fig:assort_connection_threshold}}
		\end{figure}

		The plots in Figure \ref{fig:assort_connection_threshold} are made for variations of exact local trees. The upper plots corresponds to cliques with various sizes.  The lower plot have fixed degree quota ($\kappa $=100) and the threshold is simulated using pattern in utility that is demonstrated in Appendix \ref{app:local_trees}.
		The plots show the scope for inefficiency, i.e. the gap between $\ubar{\delta} (\hat{Z} ),\bar{\delta} (\hat{Z} )$, increases with the number of agents involved. This makes sense intuitively as the two agents forming the link will fail to account for an increasing number of indirect connections between the two groups. As the number of indirect connections increases at with the squared with total number of agents then larger populations will lead to larger gaps of inefficiency.
				

		\subsection{Local trees}\label{app:local_trees}
		
		This sub-appendix provides auxiliary results for deriving the generalization of suboptimal sorting. We begin our focus on exact local trees and subsequently more generally in local tree networks, see Definition \ref{def:local_tree} in the previous sub-appendix. 
		
		We will examine a generic network $\mu $ which is perfectly sorted and assume that the subset of links for each type is a component that can be classified as either a local tree or an exact local tree. Let networks $\mu_{x }$ and $\mu_{\tilde{x} }$ be the components associated with respectively types $x ,\tilde{x} \in X $. 	We will focus on three particular moves:
		
		\begin{itemize}
		\item \textit{Pairwise deletion of a link}:		
				 Suppose two links $\iota \iota' ,j j' \in \mu $ are deleted and agents $\iota $ and $j $ have respectively type $x $ and $\tilde{x} $; thus the two links are not from the same component. Let the new network that results from removal of the links be denoted $\hat{\mu} =\mu \backslash\{\iota \iota' ,j j' \}$.
		\item \textit{Pairwise formation of a link across types}:
				This move presumes that both agents are also deleting a link. We denote this as a move where agents $\iota $ and $j $  form a link: $\breve{\mu} =\hat{\mu} \,\cup\,\{\iota j \}$.
		\item \textit{Double pairwise formation of a link across types}:
				When two links are formed across types in $\mu $ this corresponds to a non-pairwise deviation as it requires four coalition members. We denote this as a move where both agents $\iota $ and $j $ as well as $\iota' $ and $j' $ form a link:  $\tilde{\mu} =\hat{\mu} \,\cup\,\{\iota j ,\iota' j' \}$.
		\end{itemize}
		
		Finally let $i $ denote a generic agent of type $x $.	Let the shortest path in $\mu $ from $i $ to either $\iota $ or $\iota' $ be denoted $\hat{p}_{i }$ where $\hat{p}_{i }=\min (p_{i \iota }(\hat{\mu} ),p_{i \iota' }(\hat{\mu} ))$. When $\hat{p}_{i }=0$ then either $i =\iota $ or $i =\iota' $.

		\paragraph{Basic properties}
		
		We exploit that $\mu $ is a local tree (see Definition \ref{def:local_tree}). Throughout the remainder of the paper let $m =m_{n ,\kappa }$ (see Equation \ref{eq:maximum_path_length_local_tree}).			
		We express each agent's number of paths of length $p $ as a function of the number of agents and the degree quota: 
		\begin{equation}	
		\mbox{\#}_{i }^{p }(\mu )=
		\kappa (\kappa -1)^{p -1}-\mathbf{1}_{=m }(p )\cdot \Delta\mbox{\#} (n ,\kappa ),\qquad 	\Delta\mbox{\#} (n ,\kappa )= \sum_{l =1}^{m }(\kappa \cdot (\kappa -1)^{l -1})-n ,
		\label{eq:count_shortest_paths_local_tree}
		\end{equation}
		
		where $\mathbf{1}_{=m }(p )$ is the Dirac measure of whether $p =m $. 		
		Using the local tree structure we can express utility without transfers of each agent:
		
		\begin{equation*}
		u_i (\mu )=\sum_{l =1}^{m }\mbox{\#}_{i }^{l }(\mu )\cdot \delta ^l \cdot z (x ,x ). 
		\end{equation*}

		\subsubsection*{Exact local trees}
		Recall exact local trees are local trees where $\Delta\mbox{\#} (n ,\kappa )=0$. We will argue that this entails that exact local trees have the essential property that for every pair of agents there is a unique shortest path of at most length $m $ and the number of paths for every agent is prescribed by Equation \ref{eq:count_shortest_paths_local_tree}. This can be deducted as follows. 
		
		Note first that the fact that the number of walks with at most length $m $ starting in a given agent $i $ cannot exceed $\sum_{p =1}^{m }\mbox{\#} ^p_i (\cdot )$. Recall also that local trees has the property that all agents are reached within distance $m $. Moreover exact local trees has the property that for any agent $i $ it holds that  $n -1=\sum_{p =1}^{m }\mbox{\#} ^p_i (\mu )$; thus every shortest path with distances less than or equal to $m $ must be a unique path between the two particular agents.

		The uniqueness and countability of paths can be used to infer the losses when links are either removed or added to an exact local tree.
		
		\paragraph{Exact local trees - loss from deletion} In order to examine the impact of deletion of a link it is sufficient to analyze what happens to one component of types. This is sufficient as other components as the conclusions are valid for all.

		The deletion of link $\iota \iota' $ implies that any pair of agents $i ,i' $ whose (unique) shortest path in $\mu $ includes the link $\iota \iota' $ will have a new shortest routing path. For exact local trees we can exactly determine the length of the new path. 	
		Let $i $ be the agent whose distance to $\iota $ is least and let $i' $ be the agent whose distance to $\iota' $ is least, i.e. $p_{i \iota }(\mu )<p_{i \iota' }(\mu )$ and $p_{i' \iota' }(\mu )<p_{i' \iota }(\mu )$.

		First when link $\iota \iota' $ is deleted we can show there is no shortest path between $i $ and $i' $ in $\hat{\mu} $ with length below $2m -\hat{p}_i -\hat{p}_{i' }$; that is there is no $i i' $ whose shortest path in $\mu $ includes $\iota \iota' $ such that $p_{i i' }(\hat{\mu} )<2m -\hat{p}_i -\hat{p}_{i' }$. 
		Suppose this was not true. Then there would exist an agent $j $ who (1) is on the new shortest path between $i $ and $i' $ in $\hat{\mu} $ and (2) whose shortest path to agents $\iota $ and $\iota' $ does not include the link $\iota \iota' $ and (3) such that 
		
		\begin{eqnarray*}
		p_{j i }(\hat{\mu} )+p_{j i' }(\hat{\mu} )&<&2m -\hat{p}_i -\hat{p}_{i' },\\
		p_{j i }(\hat{\mu} )+p_{j i' }(\hat{\mu} )&<&2m -\min (p_{i \iota }(\mu ),\,p_{i \iota' }(\mu ))-\min (p_{i' \iota }(\mu ),\,p_{i' \iota' }(\mu )).		
		\end{eqnarray*} 
		
		As by construction $p_{i \iota }(\mu )<p_{i \iota' }(\mu )$ and $p_{i' \iota' }(\mu )<p_{i' \iota }(\mu )$ then the expression above is equivalent to:
		$p_{j i }(\hat{\mu} )+p_{j i' }(\hat{\mu} )<2m -p_{i \iota }(\mu )-p_{i' \iota' }(\mu )$. As the shortest path between $i $ and $\iota $ as well as between $i' $ and $\iota' $ are unchanged from $\mu $ to $\hat{\mu} $ it follows that we can further rewrite into:
		
		\[p_{j i }(\hat{\mu} )+p_{j i' }(\hat{\mu} )<2m -p_{i \iota }(\hat{\mu} )-p_{i' \iota' }(\hat{\mu} )\]
		
		However, the above statement implies that in network $\mu $ that either $\iota $ or $\iota' $ has two paths with lengths of at most $m $ but this violates the definition of exact local trees.
		
		We can now show that when link $\iota \iota' $ is deleted the new shortest path between $i $ and $i' $ in $\hat{\mu} $ has a length of exactly $2m -\hat{p}_i -\hat{p}_{i' }$. This is shown by demonstrating there is an agent $j $ such that $p_{j i }(\hat{\mu} )=m -\hat{p}_i $ and $p_{j i' }(\hat{\mu} )=m -\hat{p}_{i' }$. This can be shown follows. Suppose that $p_{j i }(\hat{\mu} )=m -\hat{p}_i $. We will demonstrate that $p_{j i' }(\hat{\mu} )=m -\hat{p}_{i' }$. As $p_{j i }(\hat{\mu} )=m -\hat{p}_i $ it follows that $p_{j \iota }(\hat{\mu} )=m $. From the definition of exact local trees there must exist a path of length less than $m $ between $j $ and $\iota' $ in network $\mu $. As argued in the paragraph above neither of these paths can be strictly shorter than $m $ and consequently they must both be exactly $m $.
		
		The number of shortest paths of length $p $ which become altered for agent $i $ is $(\kappa -1)^{p -\hat{p}_i -1}$ for $p =\hat{p}_i ,..,m -2,m -1$. This can be demonstrated as follows. If agent $p_{i \iota }(\mu )=m $ and $p_{i \iota' }(\mu )=m $ then no shortest paths are altered; this is clear as agent $i $ as none of the unique shortest paths includes $\iota \iota' $ as they have at most length $m $. If instead $p_{i \iota }(\mu )=m -1$ then the unique shortest path from $i $ to $\iota' $ includes $\iota \iota' $ is the last link; this implies a new shortest path if $\iota \iota' $ is deleted. Thus if $p_{i \iota }(\mu )=m -1$ then one shortest path of length $m $ is lost. When $p_{i \iota }(\mu )=m -2$ then one path of length $m -1$ is lost by the same argument; moreover $\kappa -1$ paths that has $\iota \iota' $ as the second last link. By induction this can be done at higher order and thus for shorter distances.	
		Using the number of rerouted paths shown above we can establish the total number of shortest paths in network $\hat{\mu} $ for agent $i $ that has a length of $p $: 
		
		\begin{equation}
			\mbox{\#}_{i }^{p }(\hat{\mu} )=\begin{cases}
				\kappa (\kappa -1)^{p -1}-\mathbf{1}_{>\hat{p}_{i }}(p )\cdot (\kappa -1)^{p -\hat{p}_{i }-1},&p \le m \\
				(\kappa -1)^{2m -\hat{p}_{i }-p },&p \in (m ,2m -\hat{p}_{i }].
			\end{cases}
			\label{eq:count_shortest_paths_exact_local_tree_deletion}
		\end{equation}

		By combining the count of shortest paths rerouted with their new length we can generalize the loss for any agent from the deletion of link $\iota \iota' $ when all agents are homogeneous of type $x $:	
		\begin{eqnarray}
			u_i (\mu )-u_i (\hat{\mu} )
			=
			\sum_{l =1}^{m -\hat{p}_i }\left[(\kappa -1)^{l -1}\cdot \left(\delta ^{l -1+\hat{p}_{i }}-\delta ^{2m -(l -1)-\hat{p}_{i }}\right)\right]\cdot z (x ,x ).
			\label{eq:utility_loss_exact_local_pairwise}
		\end{eqnarray}
		
		We can aggregate the losses across homogeneous agents of type $x $ and we arrive at the following expression:
		
		\begin{eqnarray}
			U (\mu )-U (\hat{\mu} )=
			\sum_{l =1}^{m }\left[2l \cdot (\kappa -1)^{l -1}\cdot \left(\delta ^{l -1}-\delta ^{2m -(l -1)}\right)\right]\cdot z (x ,x ).
			\label{eq:utility_loss_exact_local_agg}
		\end{eqnarray}
		
		\paragraph{Exact local trees - gains from linking across types}
		
		We move on to establishing the gains of establishing a link in a perfectly sorted network where each component is an exact local tree. 
		
		The gains to agents $\iota $ and $j $ of forming a link $\iota j $ are direct benefits and the new indirect connections that are accessed through the link $\iota j $. For agent $\iota $ the benefits from forming a link with $j $ can be computed with Equation \ref{eq:count_shortest_paths_exact_local_tree_deletion} where the input length is added one (as $\iota j $ is added to the shortest path). Recall $\breve{\mu} =\mu \cup\{\iota j \}\backslash\{\iota \iota' ,j j' \}$.
		
		\begin{equation}
			u_\iota (\breve{\mu} )-u_\iota (\mu )=
			\left[\sum_{l =0}^{m }(\kappa -1)^{l }\cdot \delta ^{l }+\sum_{l =0}^{m -1}(\kappa -1)^{l }\cdot \delta ^{2m -l }\right]\cdot z (x ,\tilde{x} ).\label{eq:benefits_pairwise_exact_local_tree}
		\end{equation}
		
		The above expression is relevant for evaluating the pairwise gains as it captures individual benefits for a pairwise formation of a link by $\iota $ and $j $. However, we are also interested in the sub-connected network as it allows to assess the efficiency. Suppose instead now that $\iota' $ and $j' $ also form a link; thus $\iota j ,\iota' j' $ are formed while $\iota \iota' ,j j' $ are deleted. Let $\tilde{\mu} =\mu \cup\{\iota j ,\iota' j' \}\backslash\{\iota \iota' ,j j' \}$.
		
		Let $i $ be an agent of type $x $ and let $\hat{p}_i $ still denote the least distance to either $\iota $ or $\iota' $. We can calculate the benefits for $i $ when $\iota j ,\iota' j' $ are formed. The benefits are the indirect connections to agents of type $\tilde{x} $ with whom agent $i $ has no connections in $\mu $. The aim is to count the number of paths of a given length.
		
		For a given agent $i' $ of the other type $\tilde{x} $ it must hold that the shortest path in $\tilde{\mu} $ between $i ,i' $ either contains the link $\iota j $ or the link $\iota' j' $, and thus the distance can be computed as follows:
		
		\begin{eqnarray}
		p_{i i' }(\tilde{\mu} )&=&\min [p_{i j }(\tilde{\mu} )+p_{i' j }(\tilde{\mu} ),\,\,p_{i j' }(\tilde{\mu} )+p_{i' j' }(\tilde{\mu} )]\label{eq:count_shortest_paths_local_tree_form_across_simple}
		\end{eqnarray}
		
		We further restrict the above expression. We can use that $i $ and $i' $ of type $\tilde{x} $ can be at most $2m +1$ away from each other. This follows from the fact that $p_{i \iota }(\tilde{\mu} )+p_{i \iota' }(\tilde{\mu} )=2m $ and $p_{i' j }(\tilde{\mu} )+p_{i' j' }(\tilde{\mu} )=2m $. 
		As $p_{i \iota }(\tilde{\mu} )+p_{i \iota' }(\tilde{\mu} )=2m $ and $\iota j ,\iota' j' \in \tilde{\mu} $ then it must be that $p_{i j }+p_{i j' }=2m +2$. These facts together entail we can rewrite Equation \ref{eq:count_shortest_paths_local_tree_form_across_simple}:
		
		\begin{eqnarray}
			p_{i i' }(\tilde{\mu} )&=&\min [p_{i j }(\tilde{\mu} )+p_{i' j }(\tilde{\mu} ),\,\,p_{i j' }(\tilde{\mu} )+p_{i' j' }(\tilde{\mu} )]
			\nonumber \\
			&=&\min [p_{i j }(\tilde{\mu} )+p_{i' j }(\tilde{\mu} ),\,\,4m +2-p_{i j }(\tilde{\mu} )-p_{i' j }(\tilde{\mu} )].\label{eq:count_shortest_paths_local_tree_form_across}
		\end{eqnarray}
		
		From the above expression it follows that $p_{i i' }\le 2m +1$ as the expression is maximized for $p_{i j }+p_{i' j }=2m +1$.

		The number of shortest paths from $i $ through $\iota j $ to agents of the other type $\tilde{x} $  can be found using Equation \ref{eq:count_shortest_paths_exact_local_tree_deletion} for agent $\iota $ adding  extra distance $1+\hat{p}_i $:\footnote{Shortest paths from $i $ must contain both $\iota j $ and every link in the shortest path from $i $ to $j $.} 
		\begin{itemize}
		\item for distance $p \in \{1+\hat{p}_i ,...,m +1+\hat{p}_i \}$ there are $(\kappa -1)^{p -1-\hat{p}_i }$ agents; 
		\item for distance $p \in \{m +2+\hat{p}_i ,...,\,2m +1\}$ there are $(\kappa -1)^{2m +1-(p -1-\hat{p}_i )}$.
		\end{itemize}
		The shortest paths from $i $ not routed through $\iota $ but instead through $\iota' $ are those where $p +1+\hat{p}_i >2m +1$; from Equation \ref{eq:count_shortest_paths_local_tree_form_across} we know the new shortest path length is $4m +2-p -1-\hat{p}_i $. The number of shortest paths through $\iota' $ in network $\tilde{\mu} $ will be $(\kappa -1)^{2m +1-(p -1-\hat{p}_i )}$ and the new length $4m +2-p -1-\hat{p}_i $. These facts together imply:
		
		\footnotesize\begin{equation}
			\mbox{\#}_{i }^{p }(\tilde{\mu} )-
			\mbox{\#}_{i }^{p }(\hat{\mu} )=\begin{cases}
				(\kappa -1)^{p -1-\hat{p}_i },
				&
				p \in \{\hat{p}_i +1,..,m +1+\hat{p}_i \},
				\\
				(\kappa -1)^{2m +1-p -\hat{p}_i },	
				&p \in \{m +\hat{p}_i +2,..,2m +1\},
				\\
				(\kappa -1)^{p +\hat{p}_i -2m -1},	
				&p \in \{2m +1-\hat{p}_i ,..,2m \}.
			\end{cases}
			\label{eq:count_shortest_paths_exact_local_tree_added_across}
		\end{equation}\normalsize
		
		From the number of paths above we can derive the change in utility from when $\iota j ,\iota' j' $ are added to the network for a given agent $i $ of type $x $.
		\footnotesize
		\begin{equation}
			u_i (\tilde{\mu} )-u_i (\hat{\mu} )=
			\left[\begin{array}{rl}
				& 	\sum_{l =0}^{m }(\kappa -1)^{l }\cdot \delta ^{l +\hat{p}_i }\\
				+ & 	\sum_{l =\hat{p}_i }^{m -1}(\kappa -1)^{l }\cdot \delta ^{2m -l +\hat{p}_i }\\
				+ & \sum_{l =0}^{\hat{p}_i -1}(\kappa -1)^{l }\cdot \delta ^{2m +l -\hat{p}_i }
			\end{array}\right]\cdot  z (x ,\tilde{x} ).
			\label{eq:utility_gains_exact_local_pairwise}
		\end{equation}
		\normalsize
		
		By aggregating over all agents of type the gain in benefits by forming $\iota j ,\iota' j' $ is as follows:
		
		\footnotesize
		\begin{equation}
			U (\tilde{\mu} )-U (\hat{\mu} )=
			\sum_{p =0}^{m }
			\left(
			\left[\begin{array}{l}
				\mathbf{1}_{<m }(p )\cdot 2\cdot (\kappa -1)^p +\\
				\mathbf{1}_{=m }(p )\cdot (n -2\cdot \sum_{l =1}^{m -1}(\kappa -1)^l )
			\end{array}\right]
			\cdot 	
			\left[\begin{array}{rl}
				& 	\sum_{l =0}^{m }(\kappa -1)^{l }\cdot \delta ^{l +p }\\
				+ & 	\sum_{l =p }^{m -1}(\kappa -1)^{l }\cdot \delta ^{2m -l +p }\\
				+ & \sum_{l =0}^{p -1}(\kappa -1)^{l }\cdot \delta ^{2m +l -p }
			\end{array}\right]\right)\cdot  Z (x ,\tilde{x} ).
			\label{eq:utility_gains_exact_local_agg}
		\end{equation}		
		\normalsize

		\subsubsection*{Local trees}
		We can use the analysis above on exact local trees to bound the gains and losses for (non-exact) local trees. Recall that exact local trees has the property that $\Delta\mbox{\#} (n ,\kappa )=0$ and for non-exact local trees $\Delta\mbox{\#} (n ,\kappa )>0$. Thus the difference between exact and non-exact local trees is that for a given agent the number of connected other agents at exactly distance $m $ is lower for non-exact local trees.

		Using the analysis of exact local trees we can compute the bounds on loss of utility for a given agent in the local when a link is deleted - this is done by reusing Equation \ref{eq:count_shortest_paths_exact_local_tree_deletion} as follows. 
		
		We can discount the number of agents initially at distance $m $ by $\Delta\mbox{\#} (n ,\kappa )$. Moreover, the new distance between agents $i $  and $i' $ after deletion of the link $\iota \iota' $ is at least $\min (p_{i i' },2m -2-\hat{p}_i -\hat{p}_{i' })$ at most $2m -\hat{p}_i -\hat{p}_{i' }$.\footnote{The upper bound follows from the fact that for any two agents $i $ and $i' $ in the local tree there is still always an agent $j $ at distances $p_{i j }=m -\hat{p}_i $ and $p_{i' j }=m -\hat{p}_{i' }$. The lower bound can be established by repeating an argument used for exact local trees. If the new distance between two agents $i $ and $i' $ after deletion of $\iota \iota' $ had been less than $\min (p_{i i' }(\mu ),2m -2-\hat{p}_i -\hat{p}_{i' })$ then the following would be true. There would be multiple shortest paths of length less than or equal to $m -1$ between either ($\iota $ and $j $) or ($\iota' $ and $j $). This would violate the property of local trees that all shortest paths of length $\le m -1$ are unique.} 
		From these two facts we can derive the bound on loss of utility when $\iota \iota' $ is deleted. The upper bound on loss (in terms of magnitude) is when new shortest paths have most distance, i.e. $2m -\hat{p}_i -\hat{p}_{i' }$; the lower bound is found when new distance is least, i.e. $\min (p_{i i' },2m -2-\hat{p}_i -\hat{p}_{i' })$:
		
		\footnotesize\begin{eqnarray}u_i (\mu )-u_i (\hat{\mu} )&\le &
			\sum_{l =1}^{m -\hat{p}_i }\left[\max (0,(\kappa -1)^{l -1}-\mathbf{1}_{=m }(l )\cdot \Delta\mbox{\#} (n ,\kappa ))\left(\delta ^{l -1+\hat{p}_{i }}-\delta ^{2m -(l -1)-\hat{p}_{i }}\right)\right]\cdot z (x ,x ),
			\label{eq:loss_local_trees_upper_bound}
			\\
			u_i (\mu )-u_i (\hat{\mu} )&\ge &
			\sum_{l =1}^{\tilde{m }}\left[(\kappa -1)^{l -1}\cdot \left(\delta ^{l -1+\hat{p}_{i }}-\delta ^{2m -(l +1)-\hat{p}_{i }}\right)\right]\cdot z (x ,x ),\quad \tilde{m }=\min (m -1,m -\hat{p}_i ).
			\label{eq:loss_local_trees_lower_bound}
		\end{eqnarray}\normalsize	
		
		\begin{fact}\label{fact:gross_loss_lower_attenuated}
			If $\mu $ is perfectly sorted and consists of $|X |$ components that each constitute a local tree with $n /|X |$ agents, then for any agent $i $ of type $x $ where $\hat{p}_i >0$:
			\begin{equation}
				u_i (\hat{\mu} )-u_i (\mu )>\delta ^{\hat{p}_i }\cdot [u_\iota (\hat{\mu} )-u_\iota (\mu )],\quad\hat{p}_{i }=\min (p_{i \iota }(\hat{\mu} ),p_{i \iota' }(\hat{\mu} )).
				\label{eq:gross_loss_lower_attenuated}	
			\end{equation}
		\end{fact}
		\begin{proof}	
			Inequality \ref{eq:gross_loss_lower_attenuated} can be rewritten into: $\delta ^{\hat{p}_i }\cdot [u_\iota (\mu )-u_\iota (\hat{\mu} )]-[u_i (\mu )-u_i (\hat{\mu} )]>0$. This inequality is equivalent to the expression below (derived by substituting in Inequality \ref{eq:loss_local_trees_lower_bound} for agent $\iota $ and Inequality \ref{eq:loss_local_trees_upper_bound} for agent $i $):

			\footnotesize
			\begin{eqnarray*}
				\delta ^{\hat{p}_i }\cdot \sum_{l =1}^{m -1}\left[(\kappa -1)^{l -1}\cdot \left(\delta ^{l -1}-\delta ^{2m -(l +1)}\right)\right]-
				\sum_{l =1}^{m -\hat{p}_i }\left[(\kappa -1)^{l -1}\left(\delta ^{l -1+\hat{p}_{i }}-\delta ^{2m -(l -1)-\hat{p}_{i }}\right)\right]&>& 0,\\
				\sum_{l =1}^{m -\hat{p}_i }\left[(\kappa -1)^{l -1}\cdot \left(\delta ^{
					2m -(l +1)-\hat{p}_i }-\delta ^{2m -(l +1)+\hat{p}_i }\right)\right]+
				\sum_{l =m -\hat{p}_i +1}^{m -1}\left[(\kappa -1)^{l -1}\left(\delta ^{l -1+\hat{p}_{i }}-\delta ^{2m -(l -1)-\hat{p}_{i }}\right)\right]&>& 0.				
			\end{eqnarray*} 
			\normalsize	
			
			As it holds that $2m -(l +1)-\hat{p}_i  < 2m -(l +1)+\hat{p}_i $ and it holds that 			
			$l -1+\hat{p}_{i }<2m -(l -1)-\hat{p}_{i }$ (equivalent to $l <m +1-\hat{p}_{i }$) the above inequality is satisfied.
		\end{proof}

		We can also derive bounds on the gains from connecting across types for local trees. We will not do this explicitly but instead use Definition \ref{def:local_tree_symmetric_losses} on symmetric losses in local trees. This allows to express our next result:
		
		\begin{fact}\label{fact:gross_gain_higher_attenuated}
			For the perfectly sorted network $\mu $ which consists of $|X |$ network components which each constitute a local tree of $n /|X |$ agents that has symmetric losses then it holds that for agents $i , \iota $ of type $x $ and $\hat{p}_i >0$
			\begin{equation}
				u_i (\tilde{\mu} )-u_i (\hat{\mu} )\ge \delta ^{\hat{p}_i }\cdot [u_\iota (\breve{\mu} )-u_\iota (\hat{\mu} )],\quad \hat{p}_{i }=\min (p_{i \iota }(\hat{\mu} ),p_{i \iota' }(\hat{\mu} )).
				\label{eq:gross_gain_higher_attenuated}
			\end{equation}	
		\end{fact}	
		\begin{proof}
			It holds that $u_\iota (\tilde{\mu} )-u_\iota (\hat{\mu} )\ge u_\iota (\breve{\mu} )-u_\iota (\hat{\mu} )$ as $\tilde{\mu} \subseteq \breve{\mu} $ (thus all shortest paths in $\tilde{\mu} $ cannot have a length that exceeds that in $\breve{\mu} $). Therefore it suffices to show: \begin{equation}
				u_i (\tilde{\mu} )-u_i (\hat{\mu} )\ge \delta ^{\hat{p}_i }\cdot [u_\iota (\tilde{\mu} )-u_\iota (\hat{\mu} )].\label{eq:sufficient_fact_gains_local_tree}
			\end{equation} 
			
			As the local tree has symmetric losses it follows that $u_\iota (\tilde{\mu} )-u_\iota (\hat{\mu} )=u_{\iota' }(\tilde{\mu} )-u_{\iota' }(\hat{\mu} )$; this follows from the fact that they both gain an equal number of new shortest paths through $j ,j' $, this follows as as $j ,j' $ have same number of paths after deletion of $j j' $ due to symmetric losses. 
			This entails that without loss of generality we can assume that $p_{i \iota }=\hat{p}_i $ as otherwise we could substitute $\iota $ with $\iota' $ and conduct the analysis again.
			
			For $\iota $ and some agent $i' $ of type $\tilde{x} $ it holds that  $p_{i i' }(\tilde{\mu} )\le p_{\iota i' }(\tilde{\mu} )+\hat{p}_i $. This follows as there exists a path between $i ,\iota $ and $\iota ,i' $ with respectively lengths $p_{\iota i' }(\tilde{\mu} )$ and $\hat{p}_i $; thus $p_{i i' }(\tilde{\mu} )\le p_{\iota i' }(\tilde{\mu} )+\hat{p}_i $. This implies the following inequality must hold:

			\[\sum_{x_{i' }=\tilde{x} }\delta ^{p_{i i' }(\tilde{\mu} )}
			\ge 
			\delta ^{p_{\iota i }(\tilde{\mu} )}\cdot \sum_{x_{i' }=\tilde{x} }\delta ^{p_{\iota i' }(\tilde{\mu} )}.\]

			As $u_\iota (\tilde{\mu} )-u_\iota (\hat{\mu} )=\sum_{x_{i }=\tilde{x} }\prod_{l =1}^{p_{\iota i' }(\tilde{\mu} )}\delta ^{r_l }\cdot z (x ,\tilde{x} )$ and $u_i (\tilde{\mu} )-u_i (\hat{\mu} )=\sum_{x_{i }=\tilde{x} }\prod_{l =1}^{p_{i i' }(\tilde{\mu} )}\delta ^{r_l }\cdot z (x ,\tilde{x} )$ it follows that Inequality \ref{eq:sufficient_fact_gains_local_tree} holds which proves our fact.
			
		\end{proof}

\end{document}